\documentclass[journal,onecolumn,draftclsnofoot,12pt]{IEEEtran}%
\usepackage{amsfonts}
\usepackage{amsmath}
\usepackage{amssymb}
\usepackage[dvips]{graphicx}
\usepackage{setspace}
\usepackage{enumerate}
\usepackage{cite}%
\newtheorem{theorem}{Theorem}

\newtheorem{example}{Example}

\newtheorem{lemma}{Lemma}

\newtheorem{proposition}{Proposition}
\newtheorem{remark}{Remark}

\begin{document}

\title{On the Sum-Capacity of Degraded Gaussian Multiaccess Relay Channels}
\author{Lalitha Sankar~\IEEEmembership{Member,~IEEE,}, Narayan B.
Mandayam,~\IEEEmembership{Senior~Member,~IEEE}, and~H. Vincent Poor
\IEEEmembership{Fellow,~IEEE}\thanks{The work of L.~Sankar and N.~B.~Mandayam
was supported in part by the National\ Science Foundation under
Grant~No.~{\scriptsize ITR-0205362}. The work of H. V. Poor was supported by
the National Science Foundation under Grants ANI-03-38807 and CNS-06-25632.
The material in this paper was presented in part at the Information Theory and
Applications Workshop, San Diego, CA, January 2006. L. Sankar and H.\ V. Poor
are with Princeton University. N.~B.~Mandayam is with the WINLAB, Rutgers
University. }}
\pubid{\ }
\specialpapernotice{\ }
\maketitle

\begin{abstract}
The sum-capacity is studied for a $K$-user degraded Gaussian multiaccess relay
channel (MARC) where the multiaccess signal received at the destination from
the $K$ sources and relay is a degraded version of the signal received at the
relay from all sources, given the transmit signal at the relay. An outer bound
on the capacity region is developed using cutset bounds. An achievable rate
region is obtained for the decode-and-forward (DF) strategy. It is shown that
for every choice of input distribution, the rate regions for the inner (DF)
and outer bounds are given by the intersection of two $K$-dimensional
polymatroids, one resulting from the multiaccess link at the relay and the
other from that at the destination. Although the inner and outer bound rate
regions are not identical in general, for both cases, a classical result on
the intersection of two polymatroids is used to show that the intersection
belongs to either the set of \textit{active cases} or \textit{inactive cases},
where the two bounds on the $K$-user sum-rate are active or inactive,
respectively. It is shown that DF achieves the capacity region for a class of
degraded Gaussian MARCs in which the relay has a high\ SNR\ link to the
destination relative to the multiaccess link from the sources to the relay.
Otherwise, DF is shown to achieve the sum-capacity for an \textit{active
}class of degraded Gaussian MARCs for which the DF\ sum-rate is maximized by a
polymatroid intersection belonging to the set of active cases. This class is
shown to include the class of \textit{symmetric} Gaussian MARCs where all
users transmit at the same power.

\end{abstract}


\IEEEpeerreviewmaketitle

\section{Introduction}

The multiaccess relay channel (MARC) is a network in which several users
(sources) communicate with a single destination in the presence of a relay
\cite{cap_theorems:KvW01}. The coding strategies developed for the relay
channel \cite{cap_theorems:vdM01,cap_theorems:CEG01} extend readily to the
MARC \cite{cap_theorems:KGG_IT,cap_theorems:SKM03}. For example, the strategy
of \cite[theorem~1]{cap_theorems:CEG01}, now often called
\textit{decode-and-forward} (DF), has a relay that decodes user messages
before forwarding them to the destination
\cite{cap_theorems:KGG_IT,cap_theorems:SKM03}. Similarly, the strategy in
\cite[theorem~6]{cap_theorems:CEG01}, now often called
\textit{compress-and-forward} (CF), has the relay quantize its output symbols
and transmit the resulting quantized bits to the destination
\cite{cap_theorems:SKM03}. 

Capacity results for relay channels are known only for a few special cases
such as the class of degraded relay channels \cite{cap_theorems:CEG01} and its
multi-relay generalization \cite{cap_theorems:ReznikKV01,cap_theorems:XK02_IT}%
, the class of semi-deterministic relay channels \cite{cap_theorems:EGA01},
the class of orthogonal relay channels
\cite{cap_theorems:LiangVP_ResAllocJrnl,cap_theorems:ElGZahedi01}, the class
of Gaussian relay without delay channels
\cite{cap_theorems:ElGamal_Hassanpour01,cap_theorems:vdMVanroose}, and the
class of ergodic phase-fading relay channels \cite{cap_theorems:KGG_IT}. For
the class of degraded relay channels, the degradedness condition requires that
the received signal at the destination be independent of the source signal
when conditioned on the transmit and receive signals at the relay. For the
Gaussian case, this simplifies to the requirement that the signal received at
the destination be a noisier version of that received at the relay conditioned
on the transmitted signal at the relay. This condition immediately suggests
that requiring the relay to decode the source signals should be optimal. In
fact, for this class, applying this degradedness condition simplifies the
cut-set outer bounds to coincide with the DF bounds. For the MARC, we
generalize this degradedness condition to requiring that the signal received
at the destination be independent of all source signals conditioned on the
transmit and receive signals at the relay. Applying this degradedness
condition to the cutset outer bounds for a MARC, however, does not simplify
the bounds to those achieved by DF.

A $K$-user Gaussian MARC is degraded when the multiaccess signal received at
the destination from the $K$ sources and relay is a noisier version of the
signal received at the relay from all sources, given the transmit signal at
the relay. For a $K$-user degraded Gaussian MARC, we develop the DF rate
region as an inner bound on the capacity region using Gaussian signaling at
the sources and relay. The outer bounds on the capacity region are obtained by
specializing the cut-set bounds of \cite[Th. 14.10.1]{cap_theorems:CTbook} to
the case of independent sources \cite{cap_theorems:SKM02a} and by applying the
degradedness condition. In fact, for each choice of input distribution, both
the DF and the cutset rate regions are intersections of two multiaccess rate
regions, one with the relay as the receiver and the other with the destination
as the receiver. In general, however, the inner and outer bounds differ in
their input distributions as well as the rate bounds. The outer bounds allow a
more general dependence between the source and relay signals relative to DF
where we use auxiliary random variables, one for each source, to relate the
transmitted signals at the sources and relay. For the Gaussian\ degraded MARC,
we show that using Gaussian input signals at the sources and relay maximizes
the outer bounds. For the inner bounds, we use Gaussian signaling at the
sources and the relay via $K$ Gaussian auxiliary random variables. As a
result, for each choice of the appropriate Gaussian input distribution, both
the DF and outer bounds are then parametrized by $K$ source-relay
cross-correlation coefficients, i.e., a $K$-length correlation vector.
Specifically, each DF coefficient is a product of the two power fractions
allocated for cooperation at the corresponding source and the relay,
respectively. We show that the DF rate region over all feasible correlation
vectors is a convex region. On the other hand, for the outer bounds, all the
rate bounds at the relay except for the bound on the $K$-user sum-rate are
non-concave functions of the correlation coefficients, and thus, the outer
bound rate region requires time-sharing. Finally, we also show that for every
feasible choice of the correlation vector, the multiaccess regions achieved by
the inner and outer bounds at each receiver are polymatroids, and the
resulting region is an intersection of two polymatroids.

We use a well-known result on the intersection of two polymatroids \cite[chap.
46]{cap_theorems:Schrijver01} to broadly classify polymatroid intersections
into two categories, namely, the set of \textit{active }and the set of
\textit{inactive cases}, depending on whether the constraints on the $K$-user
sum-rate at both receivers are active or inactive, respectively. In fact, we
use \cite[chap. 46]{cap_theorems:Schrijver01} to show that the $K$-user
sum-rate for the inactive cases is always bounded by the minimum of the
(inactive) $K$-user sum-rate bounds at each receiver, and thus, by the largest
such bound. For both the inner and outer bounds, the intersection of the two
rate polymatroids results in either an active or a inactive case for every
choice of correlation vectors. In fact, the minimum of the $K$-user sum-rate
bounds at the relay and destination is the effective sum-rate only if the
polymatroid intersection is an active case and is strictly an upper bound for
an inactive case. 

Irrespective of the above mentioned distinction, we first consider the problem
of maximizing the minimum of the $K$-user sum-rate bounds at the relay and
destination over the set of all correlation coefficients. We solve this
max-min optimization problem using techniques analogous to the classical
minimax problem of detection theory \cite[II.C]{cap_theorems:HVPoor01}. We
refer to a sum-rate optimal correlation vector as a \textit{max-min rule}.

For both the inner and outer bounds, we show that the max-min optimization
described above has two unique solutions. The first solution is given by the
maximum $K$-user sum-rate achievable at the relay and results when the
multiaccess link between the sources and the relay is the bottle-neck link.
For this case, we show that the intersection of the rate regions at the relay
and destination belongs to the set of active cases and is in fact the same as
the region achieved at the relay. We further show that this region is the same
for both the inner and outer bounds and is the capacity region for a class of
degraded Gaussian\ MARCs where the source and relay powers satisfy the
bottle-neck condition for this case.

The second solution pertains to the case in which the bottle-neck condition
described above is not satisfied, i.e., the $K$-user sum-rate at the relay is
at least as large as that at the destination. For this case, we show that for
both the inner and outer bounds the max-min optimization solution requires the
$K$-user sum-rate bounds at the relay and destination to be equal. In fact, we
show that both the inner and outer bounds achieve the same maximum sum-rate
for this case. Further, for both sets of bounds, we show that this maximum is
achieved by a set of correlation vectors, i.e., the max-min rule is a set
rather than a singleton. Recall, however, that the sum-rate computed thus is
achievable for either bound only if there exists at least one max-min rule for
which the polymatroid intersection belongs to the set of active cases;
otherwise, the computed maximum is strictly an upper bound on the maximum
sum-rate. Combining this with the fact that the maximum inner and outer
$K$-user sum rate bounds for this case are the same, we establish that DF
achieves the sum-capacity of an \textit{active} class of degraded Gaussian
MARCs, i.e., a class for which the maximum sum-rate is achieved because there
exists at least one max-min rule for which the polymatroid intersection is an
active case. We also show that the class of \textit{symmetric }Gaussian MARCs,
in which all sources have the same power, belongs to this active class.
Finally, for the remaining \textit{inactive} class of degraded Gaussian\ MARCs
in which no active case results for any choice of the max-min rule, we provide
a common upper bound on both the DF and the cutset sum-rates.

This paper is organized as follows. In Section \ref{DG_Sec2} we present a
model for a degraded Gaussian MARC. In Section \ref{DG_Sec3} we develop the
cut-set bounds on the capacity region of a MARC. In Section \ref{DG_Sec4} we
determine the maximum $K$-user DF sum-rate. We discuss our results and
conclude in Section \ref{DG_Sec6}.

\section{\label{DG_Sec2}Channel Model and Preliminaries}

A $K$-user degraded Gaussian MARC has $K$ user (source) nodes, one relay node,
and one destination node (see Fig. \ref{Fig_DGMARC}). The sources emit the
messages $W_{k}$, $k=1,2,\ldots,K$, which are statistically independent and
take on values uniformly in the sets $\{1,2,\ldots,M_{k}\}$. The channel is
used $n$ times so that the rate of $W_{k}$ is $R_{k}=\left.  B_{k}\right/  n$
bits per channel use where $B_{k}=\log_{2}M_{k}$ bits. In each use of the
channel, the input to the channel from source $k$ is $X_{k}$ while the relay's
input is $X_{r}$. The channel outputs $Y_{r}$ and $Y_{d}$, respectively, at
the relay and the destination are
\begin{align}
Y_{r} &  =\left(  \sum\limits_{k=1}^{K}X_{k}\right)  +Z_{r}\label{Yr_defn}\\
Y_{d} &  =\left(  \sum\limits_{k=1}^{K}X_{k}\right)  +X_{r}+Z_{d}%
\label{Yd_defn}\\
&  =Y_{r}+X_{r}+Z_{\Delta}\label{DGMARC_def}%
\end{align}
where $Z_{r}$ and $Z_{\Delta}$ are independent Gaussian random variables with
zero means and variances $N_{r}$ and $N_{\Delta}$, respectively, such that the
noise variance at the destination is
\begin{equation}
N_{d}=N_{r}+N_{\Delta}.\label{DGMARC_Noise_var}%
\end{equation}
%

%TCIMACRO{\FRAME{ftbpFU}{4.1027in}{1.7504in}{0pt}{\Qcb{A two-user
%Gaussian\ degraded MARC.}}{\Qlb{Fig_DGMARC}}{dg_marc.eps}%
%{\special{ language "Scientific Word";  type "GRAPHIC";  display "USEDEF";
%valid_file "F";  width 4.1027in;  height 1.7504in;  depth 0pt;
%original-width 3.9972in;  original-height 1.2107in;  cropleft "0";
%filename 'DG_MARC.eps';file-properties "XNPEU";}}}%
%BeginExpansion
\begin{figure}
[ptb]
\begin{center}
\includegraphics[
height=1.7504in,
width=4.1027in
]%
{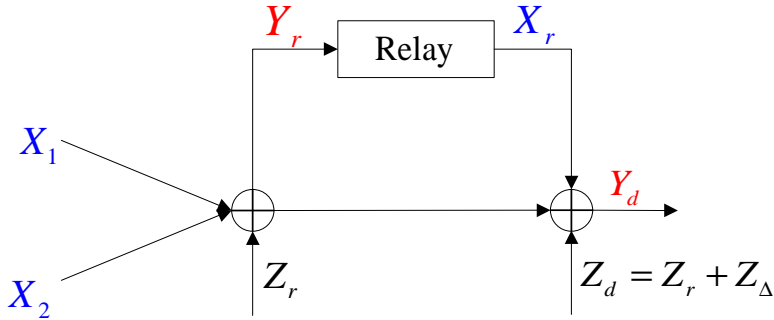}%
\caption{A two-user Gaussian\ degraded MARC.}%
\label{Fig_DGMARC}%
\end{center}
\end{figure}
%EndExpansion
We assume that the relay operates in a full-duplex manner, i.e., it can
transmit and receive simultaneously in the same bandwidth. Further, its input
$X_{r}$ in each channel use is a causal function of its outputs from previous
channel uses. We write $\mathcal{K}=\{1,2,\ldots,K\}$ for the set of sources,
$\mathcal{T}$ $=$ $\mathcal{K}$ $\cup$ $\{r\}$ for the set of transmitters,
$\mathcal{R}$ $=$ $\{r,d\}$ for the set of receivers, $X_{\mathcal{S}}%
=\{X_{k}$ $:$ $k$ $\in$ $\mathcal{S}\}$ for all $\mathcal{S}$ $\subseteq$
$\mathcal{K}$, and $\mathcal{S}^{c}$ to denote the complement of $\mathcal{S}$
in $\mathcal{K}$.

The transmitted signals from source $k$ and the relay have a per symbol power
constraint
\begin{equation}%
\begin{array}
[c]{cc}%
\left.  E\left[  \left\vert X_{k}\right\vert ^{2}\right]  \leq P_{k}\right.  &
\text{ }k\in\mathcal{T}.
\end{array}
\label{Pwr_cond}%
\end{equation}
One can equivalently express the relationship between the input and output
signals in (\ref{DGMARC_def}) as a Markov chain
\begin{equation}
\left(  X_{1},X_{2},\ldots,X_{K}\right)  -(Y_{r},X_{r})-Y_{d}.
\label{DG_Markov}%
\end{equation}
For $K$ $=$ $1$, (\ref{DG_Markov}) simplifies to the degradedness condition in
\cite[(10)]{cap_theorems:CEG01} for the classic (single source) relay channel.
A degraded Gaussian MARC is \textit{symmetric} if $P_{k}=P$, for all $k$.
Thus, a class of symmetric DG-MARCs is characterized by four parameters,
namely, $P,$ $P_{r}\,,$ $N_{r},$ and $N_{d}$.

The capacity region $\mathcal{C}_{\text{MARC}}$ is the closure of the set of
rate tuples $(R_{1},R_{2},\ldots,R_{K})$ for which the destination can, for
sufficiently large $n$, decode the $K$ source messages with an arbitrarily
small positive error probability. As further notation, we write
$R_{\mathcal{S}}=%
%TCIMACRO{\tsum \nolimits_{k\in\mathcal{S}}}%
%BeginExpansion
{\textstyle\sum\nolimits_{k\in\mathcal{S}}}
%EndExpansion
R_{k}$ and $Y_{\mathcal{R}}=\left(  Y_{r},Y_{d}\right)  $. We write
$\underline{0}$ and $\underline{1}$ to denote vectors whose entries are all
zero and one, respectively, and $C(x)=\log(1+x)/2$ to denote the capacity of
an AWGN channel with signal-to-noise ratio (SNR) $x$. We use the usual
notation for entropy and mutual information
\cite{cap_theorems:Gallager01,cap_theorems:CTbook} and take all logarithms to
the base 2 so that in each channel use our rate units are bits.

\section{\label{DG_Sec3}Outer Bounds}

An outer bound on the capacity region of a MARC is presented in
\cite{cap_theorems:SKM02a} using the cut-set bounds in \cite[Th.
14.10.1]{cap_theorems:CTbook} as applied to the case of independent sources.
We summarize the bounds below.

\begin{proposition}
\label{Prop_OB}The capacity region $\mathcal{C}_{\text{MARC}}$ is contained in
the union of the set of rate tuples $(R_{1},R_{2},\ldots,R_{K})$ that satisfy,
for all $\mathcal{S}\subseteq\mathcal{K}$,
\begin{equation}
R_{\mathcal{S}}\leq\min\left\{  I(X_{\mathcal{S}};Y_{r},Y_{d}|X_{\mathcal{S}%
^{c}},X_{r},U),I(X_{\mathcal{S}},X_{r};Y_{d}|X_{\mathcal{S}^{c}},U)\right\}
\label{MARC_OB_cutset}%
\end{equation}
where the union is over all distributions that factor as%
\begin{equation}
p(u)\cdot\left(  \prod\nolimits_{k=1}^{K}p(x_{k}|u)\allowbreak\right)  \cdot
p(x_{r}|\allowbreak x_{\mathcal{K}}\allowbreak,u)\cdot p(y_{r},y_{d}%
|x_{\mathcal{K}},x_{r}). \label{GMARC_converse_inpdist}%
\end{equation}

\end{proposition}

\begin{remark}
The \textit{time-sharing} random variable $U$ ensures that the region in
(\ref{MARC_OB_cutset}) is convex. One can apply Caratheodory's theorem
\cite{cap_theorems:Eggbook01} to this $K$-dimensional convex region to bound
the cardinality of $U$ as $\left\vert \mathcal{U}\right\vert \leq K+1$.
\end{remark}

Consider the outer bounds in Proposition \ref{Prop_OB}. For a degraded
Gaussian\ MARC applying the degradness definition in (\ref{DG_Markov})
simplifies (\ref{MARC_OB_cutset}) as
\begin{equation}%
\begin{array}
[c]{cc}%
R_{\mathcal{S}}\leq\min\left\{  I(X_{\mathcal{S}};Y_{r}|X_{r}X_{\mathcal{S}%
^{c}}U),I(X_{\mathcal{S}}X_{r};Y_{d}|X_{\mathcal{S}^{c}}U)\right\}  &
\text{for all }\mathcal{S}\subseteq\mathcal{K}%
\end{array}
\label{DGMARC_OB_1}%
\end{equation}
for the same joint distribution in (\ref{GMARC_converse_inpdist}). In the
following theorem, we develop the bounds in (\ref{DGMARC_OB_1}) with $U$ as a
constant. For notational convenience, for a constant $U$, we write
$B_{r,\mathcal{S}}$ and $B_{d,\mathcal{S}}$ to denote the first and second
terms, respectively, of the minimum on the right-side of (\ref{DGMARC_OB_1}).
The proof of the following theorem is detailed in Appendix
\ref{DG_App0_OBProof}. \newline

\begin{theorem}
\label{DGOB_Th0}For a degraded Gaussian MARC, the bounds $B_{r,\mathcal{S}}$
and $B_{d,\mathcal{S}}$ are given by
\begin{equation}
B_{r,\mathcal{S}}=\left\{
\begin{array}
[c]{ll}%
C\left(
%TCIMACRO{\tsum \limits_{k\in\mathcal{S}}}%
%BeginExpansion
{\textstyle\sum\limits_{k\in\mathcal{S}}}
%EndExpansion
\frac{P_{k}}{N_{r}}\right)  &
%TCIMACRO{\tsum \limits_{k\in\mathcal{S}^{c}}}%
%BeginExpansion
{\textstyle\sum\limits_{k\in\mathcal{S}^{c}}}
%EndExpansion
\gamma_{k}=1\\
C\left(
%TCIMACRO{\tsum \limits_{k\in\mathcal{S}}}%
%BeginExpansion
{\textstyle\sum\limits_{k\in\mathcal{S}}}
%EndExpansion
\frac{P_{k}}{N_{r}}-\frac{\left(  \sum\limits_{k\in\mathcal{S}}\sqrt
{\gamma_{k}P_{k}}\right)  ^{2}}{N_{r}\overline{\gamma}_{\mathcal{S}^{c}}%
}\right)  & \text{otherwise}%
\end{array}
\right.  \label{Con_final_B1S}%
\end{equation}
and%
\begin{equation}
B_{d,\mathcal{S}}=C\left(  \frac{%
%TCIMACRO{\tsum \limits_{k\in\mathcal{S}}}%
%BeginExpansion
{\textstyle\sum\limits_{k\in\mathcal{S}}}
%EndExpansion
P_{k}+\overline{\gamma}_{\mathcal{S}^{c}}P_{r}+2%
%TCIMACRO{\tsum \limits_{k\in\mathcal{S}}}%
%BeginExpansion
{\textstyle\sum\limits_{k\in\mathcal{S}}}
%EndExpansion
\sqrt{\gamma_{k}P_{k}P_{r}}}{N_{d}}\right)  \label{Con_final_B2S}%
\end{equation}
where $\overline{\gamma}_{\mathcal{S}^{c}}=1-\sum_{k\in\mathcal{S}^{c}}%
\gamma_{k}$ and%
\begin{equation}%
\begin{array}
[c]{cc}%
\sqrt{\gamma_{k}P_{k}P_{r}}\overset{\vartriangle}{=}E(X_{k}X_{r}) & \text{for
all }k\in\mathcal{K}.
\end{array}
\end{equation}

\end{theorem}

\begin{remark}
For $K$ $=$ $1$, the bounds in (\ref{Con_final_B1S}) and (\ref{Con_final_B2S})
simplify to the first and second bound, respectively, for the degraded relay
channel in \cite[theorem 5]{cap_theorems:CEG01}.
\end{remark}

\begin{remark}
The source-relay cross-correlation variables $\gamma_{k}$, for all $k$,
satisfy (\ref{OB_gammaK_sum}), i.e., they lie in the closed convex region
$\Gamma_{OB}$ given by%
\begin{equation}
\Gamma_{OB}=\left\{  \underline{\gamma}_{\mathcal{K}}:\sum_{k\in\mathcal{K}%
}\gamma_{k}\leq1\right\}  .\label{DGOB_TOB}%
\end{equation}
The bound $B_{r,\mathcal{S}}$ in (\ref{Con_final_B1S}), in general, is not a
concave function of $\underline{\gamma}_{\mathcal{K}}$ for any $\mathcal{S}%
\subset\mathcal{K}$. For a fixed $\underline{\gamma}_{\mathcal{S}^{c}}$, in
Appendix \ref{DG_Appen_2} we show that $B_{r,\mathcal{S}}$ is a concave
function of $\underline{\gamma}_{\mathcal{S}}$. This in turn implies that
$B_{r,\mathcal{K}}$ is a concave function of $\underline{\gamma}_{\mathcal{K}%
}$. Further, in\ Appendix \ref{DG_AppConvex} we show that for all
$\mathcal{S}$, $B_{d,\mathcal{S}}$ in (\ref{Con_final_B2S}) is a concave
function of $\underline{\gamma}_{\mathcal{K}}$.
\end{remark}

\begin{remark}
In the expression for $B_{d,\mathcal{S}}$ in (\ref{Con_final_B2S}), the terms
involving the cross-correlation coefficients quantify the coherent combining
gains that result from choosing correlated source and relay signals. On the
other hand, the expression for $B_{r,\mathcal{S}}$ in (\ref{Con_final_B1S})
quantifies the upper bounds on the rate achievable at the relay when one or
more source signals are correlated with the transmitted signal at the relay.
\end{remark}

The rate region $\mathcal{R}_{OB}$ enclosed by the cut-set outer bounds is
obtained as follows. From (\ref{Con_finalbounds}) for any choice of
$\underline{\gamma}_{\mathcal{K}}$, the rate region is an intersection of the
regions enclosed by the bounds $B_{r,\mathcal{S}}$ and $B_{d,\mathcal{S}}$ for
all $\mathcal{S}$. Since $B_{r,\mathcal{S}}$ is not a concave function of
$\underline{\gamma}_{\mathcal{K}}$, one must also consider all possible convex
combinations of $\underline{\gamma}_{\mathcal{K}}$ to obtain $\mathcal{R}%
_{OB}$. For the $K$-dimensional convex region $\mathcal{R}_{OB}$, one can
apply Caratheodory's theorem \cite{cap_theorems:Eggbook01} to express every
rate tuple $(R_{1},R_{2},\ldots,R_{K})$ in $\mathcal{R}_{OB}$ as a convex
combination of at most $K+1$ rate tuples, where each rate tuple is obtained
for a specific choice of $\underline{\gamma}_{\mathcal{K}}$. Let $\Theta$
denote the collection of all vectors \underline{$\eta$} that satisfy
\begin{equation}%
%TCIMACRO{\tsum \nolimits_{m=1}^{K+1}}%
%BeginExpansion
{\textstyle\sum\nolimits_{m=1}^{K+1}}
%EndExpansion
\eta_{m}=1\label{Wt_vect_sum}%
\end{equation}
and let $\underline{\zeta}\equiv(\{\underline{\gamma}_{\mathcal{K}}\}_{K+1}%
,$\underline{$\eta$}$)\in\Gamma_{OB}^{K+1}\times\Theta$ denote a collection of
$K+1$ power fractions and weights such that the rate tuple achieved by the
$m^{th}$ vector $\underline{\gamma}_{\mathcal{K}}^{(m)}$ is weighted by the
$m^{th}$ non-negative entry of the weight vector \underline{$\eta$}, for all
$m\in\mathcal{K\cup}\left\{  K+1\right\}  $. Finally, since $\Gamma_{OB}$ in
(\ref{DGOB_TOB}) is a closed convex set, $%
%TCIMACRO{\tsum \nolimits_{m=1}^{K+1}}%
%BeginExpansion
{\textstyle\sum\nolimits_{m=1}^{K+1}}
%EndExpansion
\eta_{m}\underline{\gamma}_{\mathcal{K}}^{(m)}\in\Gamma_{OB}$. The following
theorem presents an outer bound on the capacity region of the degraded
Gaussian MARC.

\begin{theorem}
The capacity region $\mathcal{C}_{\text{MARC }}$of a degraded Gaussian MARC is
contained in the region $\mathcal{R}_{OB}$ given as%
\begin{equation}
\mathcal{R}_{OB}=\bigcup\limits_{\underline{\zeta}\in\Gamma_{OB}}\left(
\mathcal{R}_{r}^{ob}\left(  \underline{\zeta}\right)  \cap\mathcal{R}_{d}%
^{ob}\left(  \underline{\zeta}\right)  \right)  \label{GMARC_R_OB}%
\end{equation}
where the rate region $\mathcal{R}_{j}^{ob}\left(  \underline{\zeta}\right)
$, $j=r,d$, is
\begin{equation}
\mathcal{R}_{j}^{ob}\left(  \underline{\zeta}\right)  =\left\{  \left(
R_{1},R_{2},\ldots,R_{K}\right)  :0\leq R_{\mathcal{S}}\leq\overline
{B}_{j,\mathcal{S}}\left(  \underline{\zeta}\right)  \right\}
\end{equation}
and the bound $\overline{B}_{j,\mathcal{S}}$ is given by%
\begin{equation}
\overline{B}_{j,\mathcal{S}}\left(  \underline{\zeta}\right)  =\sum
\limits_{m=1}^{K+1}\eta_{m}B_{j,\mathcal{S}}\left(  \underline{\gamma
}_{\mathcal{K}}^{(m)}\right)  .\label{Con_convex_sum}%
\end{equation}

\end{theorem}

\begin{theorem}
The regions $\mathcal{R}_{r}^{ob}\left(  \underline{\zeta}\right)  $ and
$\mathcal{R}_{d}^{ob}\left(  \underline{\zeta}\right)  $ are polymatroids.
\end{theorem}

\begin{proof}
In Appendix \ref{DG_App_PM} we show that for each choice of input distribution
satisfying (\ref{GMARC_converse_inpdist}), the bounds in
(\ref{Prop_DF_rateregion}) are submodular set functions, i.e., they enclose
regions that are polymatroids. For the optimal Gaussian input distribution,
this implies that $\mathcal{R}_{r}^{ob}\left(  \underline{\zeta}\right)  $ and
$\mathcal{R}_{d}^{ob}\left(  \underline{\zeta}\right)  $ are polymatroids for
every choice of $\underline{\zeta}$.
\end{proof}%

%TCIMACRO{\TeXButton{B}{\begin{figure*}[tbp] \centering}}%
%BeginExpansion
\begin{figure*}[tbp] \centering
%EndExpansion%
%TCIMACRO{\FRAME{itbpF}{5.4518in}{3.8839in}{0in}{}{}{all_cases_thesis.eps}%
%{\special{ language "Scientific Word";  type "GRAPHIC";  display "USEDEF";
%valid_file "F";  width 5.4518in;  height 3.8839in;  depth 0in;
%original-width 0pt;  original-height 0pt;  cropleft "0";  croptop "1";
%cropright "1";  cropbottom "0";
%filename 'All_Cases_thesis.eps';file-properties "XNPEU";}}}%
%BeginExpansion
{\includegraphics[
height=3.8839in,
width=5.4518in
]%
{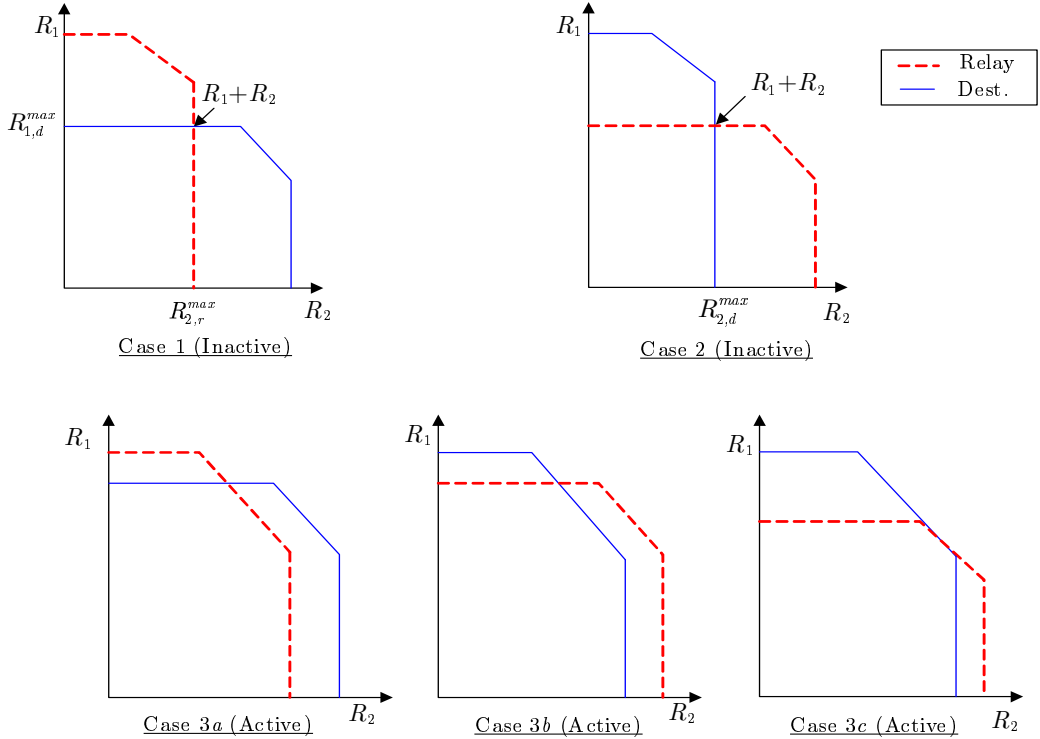}%
}%
%EndExpansion
\caption{Five possible intersections of $\mathcal{R}_{r}$ and $\mathcal{R}_{d}$ for a two-user Gaussian MARC.}\label{Fig_AllCases}%
%TCIMACRO{\TeXButton{E}{\end{figure*}}}%
%BeginExpansion
\end{figure*}%
%EndExpansion

The region $\mathcal{R}_{OB}$ in (\ref{GMARC_R_DF}) is a union of the
intersections of the regions $\mathcal{R}_{r}^{ob}$ and $\mathcal{R}_{d}^{ob}%
$, where the union is taken over all convex combinations of $\underline
{\gamma}_{\mathcal{K}}$. Since $\mathcal{R}_{OB}$ is convex, we obtain the
boundary of $\mathcal{R}_{OB}$ by maximizing the weighted sum $%
%TCIMACRO{\tsum \nolimits_{k\in\mathcal{K}}}%
%BeginExpansion
{\textstyle\sum\nolimits_{k\in\mathcal{K}}}
%EndExpansion
\mu_{k}R_{k}$ over all $\Gamma_{OB}$ and for all $\mu_{k}>0$. Specifically, we
determine the sum-rate $R_{\mathcal{K}}$ when $\mu_{k}$ $=$ $1$ for all $k$.
In general, to determine the intersecting polytope, one has to consider all
possible polytope shapes for the regions $\mathcal{R}_{r}^{ob}$ and
$\mathcal{R}_{d}^{ob}$. However, since $\mathcal{R}_{r}^{ob}$ and
$\mathcal{R}_{d}^{ob}$ are polymatroids, we use the following lemma on
polymatroid intersections \cite[p. 796, Cor. 46.1c]{cap_theorems:Schrijver01}
to broadly classify the intersection of two polymatroids into two categories.
The first \textit{inactive set }category\textit{ }includes all intersections
for which the constraints on the two $K$-user sum-rates are not active. This
implies that no rate tuple on the sum-rate plane achieved at one of the
receivers lies within or on the boundary of the rate region achieved at the
other receiver. On the other hand, the intersections for which there exists at
least one such rate tuple, i.e., the constraints on the two $K$-user sum-rates
are active in the final intersection, belong to the category of \textit{active
set}. In Fig. \ref{Fig_AllCases}, for a two-user MARC we illustrate the five
possible choices for the sum-rate resulting from an intersection of
$\mathcal{R}_{r}^{ob}(\underline{\gamma}_{\mathcal{K}})$ and $\mathcal{R}%
_{d}^{ob}(\underline{\gamma}_{\mathcal{K}})$. Cases $1$ and $2$ belong to the
inactive set while cases $3a,$ $3b$, and $3c$ belong to the active set. We
henceforth refer to members of the active and the inactive sets as active and
inactive cases, respectively. Note that Fig. \ref{Fig_AllCases} illustrates
two specific $\mathcal{R}_{r}^{ob}$ and $\mathcal{R}_{d}^{ob}$ polymatroids
for cases $3a$, $3b$, and $3c$. In general the active set includes all
intersections that satisfy the definition for this set including cases such as
$\mathcal{R}_{r}^{ob}\subseteq\mathcal{R}_{d}^{ob}$ and vice-versa. Finally,
note that the sum-rate is a minimum of the sum-rates at the two receivers only
for the active cases $3a$, $3b$, and $3c$. For the inactive cases $1$ and $2$,
the $R_{1}+R_{2}$ constraints are no longer active and the sum-rate is given
by the bounds $\overline{B}_{r,\{2\}}+\overline{B}_{d,\{1\}}$ and
$\overline{B}_{r,\{2\}}+\overline{B}_{d,\{1\}}$, respectively. We use the
following lemma on polymatroid intersections to generalize this observation
and develop an outer bound on the $K$-user sum-rate.

\begin{lemma}
\label{Lemma_PolyIntersect}Let $R_{\mathcal{S}}\leq f_{1}\left(
\mathcal{S}\right)  $ and $R_{\mathcal{S}}\leq f_{2}\left(  \mathcal{S}%
\right)  $, for all $\mathcal{S}\subseteq\mathcal{K}$, be two polymatroids
such that $f_{1}$ and $f_{2}$ are nondecreasing submodular set functions on
$\mathcal{K}$ with $f_{1}\left(  \emptyset\right)  =f_{2}\left(
\emptyset\right)  =0$. Then
\begin{equation}
\max R_{\mathcal{K}}=\min\limits_{\mathcal{S}\subseteq\mathcal{K}}\left(
f_{1}\left(  \mathcal{S}\right)  +f_{2}\left(  \mathcal{K}\backslash
\mathcal{S}\right)  \right)  .\label{DG_Lemma_RK}%
\end{equation}

\end{lemma}

From Lemma \ref{Lemma_PolyIntersect} we see that the maximum $K$-user sum-rate
$R_{\mathcal{K}}$ that results from the intersection of two polymatroids,
$R_{\mathcal{S}}\leq f_{1}\left(  \mathcal{S}\right)  $ and $R_{\mathcal{S}%
}\leq f_{2}\left(  \mathcal{S}\right)  $ is given by the minimum of the two
$K$-user sum-rate planes $f_{1}\left(  \mathcal{K}\right)  $ and $f_{2}\left(
\mathcal{K}\right)  $ only if both the sum-rates are at most as large as the
sum of the orthogonal rate planes $f_{1}\left(  \mathcal{S}\right)  $ and
$f_{2}\left(  \mathcal{K}\backslash\mathcal{S}\right)  $, for all
$\emptyset\not =\mathcal{S}\subset\mathcal{K}$. Further, the resulting
intersection belongs to the set of active cases. Conversely, when there exists
at least one $\emptyset\not =\mathcal{S}\subset\mathcal{K}$ for which the
above condition is not true, an inactive case results. Physically, an inactive
case results when a subset $\mathcal{S}$ of all users achieve better rates at
one of the receivers while the remaining subset of users achieve a better rate
at the other receiver. For such inactive cases, the maximum sum-rate in
(\ref{DG_Lemma_RK}) is the sum of two orthogonal rate planes achieved by the
two complementary subsets of users. As a result, the $K$-user sum-rate bounds
$f_{1}(\mathcal{K})$ and $f_{2}(\mathcal{K})$ are no longer active for this
case, and thus, the region of intersection is no longer a polymatroid with
$2^{K}-1$ faces.

In the following theorem we use Lemma \ref{Lemma_PolyIntersect} to develop the
upper bound on the $K$-user sum-rate. For a Gaussian input distribution, the
polymatroids $\mathcal{R}_{r}^{ob}$ and $\mathcal{R}_{d}^{ob}$ are
parametrized by $\underline{\zeta}$, and thus, Lemma \ref{Lemma_PolyIntersect}
applies for each choice of $\underline{\zeta}$.

\begin{theorem}
\label{DGOB_Th1}For each $\underline{\zeta}\in\Gamma_{OB}$, the maximum
$K$-user sum-rate $R_{\mathcal{K}}$ resulting from the intersecting
polymatroids $\mathcal{R}_{r}^{ob}$ and $\mathcal{R}_{d}^{ob}$ is
\begin{equation}
R_{\mathcal{K}}=\left\{
\begin{array}
[c]{ll}%
\overline{B}_{d,\mathcal{A}}+\overline{B}_{r,\mathcal{A}^{c}} &
\text{condition }1\text{ }\\
\min\left(  \overline{B}_{r,\mathcal{K}},\overline{B}_{d,\mathcal{K}}\right)
& \text{otherwise}%
\end{array}
\right.  \label{DG_OB_RK}%
\end{equation}
where $\overline{B}_{d,\mathcal{S}}$ and $\overline{B}_{r,\mathcal{S}}$ for
all $\mathcal{S}$ are functions of $\underline{\zeta}_{\mathcal{K}}$ and
condition $1$ is defined for any $\emptyset\not =\mathcal{A\subset K}$ as%
\begin{equation}
\overline{B}_{d,\mathcal{A}}+\overline{B}_{r,\mathcal{A}^{c}}<\min\left(
\overline{B}_{r,\mathcal{K}},\overline{B}_{d,\mathcal{K}}\right)
.\label{DGOB_NI_Cond}%
\end{equation}

\end{theorem}

\begin{remark}
The condition in (\ref{DGOB_NI_Cond}) determines whether the intersection of
two polymatroids belongs to either the set of active or the set of inactive
cases with respect to the $K$-user sum-rate.
\end{remark}

\begin{proof}
The proof follows from applying Lemma \ref{Lemma_PolyIntersect} to the
maximization of $R_{\mathcal{K}}$ for each choice of $\underline{\zeta}$.
\end{proof}

For a fixed transmit power $P_{k}$, for all $k\in\mathcal{T}$, and noise
variances $N_{r}$ and $N_{d}$, the choice of $\underline{\zeta}$ determines
whether the intersection of $\mathcal{R}_{r}^{ob}\left(  \underline{\zeta
}\right)  $ and $\mathcal{R}_{d}^{ob}\left(  \underline{\zeta}\right)  $
belongs to the set of active or inactive cases. For each choice of
$\underline{\zeta}$, from Theorem \ref{DGOB_Th1} an active case results only
if for all $2^{K}-1$ non-empty subsets $\mathcal{A}$ of $\mathcal{K}$, the
condition in (\ref{DGOB_NI_Cond}) does not hold. Further, for any
$\underline{\zeta}$ that results in an inactive case, from Theorem
\ref{DGOB_Th1}, the sum-rate is bounded as
\begin{equation}
\overline{B}_{d,\mathcal{A}}+\overline{B}_{r,\mathcal{A}^{c}}<\min\left(
\overline{B}_{r,\mathcal{K}},\overline{B}_{d,\mathcal{K}}\right)
<\max_{\underline{\zeta}\in\Gamma_{OB}}\min\left(  \overline{B}_{r,\mathcal{K}%
},\overline{B}_{d,\mathcal{K}}\right)  .
\end{equation}
To this end, we consider the optimization problem
\begin{equation}
R_{\mathcal{K}}=\max\limits_{\underline{\zeta}\in\Gamma_{OB}}\min\left(
\overline{B}_{r,\mathcal{K}}\left(  \underline{\zeta}\right)  ,\overline
{B}_{d,\mathcal{K}}\left(  \underline{\zeta}\right)  \right)  .
\label{DGOB_RK_maxmin}%
\end{equation}
In general, optimizing non-convex functions is not straightforward. However,
since $B_{r,\mathcal{K}}$ and $B_{d,\mathcal{K}}$ are concave functions of
$\underline{\gamma}_{\mathcal{K}}$, the above max-min optimization simplifies
to%
\begin{equation}
R_{\mathcal{K}}=\max\limits_{\underline{\gamma}_{\mathcal{K}}\in\Gamma_{OB}%
}\min\left\{  B_{r,\mathcal{K}}\left(  \underline{\gamma}_{\mathcal{K}%
}\right)  ,B_{d,\mathcal{K}}\left(  \underline{\gamma}_{\mathcal{K}}\right)
\right\}  . \label{DGOB_RK_MM1}%
\end{equation}
Note that the optimization is performed over the same set in
(\ref{DGOB_RK_maxmin}) and (\ref{DGOB_RK_MM1}) as $\Gamma_{OB}$ is a closed
convex set. In Appendix \ref{DG_App4OBProof}, we show that the
\textit{max-min} problem in (\ref{DGOB_RK_MM1}) is a dual of the classical
\textit{minimax} problem of detection theory, (see for e.g., \cite[II.C]%
{cap_theorems:HVPoor01}). This allows us to apply the techniques used to
obtain a minimax solution to maximize the bounds in (\ref{DGOB_RK_MM1}) over
all $\underline{\gamma}_{\mathcal{K}}$ in $\Gamma_{OB}$ (see also
\cite{cap_theorems:LiangVP_ResAllocJrnl}). We write $\underline{\gamma
}_{\mathcal{K}}^{\ast}$ to denote a sum-rate optimal allocation, i.e., a
\textit{max-min rule}, and write $\mathcal{G}$ to denote the set of all
$\underline{\gamma}_{\mathcal{K}}^{\ast}$ maximizing (\ref{DGOB_RK_MM1}). A
general solution to the max-min optimization in (\ref{DGOB_RK_MM1}) simplifies
to three cases \cite[II.C]{cap_theorems:HVPoor01}. The first two correspond to
those in which the maximum achieved by one of the two functions is smaller
than the other, while the third corresponds to the case in which the maximum
results when the two functions are equal (see Fig. \ref{Fig_Prop_1}). For
$B_{r,\mathcal{K}}$ and $B_{d,\mathcal{K}}$ defined in (\ref{Con_final_B1S})
and (\ref{Con_final_B2S}), respectively, we now show that the solution
simplifies to the consideration of only two cases. The following theorem
summarizes the solution to the max-min problem in (\ref{DGOB_RK_MM1}). The
proof is developed in Appendix \ref{DG_App4OBProof}.

\begin{theorem}
\label{DGOB_Th2}The max-min optimization
\begin{equation}
R_{\mathcal{K}}=\max\limits_{\underline{\gamma}_{\mathcal{K}}\in\Gamma_{OB}%
}\min\left\{  B_{r,\mathcal{K}}\left(  \underline{\gamma}_{\mathcal{K}%
}\right)  ,B_{d,\mathcal{K}}\left(  \underline{\gamma}_{\mathcal{K}}\right)
\right\}  \label{DGOB_RK_MM2}%
\end{equation}
simplifies to the following two cases.%
\begin{equation}%
\begin{array}
[c]{lll}%
\text{Case }1\text{:} & R_{\mathcal{K}}=C\left(  \frac{%
%TCIMACRO{\tsum \limits_{k\in\mathcal{K}}}%
%BeginExpansion
{\textstyle\sum\limits_{k\in\mathcal{K}}}
%EndExpansion
P_{k}}{N_{r}}\right)  , & B_{r,\mathcal{K}}\left(  \underline{0}\right)
<B_{d,\mathcal{K}}\left(  \underline{0}\right) \\
\text{Case }2\text{:} & R_{\mathcal{K}}=C\left(  \left(
%TCIMACRO{\tsum \limits_{k\in\mathcal{K}}}%
%BeginExpansion
{\textstyle\sum\limits_{k\in\mathcal{K}}}
%EndExpansion
\frac{P_{k}}{N_{r}}\right)  -\frac{\left(  x^{\ast}\right)  ^{2}P_{\max}%
}{N_{r}}\right)  \equiv B^{\ast}, & B_{r,\mathcal{K}}^{\ast}\left(
\underline{\gamma}_{\mathcal{K}}^{\ast}\right)  =B_{d,\mathcal{K}}^{\ast
}\left(  \underline{\gamma}_{\mathcal{K}}^{\ast}\right)
\end{array}
\label{DGOB_RKmax}%
\end{equation}
where $P_{\max}=\max_{k\in\mathcal{K}}P_{k}$, $\lambda_{k}\overset
{\vartriangle}{=}P_{k}/P_{\max}$, and $x^{\ast}\overset{\vartriangle}{=}%
\sum_{k\in\mathcal{K}}\sqrt{\lambda_{k}\gamma_{k}^{\ast}}$ is the unique
solution satisfying $B_{r.\mathcal{K}}\left(  \underline{\gamma}_{\mathcal{K}%
}^{\ast}\right)  =B_{r.\mathcal{K}}\left(  \underline{\gamma}_{\mathcal{K}%
}^{\ast}\right)  $ and is given by%
\begin{equation}
x^{\ast}=\frac{-K_{1}+\sqrt{K_{1}^{2}+\left(  K_{3}-K_{2}\right)  K_{0}}%
}{K_{0}} \label{DGOB_xstar}%
\end{equation}
with
\begin{equation}%
\begin{array}
[c]{ll}%
K_{0}=P_{\max}\left/  N_{r}\right.  , & K_{1}=\sqrt{P_{\max}P_{r}}\left/
N_{d}\right. \\
K_{2}=\frac{\sum\nolimits_{k\in\mathcal{K}}P_{k}}{N_{d}}+\frac{P_{r}}{N_{d}%
},\text{ and} & K_{3}=\frac{%
%TCIMACRO{\tsum \nolimits_{k\in\mathcal{K}}}%
%BeginExpansion
{\textstyle\sum\nolimits_{k\in\mathcal{K}}}
%EndExpansion
P_{k}}{N_{r}}.
\end{array}
\label{DGOB_xstar1}%
\end{equation}

\end{theorem}

\begin{remark}
The maximization in (\ref{DGOB_RK_MM2}) is independent of whether the optimal
$\underline{\gamma}_{\mathcal{K}}^{\ast}$ results in an active or an inactive
case. However, not all max-min rules $\underline{\gamma}_{\mathcal{K}}^{\ast
}\in\mathcal{G}$ will result in an active case. In general, active cases may
be achieved only by a subset $\mathcal{G}_{a}$ $\subseteq\mathcal{G}$.
However, irrespective of the kind of intersection, from Lemma
\ref{Lemma_PolyIntersect}, (\ref{DGOB_RKmax}) is an upper bound on the
$K$-user sum-rate cutset bounds.
\end{remark}

In the following theorem we show that it suffices to consider two conditions
in determining the largest outer bound on the $K$-user sum-capacity. We
enumerate the two conditions as%
\begin{equation}%
\begin{array}
[c]{cc}%
\text{Condition 1:} & B_{r,\mathcal{K}}(\underline{0})\leq B_{d,\mathcal{K}%
}(\underline{0})\\
\text{Condition 2:} & B_{r,\mathcal{K}}(\underline{0})>B_{d,\mathcal{K}%
}(\underline{0}).
\end{array}
\label{DG_OB_Conds}%
\end{equation}
The first condition implies that the maximum $K$-user cutset bound at the
relay is smaller than the corresponding bound at the destination; for this
case, we show that $B_{r,\mathcal{S}}(\underline{0})<B_{d,\mathcal{S}%
}(\underline{0})$ for all $\mathcal{S}\subset\mathcal{K}$, i.e.,
$\mathcal{R}_{OB}=\mathcal{R}_{r}^{ob}\subset\mathcal{R}_{d}^{ob}$. On the
other hand, when condition 2 occurs, i.e., when condition 1 does not hold in
(\ref{DG_OB_Conds}), we use the monotone properties of $B_{r,\mathcal{K}}$ and
$B_{d,\mathcal{K}}$ and Lemma \ref{Lemma_PolyIntersect} to show that
\begin{equation}
R_{K}\leq\max\limits_{\underline{\gamma}_{\mathcal{K}}\in\Gamma_{OB}}%
\min\left\{  B_{r,\mathcal{K}}\left(  \underline{\gamma}_{\mathcal{K}}\right)
,B_{d,\mathcal{K}}\left(  \underline{\gamma}_{\mathcal{K}}\right)  \right\}
\label{DG_OB_RKbound}%
\end{equation}
with equality achieved in (\ref{DG_OB_RKbound}) when the polymatroid
intersection is an active case. From Theorem \ref{DGOB_Th2} we have that a
continuous set, $\mathcal{G}$, of $\underline{\gamma}_{\mathcal{K}}^{\ast}$
maximizes the right-hand-side of (\ref{DG_OB_RKbound}). We show that the bound
in (\ref{DG_OB_RKbound}) is achieved with equality when there exists a
$\underline{\gamma}_{\mathcal{K}}^{\ast}$ that results in an active case,
i.e., in a non-empty $\mathcal{G}_{a}$. Finally, for the class of symmetric
degraded G-MARCs, we prove the existence of an active case that maximizes the sum-rate.

\begin{theorem}
\label{DGOB_Th3}The largest outer bound $R_{\mathcal{K}}^{ob}$ on the $K$-user
sum-rate is
\begin{equation}%
\begin{array}
[c]{ll}%
R_{\mathcal{K}}^{ob}=C\left(  \sum_{k\in\mathcal{K}}P_{k}/N_{r}\right)  , &
\text{if }B_{r,\mathcal{K}}(\underline{0})<B_{d,\mathcal{K}}(\underline{0})\\
R_{\mathcal{K}}^{ob}\leq C\left(  \left(
%TCIMACRO{\tsum \limits_{k\in\mathcal{K}}}%
%BeginExpansion
{\textstyle\sum\limits_{k\in\mathcal{K}}}
%EndExpansion
\frac{P_{k}}{N_{r}}\right)  -\frac{\left(  x^{\ast}\right)  ^{2}P_{\max}%
}{N_{r}}\right)  , & \text{otherwise}%
\end{array}
\label{RK_rate_defn}%
\end{equation}
where $x^{\ast}\overset{\vartriangle}{=}\sum_{k\in\mathcal{K}}\sqrt
{\lambda_{k}\gamma_{k}^{\ast}}$ is the unique solution satisfying
$B_{r,\mathcal{K}}(\underline{\gamma}^{\ast})=B_{d,\mathcal{K}}(\underline
{\gamma}^{\ast})$ and is given by (\ref{DGOB_xstar}) and (\ref{DGOB_xstar1}).
The bound in (\ref{RK_rate_defn}) is achieved with equality only when the
intersection of $\mathcal{R}_{r}^{ob}(\underline{\gamma}_{\mathcal{K}}^{\ast
})$ and $\mathcal{R}_{d}^{ob}(\underline{\gamma}_{\mathcal{K}}^{\ast})$
results in an active case. The bound is achieved with equality for the class
of symmetric degraded G-MARCs.
\end{theorem}

\begin{proof}
Let $\underline{\gamma}_{\mathcal{K}}=\underline{0}$. From
(\ref{Con_final_B1S}) we see that $B_{\mathcal{S},r}(\underline{\gamma
}_{\mathcal{S}}\not =\underline{0})<B_{\mathcal{S},r}(\underline{0})$, for all
$\mathcal{S}\subseteq\mathcal{K}$, i.e, the region $\mathcal{R}_{r}%
^{(ob)}(\underline{\gamma}_{\mathcal{K}})$ is largest at $\underline{\gamma
}_{\mathcal{K}}=\underline{0}$. Expanding $B_{\mathcal{S},r}$ and
$B_{\mathcal{S},d}$ at $\underline{\gamma}_{\mathcal{K}}=\underline{0}$ from
(\ref{Con_final_B1S}) and (\ref{Con_final_B2S}), respectively, we have
\begin{align}
B_{r,\mathcal{S}}\left(  \underline{0}\right)   &  =C\left(  \frac{\sum
_{k\in\mathcal{S}}P_{k}}{N_{r}}\right) \label{DGOB_BrS0}\\
B_{d,\mathcal{S}}\left(  \underline{0}\right)   &  =C\left(  \frac{\sum
_{k\in\mathcal{S}}P_{k}}{N_{d}}+\frac{P_{r}}{N_{d}}\right)  .
\label{DGOB_BdS0}%
\end{align}
The sum-rate resulting from the intersection of $\mathcal{R}_{r}^{ob}\left(
\underline{0}\right)  $ and $\mathcal{R}_{d}^{ob}\left(  \underline{0}\right)
$ falls into one of following two cases.

\textit{Case 1}: The first case results when $B_{\mathcal{K},r}\left(
\underline{0}\right)  \leq B_{\mathcal{K},d}\left(  \underline{0}\right)  $.
From (\ref{DGOB_BrS0}) and (\ref{DGOB_BdS0}) this condition simplifies to
\begin{equation}
\frac{\sum\limits_{k\in\mathcal{K}}P_{k}}{N_{r}}\leq\frac{\sum\limits_{k\in
\mathcal{K}}P_{k}}{N_{d}}+\frac{P_{r}}{N_{d}}.\label{DGOB_Case1_CondK}%
\end{equation}
Expanding (\ref{DGOB_Case1_CondK}), we have, for any $\mathcal{S\subset K}$,
\begin{align}
\frac{%
%TCIMACRO{\tsum \limits_{k\in\mathcal{S}}}%
%BeginExpansion
{\textstyle\sum\limits_{k\in\mathcal{S}}}
%EndExpansion
P_{k}}{N_{r}} &  \leq\frac{%
%TCIMACRO{\tsum \limits_{k\in\mathcal{S}}}%
%BeginExpansion
{\textstyle\sum\limits_{k\in\mathcal{S}}}
%EndExpansion
P_{k}+P_{r}}{N_{d}}-\frac{%
%TCIMACRO{\tsum \limits_{k\in\mathcal{S}^{c}}}%
%BeginExpansion
{\textstyle\sum\limits_{k\in\mathcal{S}^{c}}}
%EndExpansion
P_{k}\left(  N_{d}-N_{r}\right)  }{N_{d}N_{r}}\\
&  <\frac{%
%TCIMACRO{\tsum \limits_{k\in\mathcal{S}}}%
%BeginExpansion
{\textstyle\sum\limits_{k\in\mathcal{S}}}
%EndExpansion
P_{k}+P_{r}}{N_{d}}\label{DGOB_Case1_CondS}%
\end{align}
where (\ref{DGOB_Case1_CondS}) follows from the degradedness condition in
(\ref{DGMARC_Noise_var}). Thus, $B_{r,\mathcal{K}}(\underline{0})\leq
B_{d,\mathcal{K}}(\underline{0})$ implies that $B_{r,\mathcal{S}}%
(\underline{0})<B_{d,\mathcal{S}}(\underline{0})$ for all $\mathcal{S}%
\subset\mathcal{K}$, i.e., $\mathcal{R}_{r}^{ob}(\underline{0})\subset
\mathcal{R}_{d}^{ob}(\underline{0})$, and $\mathcal{R}_{OB}(\underline
{0})=\mathcal{R}_{r}^{ob}(\underline{0})$. The maximum $K$-user sum-rate upper
bound for this active case is then
\begin{equation}
R_{\mathcal{K}}^{ob}=B_{r,\mathcal{K}}=C(%
%TCIMACRO{\tsum \nolimits_{k\in\mathcal{K}}}%
%BeginExpansion
{\textstyle\sum\nolimits_{k\in\mathcal{K}}}
%EndExpansion
P_{k}\left/  N_{r}\right.  ).
\end{equation}

\textit{Case 2}: The second case results when $B_{\mathcal{K},r}\left(
\underline{0}\right)  >B_{\mathcal{K},d}\left(  \underline{0}\right)  ,$ i.e.,
when
\begin{equation}
\frac{\sum\limits_{k\in\mathcal{K}}P_{k}}{N_{r}}>\frac{\sum\limits_{k\in
\mathcal{K}}P_{k}}{N_{d}}+\frac{P_{r}}{N_{d}}.\label{DGOB_Case2_CondK}%
\end{equation}
Unlike \textit{case }$1$, (\ref{DGOB_Case2_CondK}) does not imply that
$B_{\mathcal{S},r}\left(  \underline{0}\right)  >B_{\mathcal{S},d}\left(
\underline{0}\right)  $ or vice-versa, for all $\mathcal{S}\subset\mathcal{K}%
$. From Theorem \ref{DGOB_Th1}, the intersection of $\mathcal{R}_{r}%
^{ob}\left(  \underline{0}\right)  $ and $\mathcal{R}_{d}^{ob}\left(
\underline{0}\right)  $ can result in either an active or an inactive case and
thus, from (\ref{DGOB_NI_Cond}), we have
\begin{equation}
R_{\mathcal{K}}^{ob}\leq\min(B_{r,\mathcal{K}}\left(  \underline{0}\right)
,B_{d,\mathcal{K}}\left(  \underline{0}\right)  )=B_{d,\mathcal{K}}\left(
\underline{0}\right)
\end{equation}
with equality for the active case. Note that from symmetry an active case
results for the symmetric G-MARC. We now show that the sum-rate is increased
for a $\underline{\gamma}_{\mathcal{K}}^{\ast}\not =\underline{0}$ such that
$B_{r,\mathcal{K}}\left(  \underline{\gamma}_{\mathcal{K}}^{\ast}\right)
=B_{d,\mathcal{K}}\left(  \underline{\gamma}_{\mathcal{K}}^{\ast}\right)  $.
To simplify the exposition, we write $B_{r,\mathcal{K}}$ and $B_{d,\mathcal{K}%
}$ in (\ref{Con_final_B1S}) and (\ref{Con_final_B2S}) as%
\begin{align}
B_{r,\mathcal{K}}\left(  x\right)   &  =C\left(  \frac{%
%TCIMACRO{\tsum \limits_{k\in\mathcal{K}}}%
%BeginExpansion
{\textstyle\sum\limits_{k\in\mathcal{K}}}
%EndExpansion
P_{k}}{N_{r}}-\frac{x^{2}P_{\max}}{N_{r}}\right)  \label{DGOB_BrKx}\\
B_{d,\mathcal{K}}\left(  x\right)   &  =C\left(  \frac{%
%TCIMACRO{\tsum \limits_{k\in\mathcal{K}}}%
%BeginExpansion
{\textstyle\sum\limits_{k\in\mathcal{K}}}
%EndExpansion
P_{k}}{N_{d}}+\frac{P_{r}}{N_{d}}+\frac{2x\sqrt{P_{\max}P_{r}}}{N_{d}}\right)
\label{DGOB_BdKx}%
\end{align}
where
\begin{equation}
x\overset{\vartriangle}{=}\sum_{k=1}^{K}\sqrt{\gamma_{k}\lambda_{k}%
}\label{OB_x_def}%
\end{equation}
and $\lambda_{k}=P_{k}/P_{\max}$ where $P_{\max}=\max_{k\in\mathcal{K}}P_{k}$,
for all $k$. For all $\gamma_{k}\in\lbrack0,1]$, we have
\begin{align}
&
\begin{array}
[c]{cc}%
\frac{\partial x}{\partial\gamma_{k}}=\frac{\sqrt{\lambda_{k}}}{2\sqrt
{\gamma_{k}}} & k\in\mathcal{K}%
\end{array}
\\
&
\begin{array}
[c]{cc}%
\frac{\partial^{2}x}{\partial\gamma_{k}^{2}}=-\frac{\sqrt{\lambda_{k}}%
}{4\gamma_{k}^{3/2}} & k\in\mathcal{K}%
\end{array}
\\
&
\begin{array}
[c]{cc}%
\frac{\partial^{2}x}{\partial\gamma_{k}\partial\gamma_{j}}=0 & k\not =%
j\text{.}%
\end{array}
\end{align}
Thus, $x$ is a concave function of $\underline{\gamma}_{\mathcal{K}}$ over the
hyper-cube $\gamma_{k}\in\lbrack0,1]$, for all $k$, and therefore, is concave
for all $\gamma_{k}$ satisfying (\ref{DGOB_TOB}). Further, from
(\ref{DGOB_TOB}), we see that $x$ is maximized when the entries of
$\underline{\gamma}_{\mathcal{K}}$ satisfy $\sum_{k=1}^{K}\gamma_{k}=1$. Using
techniques similar to those in Appendix \ref{DG_AppConvex}, one can show that
$x$ achieves its maximum for a $\underline{\gamma}_{\mathcal{K}}^{\prime}$
with entries
\begin{equation}%
\begin{array}
[c]{cc}%
\gamma_{k}^{\prime}=\frac{\lambda_{k}}{\sum_{k=1}^{K}\lambda_{k}} & \text{for
all }k\text{,}%
\end{array}
\label{DGOB_Gprime}%
\end{equation}
and thus, we have
\begin{equation}
x\in\left[  0,\sqrt{\sum\nolimits_{k=1}^{K}\lambda_{k}}\right]  \subseteq
\lbrack0,\sqrt{K}]\text{.}\label{DGOB_x_range}%
\end{equation}
The functions $B_{r,\mathcal{K}}\left(  x\right)  $ and $B_{d,\mathcal{K}%
}\left(  x\right)  $ in (\ref{DGOB_BrKx}) and (\ref{DGOB_BdKx}) are
monotonically decreasing and increasing functions of $x$, respectively.
Substituting (\ref{DGOB_Gprime}) in (\ref{Con_final_B1S}), we have
$B_{r,\mathcal{S}}(\underline{\gamma}_{\mathcal{K}}^{\prime})=0$ for all
$\mathcal{S}\subseteq\mathcal{K}$. Thus, for the case in which $B_{\mathcal{K}%
,r}\left(  \underline{0}\right)  >B_{\mathcal{K},d}\left(  \underline
{0}\right)  $, one can shrink the region $\mathcal{R}_{r}^{ob}$ from
$\mathcal{R}_{r}^{ob}\left(  \underline{0}\right)  $ just sufficiently such
that for some $\underline{\gamma}_{\mathcal{K}}^{\ast}\not =\underline{0}$,
$B_{r,\mathcal{K}}\left(  \underline{\gamma}_{\mathcal{K}}^{\ast}\right)
=B_{d,\mathcal{K}}\left(  \underline{\gamma}_{\mathcal{K}}^{\ast}\right)
>B_{\mathcal{K},d}\left(  \underline{0}\right)  $. In Theorem \ref{DGOB_Th3}
we show that $B_{\mathcal{K},r}=B_{\mathcal{K},d}$ is maximized by a set
$\mathcal{G}$ of $\underline{\gamma}_{\mathcal{K}}^{\ast}$ satisfying
\begin{equation}
\mathcal{G}=\left\{  \underline{\gamma}_{\mathcal{K}}^{\ast}:%
%TCIMACRO{\tsum \limits_{k\in\mathcal{K}}}%
%BeginExpansion
{\textstyle\sum\limits_{k\in\mathcal{K}}}
%EndExpansion
\gamma_{k}^{\ast}\lambda_{k}=\left(  x^{\ast}\right)  ^{2}\right\}
\label{OB_gammak_opt}%
\end{equation}
where $x^{\ast}$ is the unique value satisfying the quadratic $B_{\mathcal{K}%
,r}\left(  x\right)  =B_{\mathcal{K},d}\left(  x\right)  $. For $\underline
{\gamma}_{\mathcal{K}}^{\ast}\not =\underline{0}$, from (\ref{Con_final_B1S})
one can verify that $B_{r,\mathcal{S}}\left(  \underline{\gamma}_{\mathcal{K}%
}^{\ast}\right)  <B_{r,\mathcal{S}}\left(  \underline{0}\right)  $ for all
$\mathcal{S}$, i.e., $\mathcal{R}_{r}^{ob}\left(  \underline{\gamma
}_{\mathcal{K}}^{\ast}\right)  \subset\mathcal{R}_{r}^{ob}\left(
\underline{0}\right)  $. On the other hand, substituting $\underline{\gamma
}_{\mathcal{K}}^{\ast}$ in (\ref{Con_final_B2S}), $B_{d,\mathcal{S}}$ for all
$\mathcal{S}\subset\mathcal{K}$ simplifies to
\begin{equation}
B_{d,\mathcal{S}}\left(  \underline{\gamma}_{\mathcal{K}}^{\ast}\right)
=C\left(  \frac{%
%TCIMACRO{\tsum \limits_{k\in\mathcal{S}}}%
%BeginExpansion
{\textstyle\sum\limits_{k\in\mathcal{S}}}
%EndExpansion
P_{k}+(1-\sum\nolimits_{k\in\mathcal{S}^{c}}\gamma_{k}^{\ast})P_{r}+2%
%TCIMACRO{\tsum \limits_{k\in\mathcal{S}}}%
%BeginExpansion
{\textstyle\sum\limits_{k\in\mathcal{S}}}
%EndExpansion
\sqrt{\gamma_{k}^{\ast}P_{k}P_{r}}}{N_{d}}\right)  .\label{DGOB_BdSG}%
\end{equation}
Comparing $B_{d,\mathcal{S}}\left(  \underline{0}\right)  $ in
(\ref{DGOB_BdS0}) with $B_{d,\mathcal{S}}\left(  \underline{\gamma
}_{\mathcal{K}}^{\ast}\right)  $ in (\ref{DGOB_BdSG}) above, one cannot in
general show that $\mathcal{R}_{d}^{ob}\left(  \underline{\gamma}%
_{\mathcal{K}}^{\ast}\right)  \supseteq\mathcal{R}_{d}^{ob}\left(
\underline{0}\right)  $. In fact, the $\underline{\gamma}_{\mathcal{K}}^{\ast
}$ chosen will determine the relationship between $B_{d,\mathcal{S}%
}(\underline{\gamma}_{\mathcal{K}}^{\ast})$ and $B_{d,\mathcal{S}}%
(\underline{0})$ for any $\mathcal{S}$. Thus, for any $\underline{\gamma
}_{\mathcal{K}}^{\ast}$ that equalizes $B_{r,\mathcal{K}}$ and
$B_{d,\mathcal{K}}$ the polytope $\mathcal{R}_{r}^{ob}\cap\mathcal{R}_{d}%
^{ob}$ belongs to either the set of active or inactive cases. Recall that we
write $\mathcal{G}_{a}$ to denote the set of $\underline{\gamma}_{\mathcal{K}%
}^{\ast}$ that results in an active case, i.e., the set of $\underline{\gamma
}_{\mathcal{K}}^{\ast}$ for which the condition in (\ref{DGOB_NI_Cond}) does
not hold for all $2^{K}-1$ non-empty subsets $\mathcal{A}$ of $\mathcal{K}$.
From Theorem \ref{DGOB_Th1}, we have that the sum-rate for the inactive case
is always bounded by the maximum sum-rate developed in Theorem \ref{DGOB_Th2}.
Thus, the maximum $K$-user sum-rate when $B_{r,\mathcal{K}}(\underline
{0})>B_{d,\mathcal{K}}(\underline{0})$ is
\begin{equation}
R_{\mathcal{K}}=\left\{
\begin{array}
[c]{ll}%
B_{d,\mathcal{K}}(\underline{\gamma}_{\mathcal{K}}^{\ast})=B_{r,\mathcal{K}%
}(\underline{\gamma}_{\mathcal{K}}^{\ast})\equiv B^{\ast} & \underline{\gamma
}_{\mathcal{K}}^{\ast}\in\mathcal{G}_{a}\not =\emptyset\\
\max\limits_{\underline{\xi}}\overline{B}_{d,\mathcal{A}}(\underline{\xi
})+\overline{B}_{r,\mathcal{A}^{c}}(\underline{\xi})<B^{\ast} & \mathcal{G}%
_{a}=\emptyset
\end{array}
\right.  \label{DGOB_Th3RKmax}%
\end{equation}
where $B^{\ast}$ is defined in Theorem \ref{DGOB_Th2}. We now show that for
the class of symmetric G-MARC channels the bound $B^{\ast}$ is achieved, i.e.,
$\mathcal{G}_{a}\not =\emptyset$. For this class since $P_{k}=P$ for all
$k\in\mathcal{K}$, from symmetry $B^{\ast}$ can be achieved by choosing
$\gamma_{k}^{\ast}=\gamma^{\ast}$ for all $k$ such that from (\ref{OB_x_def}),
we have%
\begin{equation}
\gamma^{\ast}=\left(  x^{\ast}\right)  ^{2}/K^{2}.\label{DGOB_gam_sym}%
\end{equation}
From (\ref{DGOB_x_range}), since $0\leq x^{\ast}\leq\sqrt{K}$, there exists an
$\gamma^{\ast}<1$. From (\ref{DGOB_TOB}), we also require $\gamma^{\ast}<1/K$.
In Theorem \ref{DGDF_Th3} in Section \ref{DG_Sec4} below, we prove the
existence of a $\gamma^{\ast}<1/K$ for symmetric channels. From symmetry,
since no subset of users can achieve better rates at one receiver than the
other, the resulting $\mathcal{R}_{r}\left(  \gamma^{\ast}\right)
\cap\mathcal{R}_{d}\left(  \gamma^{\ast}\right)  $ belongs to the set of
inactive cases. The $K$-user sum-rate cutset bound for this class is given by
the $B^{\ast}$ in (\ref{DGOB_RKmax}) with $P_{\max}=P$ and $\lambda_{k}=1$ for
all $k\in\mathcal{K}$.\newline\qquad Finally, from continuity, one can expect
that for small perturbations of user powers from the symmetric case, an active
case will result. However, for arbitrary user powers, it is possible that
$\mathcal{G}_{a}=\emptyset$, i.e., the set of all feasible $\underline{\gamma
}_{\mathcal{K}}^{\ast}$ results in non-inactive cases. In general, however,
obtaining a closed-form expression for the maximum sum-rate for the inactive
cases is not straightforward.
\end{proof}

\section{\label{DG_Sec4}Decode-and-Forward\ }

A DF code construction for a discrete memoryless MARC using block Markov
encoding and backward decoding is developed in \cite[Appendix~A]%
{cap_theorems:KGG_IT} (see also \cite{cap_theorems:SKM_journal01}) and we
extend it here to the degraded Gaussian MARC. We first summarize the rate
region achieved by DF below.

\begin{proposition}
\label{Prop_MARC_DF}The DF rate region is the union of the set of rate tuples
$(R_{1},R_{2},\ldots,$ $R_{K})$ that satisfy, for all $\mathcal{S}%
\subseteq\mathcal{K}$,%
\begin{equation}
R_{\mathcal{S}}\leq\min\left\{  I(X_{\mathcal{S}};Y_{r}|X_{\mathcal{S}^{c}%
}V_{\mathcal{K}}X_{r}U),I(X_{\mathcal{S}}X_{r};Y_{d}|X_{\mathcal{S}^{c}%
}V_{\mathcal{S}^{c}}U)\right\}  \label{Prop_DF_rateregion}%
\end{equation}
where the union is over all distributions that factor as%
\begin{equation}
p(u)\cdot\left(
%TCIMACRO{\tprod \nolimits_{k=1}^{K}}%
%BeginExpansion
{\textstyle\prod\nolimits_{k=1}^{K}}
%EndExpansion
p(v_{k}|u)p(x_{k}|v_{k},u)\right)  \cdot p(x_{r}|v_{\mathcal{K}},u)\cdot
p(y_{r},y_{d}|x_{\mathcal{T}}). \label{Prop_inp_dist}%
\end{equation}

\end{proposition}

\begin{proof}
See \cite{cap_theorems:SKM_journal01}.
\end{proof}

\begin{remark}
The \textit{time-sharing} random variable $U$ ensures that the region of
Theorem \ref{Prop_MARC_DF} is convex.
\end{remark}

\begin{remark}
The independent auxiliary random variables $V_{k}$, $k=1,2,\ldots,K$, help the
sources cooperate with the relay.
\end{remark}

For the degraded Gaussian\ MARC, we employ the following code construction. We
generate zero-mean, unit variance, independent and identically distributed
(i.i.d.) Gaussian random variables $V_{k}$, $V_{k,0},$ and $V_{r,0}$, for all
$k\in\mathcal{K}$, such that the channel inputs from source $k$ and the relay
are
\begin{align}
&
\begin{array}
[c]{cc}%
X_{k}=\sqrt{\alpha_{k}P_{k}}V_{k,0}+\sqrt{\left(  1-\alpha_{k}\right)  P_{k}%
}V_{k},\text{ \ \ \ \ \ \ \ \ \ \ } & k\in\mathcal{K},
\end{array}
\label{GMARC_DF_source_sig}\\
&
\begin{array}
[c]{cc}%
X_{r}=\sum\limits_{k=1}^{K}\sqrt{\beta_{k}P_{r}}V_{k}+\sqrt{\left(  1-%
%TCIMACRO{\tsum \limits_{k=1}^{K}}%
%BeginExpansion
{\textstyle\sum\limits_{k=1}^{K}}
%EndExpansion
\beta_{k}\right)  P_{r}}V_{r,0}\text{ } &
\end{array}
\label{GMARC_deg_relaysig}%
\end{align}
where $\alpha_{k}\in\lbrack0,1]$ and $\beta_{k}\in\lbrack0,1]$ are power
fractions for all $k$. We write
\begin{align}
\underline{\alpha}_{\mathcal{K}}  &  =%
\begin{pmatrix}
\alpha_{1}, & \alpha_{2}, & \ldots, & \alpha_{K}%
\end{pmatrix}
\label{alphaK_defn}\\
\underline{\beta}_{\mathcal{K}}  &  =%
\begin{pmatrix}
\beta_{1}, & \beta_{2}, & \ldots, & \beta_{K}%
\end{pmatrix}
\label{betaK_defn}%
\end{align}
and
\begin{equation}
\Gamma=\left\{  \left(  \underline{\alpha}_{\mathcal{K}},\underline{\beta
}_{\mathcal{K}}\right)  :%
\begin{array}
[c]{cc}%
\alpha_{k}\in\lbrack0,1],0\leq\sum_{k\in\mathcal{K}}\beta_{k}\leq1 & \text{for
all }k.
\end{array}
\right\}  \label{alp_beta_bounds}%
\end{equation}
for the set of feasible power fractions $\underline{\alpha}_{\mathcal{K}}$ and
$\underline{\beta}_{\mathcal{K}}$. Substituting (\ref{GMARC_DF_source_sig})
and (\ref{GMARC_deg_relaysig}) in (\ref{Prop_DF_rateregion}), for any
$(\underline{\alpha}_{\mathcal{K}},\underline{\beta}_{\mathcal{K}})\in\Gamma$,
we obtain%
\begin{equation}%
\begin{array}
[c]{cc}%
R_{\mathcal{S}}\leq\min\left(  I_{r,\mathcal{S}}\left(  \underline{\alpha
}_{\mathcal{K}}\right)  ,I_{d,\mathcal{S}}\left(  \underline{\alpha
}_{\mathcal{K}},\underline{\beta}_{\mathcal{K}}\right)  \right)  & \text{for
all }\mathcal{S}\subseteq\mathcal{K}%
\end{array}
\label{Ratebounds_1}%
\end{equation}
where $I_{r,\mathcal{S}}$ and $I_{d,\mathcal{S}}$, the bounds at the relay and
destination respectively, are
\begin{align}
I_{r,\mathcal{S}}  &  =C\left(  \frac{%
%TCIMACRO{\tsum \limits_{k\in\mathcal{S}}}%
%BeginExpansion
{\textstyle\sum\limits_{k\in\mathcal{S}}}
%EndExpansion
\alpha_{k}P_{k}}{N_{r}}\right) \label{GMARC_DF_IrG}\\
I_{d,\mathcal{S}}  &  =C\left(  \frac{%
%TCIMACRO{\tsum \limits_{k\in\mathcal{S}}}%
%BeginExpansion
{\textstyle\sum\limits_{k\in\mathcal{S}}}
%EndExpansion
P_{k}}{N_{d}}+\frac{\left(  1-%
%TCIMACRO{\tsum \limits_{k\in\mathcal{S}^{c}}}%
%BeginExpansion
{\textstyle\sum\limits_{k\in\mathcal{S}^{c}}}
%EndExpansion
\beta_{k}\right)  P_{r}}{N_{d}}+2%
%TCIMACRO{\tsum \limits_{k\in\mathcal{S}}}%
%BeginExpansion
{\textstyle\sum\limits_{k\in\mathcal{S}}}
%EndExpansion
\sqrt{\left(  1-\alpha_{k}\right)  \beta_{k}\frac{P_{k}}{N_{d}}\frac{P_{r}%
}{N_{d}}}\right)  . \label{GMARC_DF_IdG}%
\end{align}
\newline From the concavity of the $\log$ function it follows that
$I_{r,\mathcal{S}}$, for all $\mathcal{S}$, is a concave function of
$\underline{\alpha}_{\mathcal{K}}$. In Appendix \ref{DG_AppConvex} we show
that $I_{d,\mathcal{S}}$ is a concave function of $\underline{\alpha
}_{\mathcal{K}}$ and $\underline{\beta}_{\mathcal{K}}$. The DF rate region,
$\mathcal{R}_{DF}$, achieved over all $(\underline{\alpha}_{\mathcal{K}%
},\underline{\beta}_{\mathcal{K}})\in\Gamma$, is then given by the following theorem.

\begin{theorem}
The DF\ rate region $\mathcal{R}_{DF}$ for a degraded Gaussian MARC is
\begin{equation}
\mathcal{R}_{DF}=\bigcup\limits_{\left(  \underline{\alpha}_{\mathcal{K}%
},\underline{\beta}_{\mathcal{K}}\right)  \in\Gamma}\left(  \mathcal{R}%
_{r}\left(  \underline{\alpha}_{\mathcal{K}}\right)  \cap\mathcal{R}%
_{d}\left(  \underline{\alpha}_{\mathcal{K}},\underline{\beta}_{\mathcal{K}%
}\right)  \right)  \label{GMARC_R_DF}%
\end{equation}
where the rate region $\mathcal{R}_{t}$, $t=r,d$, is
\begin{equation}
\mathcal{R}_{t}\left(  \underline{\alpha}_{\mathcal{K}},\underline{\beta
}_{\mathcal{K}}\right)  =\left\{
\begin{array}
[c]{cc}%
\left(  R_{1},R_{2},\ldots,R_{K}\right)  :0\leq R_{\mathcal{S}}\leq
I_{t,\mathcal{S}}\left(  \underline{\alpha}_{\mathcal{K}},\underline{\beta
}_{\mathcal{K}}\right)  \text{,} & \text{for all }\mathcal{S}\subseteq
\mathcal{K}%
\end{array}
\right\}  . \label{DF_both_RB}%
\end{equation}

\end{theorem}

\begin{proof}
The rate region $\mathcal{R}_{DF}$ follows directly from Proposition
\ref{Prop_MARC_DF}, the code construction in (\ref{GMARC_DF_source_sig}%
)-(\ref{GMARC_deg_relaysig}), and the fact that $I_{r,\mathcal{S}}$ and
$I_{d,\mathcal{S}}$ are concave functions of $(\underline{\alpha}%
_{\mathcal{K}},\underline{\beta}_{\mathcal{K}})$.
\end{proof}

\begin{theorem}
The rate region $\mathcal{R}_{DF}$ is convex.
\end{theorem}

\begin{proof}
To show that $\mathcal{R}_{DF}$ is convex, it suffices to show that
$I_{r,\mathcal{S}}$ and $I_{d,\mathcal{S}}$, for all $\mathcal{S}$, are
concave functions over the convex set $\Gamma$ of $(\underline{\alpha
}_{\mathcal{K}},\underline{\beta}_{\mathcal{K}})$. This is because the
concavity of $I_{r,\mathcal{S}}$ and $I_{d,\mathcal{S}}$, for all
$\mathcal{S}$, ensures that a convex sum of two or more rate tuples in
$\mathcal{R}_{DF}$, each corresponding to a different value of $(\underline
{\alpha}_{\mathcal{K}},\underline{\beta}_{\mathcal{K}})$ tuple, also belongs
to $\mathcal{R}_{DF}$, i.e., satisfies (\ref{DF_both_RB}) for $t=r,d$.
\end{proof}

\begin{theorem}
The rate regions $\mathcal{R}_{r}$ and $\mathcal{R}_{d}$ are polymatroids.
\end{theorem}

\begin{proof}
In Appendix \ref{DG_App_PM} we show that for every choice of input
distribution satisfying (\ref{Prop_inp_dist}) the bounds in
(\ref{Prop_DF_rateregion}) are submodular set functions, and thus, enclose
regions that are polymatroids. For the Gaussian input distribution in
(\ref{GMARC_DF_source_sig}) and (\ref{GMARC_deg_relaysig}), this implies that
$\mathcal{R}_{r}\left(  \underline{\alpha}\right)  $ and $\mathcal{R}%
_{d}\left(  \underline{\alpha},\underline{\beta}\right)  $ are polymatroids
for every choice of $(\underline{\alpha},\underline{\beta})$, i.e.,
$\mathcal{R}_{r}$ and $\mathcal{R}_{d}$ are completely defined by the corner
(vertex) points on their dominant $K$-user sum-rate face \cite[Chap.
44]{cap_theorems:Schrijver01}.
\end{proof}

The region $\mathcal{R}_{DF}$ in (\ref{GMARC_R_DF}) is a union of the
intersection of the regions $\mathcal{R}_{r}$ and $\mathcal{R}_{d}$ achieved
at the relay and destination respectively, where the union is over all
$(\underline{\alpha}_{\mathcal{K}},\underline{\beta}_{\mathcal{K}})$
$\in\Gamma$. Since $\mathcal{R}_{DF}$ is convex, each point on the boundary of
$\mathcal{R}_{DF}$ is obtained by maximizing the weighted sum $%
%TCIMACRO{\tsum \nolimits_{k\in\mathcal{K}}}%
%BeginExpansion
{\textstyle\sum\nolimits_{k\in\mathcal{K}}}
%EndExpansion
\mu_{k}R_{k}$ over all $\Gamma$, and for all $\mu_{k}>0$. Specifically, we
determine the optimal policy $(\underline{\alpha}_{\mathcal{K}}^{\ast
},\underline{\beta}_{\mathcal{K}}^{\ast})$ that maximizes the sum-rate
$R_{\mathcal{K}}$ when $\mu_{k}$ $=$ $1$ for all $k$. From (\ref{GMARC_R_DF}),
we see that every point on the boundary of $\mathcal{R}_{DF}$ results from the
intersection of the polymatroids $\mathcal{R}_{r}(\underline{\alpha
}_{\mathcal{K}})$ and $\mathcal{R}_{d}(\underline{\alpha}_{\mathcal{K}%
},\underline{\beta}_{\mathcal{K}})$ for some $(\underline{\alpha}%
_{\mathcal{K}},\underline{\beta}_{\mathcal{K}})$. Since $\mathcal{R}_{r}$ and
$\mathcal{R}_{d}$ are polymatroids, as with the outer bound analysis, here too
we use Lemma \ref{Lemma_PolyIntersect} on polymatroid intersections to broadly
classify the intersection of two polymatroids into the categories of active
and inactive sets. In the following theorem we use Lemma
\ref{Lemma_PolyIntersect} to write the bound on the $K$-user DF sum-rate. We
remark that $\mathcal{R}_{r}$ and $\mathcal{R}_{d}$ are polymatroids
parametrized by $(\underline{\alpha}_{\mathcal{K}},\underline{\beta
}_{\mathcal{K}})$, and thus, Lemma \ref{Lemma_PolyIntersect} applies for each
choice of $(\underline{\alpha}_{\mathcal{K}},\underline{\beta}_{\mathcal{K}})$.

\begin{theorem}
\label{Th_GMARC_DF_1}For any $(\underline{\alpha}_{\mathcal{K}},\underline
{\beta}_{\mathcal{K}})$, the maximum $K$-user sum-rate $R_{\mathcal{K}}$
resulting from the intersecting polymatroids $\mathcal{R}_{r}$ and
$\mathcal{R}_{d}$ is
\begin{equation}
R_{\mathcal{K}}=\left\{
\begin{array}
[c]{ll}%
I_{d,\mathcal{A}}+I_{r,\mathcal{A}^{c}}, & \text{condition }2\text{ }\\
\min\left(  I_{r,\mathcal{K}},I_{d,\mathcal{K}}\right)  , & \text{otherwise}%
\end{array}
\right.  \label{DF_Rsum_cond12}%
\end{equation}
where condition $2$ is defined for a $\emptyset\not =\mathcal{A\subset K}$ as%
\begin{equation}
I_{d,\mathcal{A}}+I_{r,\mathcal{A}^{c}}<\min\left(  I_{r,\mathcal{K}%
},I_{d,\mathcal{K}}\right)  . \label{NI_Case_Cond_Sum}%
\end{equation}

\end{theorem}

\begin{remark}
The condition in (\ref{NI_Case_Cond_Sum}) determines whether the intersection
of two polymatroids belongs to either the set of active or inactive cases with
respect to the $K$-user sum-rate.
\end{remark}

\begin{proof}
The proof follows from applying Lemma \ref{Lemma_PolyIntersect} to the
maximization $R_{\mathcal{K}}=%
%TCIMACRO{\tsum \nolimits_{k\in\mathcal{K}}}%
%BeginExpansion
{\textstyle\sum\nolimits_{k\in\mathcal{K}}}
%EndExpansion
R_{k}$ for each choice of $(\underline{\alpha}_{\mathcal{K}},\underline{\beta
}_{\mathcal{K}})$.
\end{proof}

We seek to determine the maximum sum-rate $R_{\mathcal{K}}$ over all
$(\underline{\alpha}_{\mathcal{K}},\underline{\beta}_{\mathcal{K}})\in\Gamma$.
To this end, we first consider the optimization problem
\begin{equation}
R_{\mathcal{K}}=\max\limits_{(\underline{\alpha}_{\mathcal{K}},\underline
{\beta}_{\mathcal{K}})\in\Gamma}\min\left(  I_{r,\mathcal{K}}\left(
\underline{\alpha}_{\mathcal{K}}\right)  ,I_{d,\mathcal{K}}\left(
\underline{\alpha}_{\mathcal{K}},\underline{\beta}_{\mathcal{K}}\right)
\right)  . \label{RK_max_min}%
\end{equation}
We write $(\underline{\alpha}_{\mathcal{K}}^{\ast},\underline{\beta
}_{\mathcal{K}}^{\ast})$ to denote the \textit{max-min rule} optimizing
(\ref{RK_max_min}) and write $\mathcal{P}$ to denote the set of all
$(\underline{\alpha}_{\mathcal{K}}^{\ast},\underline{\beta}_{\mathcal{K}%
}^{\ast})$ maximizing (\ref{DGOB_RK_MM1}). A general solution to the max-min
optimization in (\ref{DGOB_RK_MM1}) simplifies to three cases \cite[II.C]%
{cap_theorems:HVPoor01}. The first two correspond to those in which the
maximum achieved by one of the two functions is smaller than the other, while
the third corresponds to the case in which the maximum results when the two
functions are equal (see Fig. \ref{Fig_Prop_1}). For $I_{r,\mathcal{K}}$ and
$I_{d,\mathcal{K}}$ defined in (\ref{GMARC_DF_IrG}) and (\ref{GMARC_DF_IdG}),
respectively, we can show that the solution simplifies to the consideration of
only two cases. The following theorem summarizes the solution to the max-min
problem in (\ref{RK_max_min}). The proof is developed in Appendix
\ref{DG_App4OBProof}.

\begin{theorem}
\label{Th_GMARC_DF_2}The max-min optimization
\begin{equation}
R_{\mathcal{K}}=\max\limits_{(\underline{\alpha}_{\mathcal{K}},\underline
{\beta}_{\mathcal{K}})\in\Gamma}\min\left(  I_{r,\mathcal{K}}\left(
\underline{\alpha}_{\mathcal{K}}\right)  ,I_{d,\mathcal{K}}\left(
\underline{\alpha}_{\mathcal{K}},\underline{\beta}_{\mathcal{K}}\right)
\right)
\end{equation}
simplifies to the following two cases.%
\begin{equation}%
\begin{array}
[c]{lll}%
\text{Case }1\text{:} & R_{\mathcal{K}}=C\left(  \frac{%
%TCIMACRO{\tsum \limits_{k\in\mathcal{K}}}%
%BeginExpansion
{\textstyle\sum\limits_{k\in\mathcal{K}}}
%EndExpansion
P_{k}}{N_{r}}\right)  & I_{r,\mathcal{K}}\left(  \underline{1}\right)
<I_{d,\mathcal{K}}\left(  \underline{1},\underline{0}\right) \\
\text{Case }2\text{:} & R_{\mathcal{K}}=C\left(  \left(
%TCIMACRO{\tsum \limits_{k\in\mathcal{K}}}%
%BeginExpansion
{\textstyle\sum\limits_{k\in\mathcal{K}}}
%EndExpansion
\frac{P_{k}}{N_{r}}\right)  -\frac{\left(  q^{\ast}\right)  ^{2}P_{\max}%
}{N_{r}}\right)  \equiv I^{\ast} & I_{r,\mathcal{K}}^{\ast}=I_{d,\mathcal{K}%
}^{\ast}%
\end{array}
\label{DGDF_RKmax}%
\end{equation}
where $I_{t,\mathcal{K}}^{\ast}=I_{t,\mathcal{K}}(\underline{\alpha
}_{\mathcal{K}}^{\ast},\underline{\beta}_{\mathcal{K}}^{\ast})$, $t=r,d$,
$P_{\max}=\max_{k}P_{k}$ with $\lambda_{k}=P_{k}/P_{\max}$, and
\begin{equation}
q^{\ast}\overset{\vartriangle}{=}\sum\lambda_{k}\left(  1-\alpha_{k}^{\ast
}\right)  \label{DGDF_alphastar}%
\end{equation}
is the unique value satisfying the quadratic $I_{r,\mathcal{K}}(\underline
{\alpha}_{\mathcal{K}}^{\ast},\underline{\beta}_{\mathcal{K}}^{\ast
})=I_{d,\mathcal{K}}(\underline{\alpha}_{\mathcal{K}}^{\ast},\underline{\beta
}_{\mathcal{K}}^{\ast})$ and is given by%
\begin{equation}
q^{\ast}=\frac{-K_{1}+\sqrt{K_{1}^{2}+\left(  K_{3}-K_{2}\right)  K_{0}}%
}{K_{0}} \label{DGDF_qstar}%
\end{equation}
with
\begin{equation}%
\begin{array}
[c]{ll}%
K_{0}=P_{\max}\left/  N_{r}\right.  , & K_{1}=\sqrt{P_{\max}P_{r}}\left/
N_{d}\right. \\
K_{2}=\frac{\sum\nolimits_{k\in\mathcal{K}}P_{k}}{N_{d}}+\frac{P_{r}}{N_{d}%
},\text{ and } & K_{3}=\frac{%
%TCIMACRO{\tsum \nolimits_{k\in\mathcal{K}}}%
%BeginExpansion
{\textstyle\sum\nolimits_{k\in\mathcal{K}}}
%EndExpansion
P_{k}}{N_{r}}.
\end{array}
\end{equation}
The entries of the optimal $\underline{\beta}_{\mathcal{K}}^{\ast}$ are given
by%
\begin{equation}%
\begin{array}
[c]{cc}%
\beta_{k}^{\ast}=\left\{
\begin{array}
[c]{ll}%
\frac{\left(  1-\alpha_{k}^{\ast}\right)  P_{k}}{\sum_{k=1}^{K}\left(
1-\alpha_{k}^{\ast}\right)  P_{k}} & \underline{\alpha}_{\mathcal{K}}^{\ast
}\not =\underline{1}\\
0 & \underline{\alpha}_{\mathcal{K}}^{\ast}=\underline{1}%
\end{array}
\right.  & \text{for all }k\in\mathcal{K}.
\end{array}
\label{DGDF_betastar}%
\end{equation}

\end{theorem}

\begin{remark}
\label{DGDF_RemTh2}The optimal $q^{\ast}$ in (\ref{DGDF_qstar}) is the same as
that for the optimal $x^{\ast}$ in (\ref{DGOB_xstar}). Thus, from
(\ref{DGOB_RKmax}) and (\ref{DGDF_RKmax}), we see that for both cases, the
maximum cutset bound is equal to the maximum DF bound on $R_{\mathcal{K}}$.
\end{remark}

From Lemma \ref{Lemma_PolyIntersect} we see that the maximum sum-rate can be
achieved by either an active or an inactive case. In the following theorem we
show that it suffices to consider two conditions in determining the maximum
$K$-user DF\ sum-rate. We enumerate the two conditions as%
\begin{equation}%
\begin{array}
[c]{cc}%
\text{Condition 1:} & I_{r,\mathcal{K}}(\underline{1})\leq I_{d,\mathcal{K}%
}(\underline{1},\underline{0})\\
\text{Condition 2:} & I_{r,\mathcal{K}}(\underline{1})>I_{d,\mathcal{K}%
}(\underline{1},\underline{0}).
\end{array}
\label{RK_DF_conds}%
\end{equation}
The first condition implies that the maximum sum-rate at the relay is smaller
than the corresponding rate at the destination; for this case, we show that
$I_{r,\mathcal{S}}(\underline{1})<I_{d,\mathcal{S}}(\underline{1}%
,\underline{0})$ for all $\mathcal{S}\subset\mathcal{K}$, i.e., $\mathcal{R}%
_{DF}=\mathcal{R}_{r}\subset\mathcal{R}_{d}$. Physically, this corresponds to
the case where the relay has a high SNR link to the destination and the
multiaccess link from the sources to the relay is the bottleneck link. Under
this condition, we show that the sum-capacity of a degraded Gaussian\ MARC is
achieved by DF. On the other hand, when condition 2 occurs, i.e., when
condition 1 does not hold in (\ref{RK_DF_conds}), we use the monotone
properties of $I_{r,\mathcal{K}}$ and $I_{d,\mathcal{K}}$ and Lemma
\ref{Lemma_PolyIntersect} to show that
\begin{equation}
R_{K}\leq\max\limits_{(\underline{\alpha}_{\mathcal{K}},\underline{\beta
}_{\mathcal{K}})\in\Gamma}\min\left\{  I_{r,\mathcal{K}}\left(  \underline
{\alpha}_{\mathcal{K}}\right)  ,I_{d,\mathcal{K}}\left(  \underline{\alpha
}_{\mathcal{K}},\underline{\beta}_{\mathcal{K}}\right)  \right\}
\label{RK_DF_maxmin}%
\end{equation}
with equality when the intersection of $\mathcal{R}_{r}(\underline{\alpha
}_{\mathcal{K}})$ and $\mathcal{R}_{d}(\underline{\alpha}_{\mathcal{K}%
},\underline{\beta}_{\mathcal{K}})$ results in an active case. From Theorem
\ref{Th_GMARC_DF_2}, a continuous set $\mathcal{P}$ of $(\underline{\alpha
}_{\mathcal{K}}^{\ast},\underline{\beta}_{\mathcal{K}}^{\ast})$ with a unique
$\underline{\beta}_{\mathcal{K}}^{\ast}$ for each choice of $\underline
{\alpha}_{\mathcal{K}}^{\ast}$ maximizes the right-side of (\ref{RK_DF_maxmin}%
). Furthermore, we show that the bound in (\ref{RK_DF_maxmin}) is the
sum-capacity when an active case achieves the maximum sum-rate. Finally, for
the class of symmetric degraded G-MARCs, we prove the existence of an active
case that achieves the sum-capacity.

\begin{theorem}
\label{DGDF_Th3}The $K$-user DF sum-rate $R_{\mathcal{K}}$ for a degraded
Gaussian\ MARC is
\begin{equation}%
\begin{array}
[c]{ll}%
R_{\mathcal{K}}=C\left(  \sum_{k\in\mathcal{K}}P_{k}/N_{r}\right)  , &
I_{r,\mathcal{K}}(\underline{1})<I_{d,\mathcal{K}}(\underline{1},\underline
{0})\\
R_{\mathcal{K}}\leq C\left(  \left(
%TCIMACRO{\tsum \limits_{k\in\mathcal{K}}}%
%BeginExpansion
{\textstyle\sum\limits_{k\in\mathcal{K}}}
%EndExpansion
\frac{P_{k}}{N_{r}}\right)  -\frac{\left(  q^{\ast}\right)  ^{2}P_{\max}%
}{N_{r}}\right)  , & \text{otherwise}.
\end{array}
\label{DF_RK_rate_defn}%
\end{equation}
For $I_{r,\mathcal{K}}(\underline{1})<I_{d,\mathcal{K}}(\underline
{1},\underline{0})$, DF achieves the capacity region and the sum-capacity of
the degraded Gaussian\ MARC. The upper bound on $R_{\mathcal{K}}$ in
(\ref{DF_RK_rate_defn}) is achieved with equality only for a class of
\textit{active} degraded Gaussian MARCs for which there exists a
$(\underline{\alpha}_{\mathcal{K}}^{\ast},\underline{\beta}_{\mathcal{K}%
}^{\ast})\in\mathcal{P}$ such that $\mathcal{R}_{r}(\underline{\alpha
}_{\mathcal{K}}^{\ast})\cap\mathcal{R}_{d}(\underline{\alpha}_{\mathcal{K}%
}^{\ast},\underline{\beta}_{\mathcal{K}}^{\ast})$ is an active case and is the
sum-capacity for this class. This active class also includes the class of
symmetric degraded Gaussian MARCs.
\end{theorem}

\begin{proof}
Let $\underline{\alpha}_{\mathcal{K}}=\underline{1}$ and $\underline{\beta
}_{\mathcal{K}}=\underline{0}$. From (\ref{GMARC_DF_IrG}) and
(\ref{GMARC_DF_IdG}), we see that $I_{\mathcal{S},r}$ and $I_{\mathcal{S},d}$
are monotonically increasing and decreasing functions of $\underline{\alpha
}_{\mathcal{K}}$, respectively, for a fixed $\underline{\beta}_{\mathcal{K}}$,
i.e., for any $\underline{\alpha}_{\mathcal{K}}^{(1)}$ and $\underline{\alpha
}_{\mathcal{K}}^{(2)}$ satisfying (\ref{alp_beta_bounds}), with entries
$\alpha_{k}^{(1)}\leq\alpha_{k}^{(2)}$ for all $k\in\mathcal{K}$,
$\mathcal{R}_{r}(\underline{\alpha}_{\mathcal{K}}^{(1)})\subseteq
\mathcal{R}_{r}(\underline{\alpha}_{\mathcal{K}}^{(2)})$ and $\mathcal{R}%
_{d}(\underline{\alpha}_{\mathcal{K}}^{(1)},\underline{\beta}_{\mathcal{K}%
})\supseteq\mathcal{R}_{d}(\underline{\alpha}_{\mathcal{K}}^{(2)}%
,\underline{\beta}_{\mathcal{K}})$. Thus, $\mathcal{R}_{r}(\underline{\alpha
}_{\mathcal{K}})$ achieves its largest region for $\underline{\alpha
}_{\mathcal{K}}=\underline{1}$. The bounds $I_{r,\mathcal{S}}$ and
$I_{d,\mathcal{S}}$ can be expanded for this case using (\ref{GMARC_DF_IrG})
and (\ref{GMARC_DF_IdG}), respectively, as%
\begin{align}
I_{r,\mathcal{S}}  &  =C\left(  \frac{\sum_{k\in\mathcal{S}}P_{k}}{N_{r}%
}\right) \label{DF_Case1_IrS}\\
I_{d,\mathcal{S}}  &  =C\left(  \frac{\sum_{k\in\mathcal{S}}P_{k}}{N_{d}%
}+\frac{P_{r}}{N_{d}}\right)  . \label{DF_Case1_IdS}%
\end{align}
The resulting sum-rate satisfies one of two conditions and we enumerate them below.

\textit{Condition 1}: The first condition is $I_{r,\mathcal{K}}(\underline
{1})\leq I_{d,\mathcal{K}}(\underline{1},\underline{0})$. From
(\ref{DF_Case1_IrS}) and (\ref{DF_Case1_IdS}), this case requires%
\begin{equation}
\frac{\sum\limits_{k\in\mathcal{K}}P_{k}}{N_{r}}\leq\frac{\sum\limits_{k\in
\mathcal{K}}P_{k}}{N_{d}}+\frac{P_{r}}{N_{d}}. \label{DF_Case1_CondK}%
\end{equation}
Expanding (\ref{DF_Case1_CondK}), we have, for any $\mathcal{S\subset K}$,
\begin{align}
\frac{%
%TCIMACRO{\tsum \limits_{k\in\mathcal{S}}}%
%BeginExpansion
{\textstyle\sum\limits_{k\in\mathcal{S}}}
%EndExpansion
P_{k}}{N_{r}}  &  \leq\frac{%
%TCIMACRO{\tsum \limits_{k\in\mathcal{S}}}%
%BeginExpansion
{\textstyle\sum\limits_{k\in\mathcal{S}}}
%EndExpansion
P_{k}+P_{r}}{N_{d}}-\frac{%
%TCIMACRO{\tsum \limits_{k\in\mathcal{S}^{c}}}%
%BeginExpansion
{\textstyle\sum\limits_{k\in\mathcal{S}^{c}}}
%EndExpansion
P_{k}\left(  N_{d}-N_{r}\right)  }{N_{d}N_{r}}\nonumber\\
&  <\frac{%
%TCIMACRO{\tsum \limits_{k\in\mathcal{S}}}%
%BeginExpansion
{\textstyle\sum\limits_{k\in\mathcal{S}}}
%EndExpansion
P_{k}+P_{r}}{N_{d}} \label{DF_Case2_CondS}%
\end{align}
where (\ref{DF_Case2_CondS}) follows from (\ref{DGMARC_Noise_var}). Thus,
$I_{r,\mathcal{K}}(\underline{1})\leq I_{d,\mathcal{K}}(\underline
{1},\underline{0})$ implies that $I_{r,\mathcal{S}}(\underline{1}%
)<I_{d,\mathcal{S}}(\underline{1},\underline{0})$ for all $\mathcal{S}%
\subset\mathcal{K}$, i.e., $\mathcal{R}_{r}(\underline{1})\subset
\mathcal{R}_{d}(\underline{1})$ and thus, $\mathcal{R}_{DF}(\underline
{1})=\mathcal{R}_{r}(\underline{1})$. Further, since $\mathcal{R}%
_{r}(\underline{1})\cap\mathcal{R}_{d}(\underline{1},\underline{0}%
)=\mathcal{R}_{r}(\underline{1})$, the polymatroid intersection for this
condition belongs to the intersecting set. Finally, recall that we chose
$\underline{\beta}_{\mathcal{K}}=\underline{0}$. From (\ref{GMARC_DF_IrG}), we
see that the choice of $\underline{\beta}_{\mathcal{K}}$ does not affect
$\mathcal{R}_{r}$. Further, a non-zero $\underline{\beta}_{\mathcal{K}}$ does
not increase $I_{d,\mathcal{K}}$. However, it can decrease $I_{d,\mathcal{S}}$
for some or all $\mathcal{S}\subset\mathcal{K}$ as
\begin{equation}
I_{d,\mathcal{S}}\left(  \underline{1},\underline{\beta}_{\mathcal{K}}\right)
=C\left(  \frac{\left(
%TCIMACRO{\tsum \limits_{k\in\mathcal{S}}}%
%BeginExpansion
{\textstyle\sum\limits_{k\in\mathcal{S}}}
%EndExpansion
P_{k}\right)  +P_{r}\left(  1-%
%TCIMACRO{\tsum \limits_{k\in\mathcal{S}^{c}}}%
%BeginExpansion
{\textstyle\sum\limits_{k\in\mathcal{S}^{c}}}
%EndExpansion
\beta_{k}\right)  }{N_{d}}\right)  \leq I_{d,\mathcal{S}}\left(  \underline
{1},\underline{0}\right)
\end{equation}
thereby potentially decreasing $\mathcal{R}_{DF}\left(  \underline{1}\right)
$. Thus, for the condition in (\ref{DF_Case1_CondK}) and from Theorem
\ref{DGOB_Th2}, the $K$-user sum-capacity of a degraded G-MARC for this case
is
\begin{equation}
R_{\mathcal{K}}=I_{r,\mathcal{K}}\left(  \underline{1}\right)
=B_{r,\mathcal{K}}\left(  \underline{0}\right)  =C(%
%TCIMACRO{\tsum \nolimits_{k\in\mathcal{K}}}%
%BeginExpansion
{\textstyle\sum\nolimits_{k\in\mathcal{K}}}
%EndExpansion
P_{k}\left/  N_{r}\right.  ).
\end{equation}
The max-min rule for this condition is $(\underline{\alpha}_{\mathcal{K}%
}^{\ast},\underline{\beta}_{\mathcal{K}}^{\ast})=(\underline{1},\underline
{0})$. Finally, from \textit{condition} \textit{1} in Theorem \ref{DGOB_Th3}
for a class of degraded Gaussian\ MARCs where the source and relay powers
satisfy (\ref{DF_Case1_CondK}), DF achieves the capacity region since
\begin{equation}
\mathcal{R}_{DF}=\mathcal{R}_{r}\left(  \underline{1}\right)  =\mathcal{R}%
_{r}^{ob}(\underline{0}).
\end{equation}

\textit{Condition 2}: The second condition requires $I_{\mathcal{K}%
,r}(\underline{1})>I_{\mathcal{K},d}(\underline{1},\underline{0}),$ i.e.,
\begin{equation}
\frac{\sum\limits_{k\in\mathcal{K}}P_{k}}{N_{r}}>\frac{\sum\limits_{k\in
\mathcal{K}}P_{k}}{N_{d}}+\frac{P_{r}}{N_{d}}.\label{DF_Case2_CondK}%
\end{equation}
Unlike condition 1, one cannot show here that $I_{\mathcal{S},r}%
>I_{\mathcal{S},d}$ for all $\mathcal{S}\subset\mathcal{K}$ or vice-versa.
Thus, from Theorem \ref{Th_GMARC_DF_1}, the intersection of $\mathcal{R}%
_{r}\left(  \underline{1}\right)  $ and $\mathcal{R}_{d}\left(  \underline
{1},\underline{0}\right)  $ can result in either an active or an inactive
case. From (\ref{DF_Rsum_cond12}) in Theorem \ref{Th_GMARC_DF_1}, we then
have
\begin{equation}
R_{\mathcal{K}}\leq\min\left\{  I_{r,\mathcal{K}}\left(  \underline{1}\right)
,I_{d,\mathcal{K}}\left(  \underline{1},\underline{0}\right)  \right\}
=I_{d,\mathcal{K}}\left(  \underline{1},\underline{0}\right)
\end{equation}
with equality for the active case. Note that from symmetry an active case
results for the symmetric G-MARC. However, the bound on the sum-rate, and
thus, the sum-rate too, can be increased using the fact that $I_{r,\mathcal{K}%
}$ and $I_{d,\mathcal{K}}$ are monotonically increasing and decreasing
functions of $\underline{\alpha}_{\mathcal{K}}$, respectively. In fact, from
(\ref{GMARC_DF_IrG}) and (\ref{GMARC_DF_IdG}), we see that reducing some or
all of the entries of $\underline{\alpha}_{\mathcal{K}}$ from their maximum
value of $1$ reduces $I_{r,\mathcal{K}}$ and either reduces or keeps unchanged
some or all $I_{r,\mathcal{S}}$ while increasing $I_{d,\mathcal{K}}$. Further,
since $I_{r,\mathcal{S}}(\underline{0})=0$ for all $\mathcal{S}\subseteq
\mathcal{K}$, one can shrink the region $\mathcal{R}_{r}$ just sufficiently to
ensure that there exists some $\underline{\alpha}_{\mathcal{K}}^{\ast}$ and
$\underline{\beta}_{\mathcal{K}}^{\ast}$ such that $I_{r,\mathcal{K}%
}(\underline{\alpha}_{\mathcal{K}}^{\ast})=I_{d,\mathcal{K}}(\underline
{\alpha}_{\mathcal{K}}^{\ast},\underline{\beta}_{\mathcal{K}}^{\ast})$. From
Theorem \ref{Th_GMARC_DF_2} $I_{r,\mathcal{K}}=I_{d,\mathcal{K}}$ is maximized
by a set $\mathcal{P}$ of $(\underline{\alpha}_{\mathcal{K}}^{\ast}%
,\underline{\beta}_{\mathcal{K}}^{\ast})$ where $\underline{\alpha
}_{\mathcal{K}}^{\ast}$ and $\underline{\beta}_{\mathcal{K}}^{\ast}$ satisfy
(\ref{DGDF_alphastar}) and (\ref{DGDF_betastar}), respectively. Evaluating
$I_{d,\mathcal{S}}$ at a max-min rule $(\underline{\alpha}_{\mathcal{K}}%
^{\ast},\underline{\beta}_{\mathcal{K}}^{\ast})$, we have
\begin{equation}
I_{d,\mathcal{S}}=C\left(  \frac{%
%TCIMACRO{\tsum \limits_{k\in\mathcal{S}}}%
%BeginExpansion
{\textstyle\sum\limits_{k\in\mathcal{S}}}
%EndExpansion
P_{k}}{N_{d}}+\frac{\sum\nolimits_{k\in\mathcal{S}}\left(  1-\alpha_{k}^{\ast
}\right)  P_{k}P_{r}}{N_{d}\left(  q_{\mathcal{K}}^{\ast}\right)  ^{2}}%
+2\sqrt{%
%TCIMACRO{\tsum \limits_{k\in\mathcal{S}}}%
%BeginExpansion
{\textstyle\sum\limits_{k\in\mathcal{S}}}
%EndExpansion
\frac{\left(  1-\alpha_{k}^{\ast}\right)  P_{k}P_{r}}{N_{d}^{2}}}\right)
.\label{DF_Case2_IdS}%
\end{equation}
For $\underline{\alpha}_{\mathcal{K}}^{\ast}\not =\underline{1}$, since
$I_{r,\mathcal{S}}$, for all $\mathcal{S}$, is a monotonically decreasing
function of $\underline{\alpha}_{\mathcal{K}}$ we have $\mathcal{R}_{r}\left(
\underline{\alpha}_{\mathcal{K}}^{\ast}\right)  \subset\mathcal{R}_{r}\left(
\underline{1}\right)  $. On the other hand, comparing (\ref{DF_Case1_IdS}) and
(\ref{DF_Case2_IdS}) one cannot in general show that $\mathcal{R}_{d}\left(
\underline{\alpha}_{\mathcal{K}}^{\ast},\underline{\beta}_{\mathcal{K}}^{\ast
}\right)  \supseteq\mathcal{R}_{d}\left(  \underline{1},\underline{0}\right)
$. In fact, the $\underline{\alpha}_{\mathcal{K}}^{\ast}$ chosen will
determine the relationship between $I_{d,\mathcal{S}}(\underline{\alpha
}_{\mathcal{K}}^{\ast},\underline{\beta}_{\mathcal{K}}^{\ast})$ and
$I_{d,\mathcal{S}}(\underline{1},\underline{0})$ for any $\mathcal{S}$. Thus,
for any $(\underline{\alpha}_{\mathcal{K}}^{\ast},\underline{\beta
}_{\mathcal{K}}^{\ast})$ that equalizes $I_{r,\mathcal{K}}$ and
$I_{d,\mathcal{K}}$, the polytope $\mathcal{R}_{r}\cap\mathcal{R}_{d}$ belongs
to either the set of active or inactive cases. Let $\mathcal{P}_{a}%
\subseteq\mathcal{P}$ denote the set of $(\underline{\alpha}_{\mathcal{K}%
}^{\ast},\underline{\beta}_{\mathcal{K}}^{\ast})$ that result in active cases.
From Theorem \ref{Th_GMARC_DF_1}, we can write the maximum $K$-user DF
sum-rate when $I_{r,\mathcal{K}}(\underline{1})>I_{d,\mathcal{K}}%
(\underline{1},\underline{0})$ as
\begin{equation}
R_{\mathcal{K}}=\left\{
\begin{array}
[c]{ll}%
I_{d,\mathcal{K}}(\underline{\alpha}_{\mathcal{K}}^{\ast},\underline{\beta
}_{\mathcal{K}}^{\ast})=I_{r,\mathcal{K}}(\underline{\alpha}_{\mathcal{K}%
}^{\ast})=I^{\ast}, & (\underline{\alpha}_{\mathcal{K}}^{\ast},\underline
{\beta}_{\mathcal{K}}^{\ast})\in\mathcal{P}_{a}\not =\emptyset\\
\max\limits_{(\underline{\alpha}_{\mathcal{K}}^{\ast},\underline{\beta
}_{\mathcal{K}}^{\ast})\in\mathcal{P}}I_{d,\mathcal{A}}(\underline{\alpha
}_{\mathcal{K}}^{\ast},\underline{\beta}_{\mathcal{K}}^{\ast}%
)+I_{r,\mathcal{A}^{c}}(\underline{\alpha}_{\mathcal{K}}^{\ast})<I^{\ast}, &
\mathcal{P}_{a}=\emptyset
\end{array}
\right.  \label{DF_RKmax}%
\end{equation}
where $I^{\ast}$ is given by (\ref{DGDF_RKmax}) in Theorem \ref{Th_GMARC_DF_2}%
. Finally, as shown in remark \ref{DGDF_RemTh2}, $I^{\ast}=B^{\ast}$ where
$B^{\ast}$ is the maximum outer bound sum-rate. \newline\qquad We now show
that for class of symmetric G-MARC channels, when the condition in
(\ref{DF_Case2_CondK}) holds, we achieve the $K$-user sum-capacity. For this
class, since $P_{k}=P$, from symmetry, $I_{d,\mathcal{K}}=I_{r,\mathcal{K}}$
in (\ref{GMARC_DF_IdG}) can be maximized by choosing $\alpha_{k}^{\ast}%
=\alpha^{\ast}$ for all $k$ in (\ref{DGDF_alphastar}) such that
\begin{equation}
\left(  1-\alpha^{\ast}\right)  =\left(  q^{\ast}\right)  ^{2}%
/K.\label{DGDF_alp_sym}%
\end{equation}
From (\ref{DGDF_alphastar}), since $0<\left(  q^{\ast}\right)  ^{2}<\sum
_{k=1}^{K}\lambda_{k}=K$, there exists an $0<\alpha^{\ast}<1$ that achieves
$I^{\ast}$ in (\ref{DF_RKmax}). Further, from symmetry, no subset of users
achieves a larger rate at one of the receiver than any other subset, i.e., for
$\alpha_{k}^{\ast}=\alpha^{\ast}$ and $\beta_{k}=1/K$, for all $k$,
$\mathcal{R}_{r}\cap\mathcal{R}_{d}$ belongs to the set of active cases and
the maximum $K$-user sum-rate for this class is $I^{\ast}=B^{\ast}$. Recall
that for the outer bound in Theorem \ref{DGOB_Th3}, we need to prove that
$\underline{\gamma}_{\mathcal{K}}^{\ast}\in\Gamma_{OB}$ where $\underline
{\gamma}_{\mathcal{K}}^{\ast}$ has entries $\gamma^{\ast}$ given by
(\ref{DGOB_gam_sym}) for all $k$. From (\ref{DGOB_TOB}) and
(\ref{alp_beta_bounds}), we can write
\begin{equation}%
\begin{array}
[c]{cc}%
\gamma_{k}=\left(  1-\alpha_{k}\right)  \beta_{k} & \text{where }%
(\underline{\alpha}_{\mathcal{K}},\underline{\beta}_{\mathcal{K}})\in
\Gamma\text{.}%
\end{array}
\end{equation}
. We then have
\begin{equation}
\sum_{k\in\mathcal{K}}\gamma_{k}=\sum_{k\in\mathcal{K}}\left(  1-\alpha
_{k}\right)  \beta_{k}<1\label{DGDF_gamsum}%
\end{equation}
where (\ref{DGDF_gamsum}) follows from (\ref{alp_beta_bounds}) and the fact
that $\left(  1-\alpha_{k}\right)  \beta_{k}<\beta_{k}$ for all $(\underline
{\alpha}_{\mathcal{K}},\underline{\beta}_{\mathcal{K}})\in\Gamma$. For the
symmetric case, this implies that there exists a $\gamma^{\ast}=\left(
1-\alpha^{\ast}\right)  /K$ satisfying (\ref{DGDF_gamsum}). In fact, for
$\alpha^{\ast}$ in (\ref{DGDF_alp_sym}), we obtain $\gamma^{\ast}=\left(
q^{\ast}\right)  ^{2}/K^{2}=\left(  x^{\ast}\right)  ^{2}/K^{2}<1$, i.e., the
symmetric $\gamma^{\ast}$ in (\ref{DGOB_gam_sym}) is feasible and results in
an active case. Since an active case achieves the same maximum sum-rate for
both the inner and outer bound, we see that DF\ achieves the sum-capacity for
the class of symmetric Gaussian\ MARCs. \newline\qquad For the general case of
arbitrary $P_{k}$, from (\ref{DF_RKmax}) and (\ref{DGOB_Th3RKmax}) we see that
DF achieves the maximum $K$-user sum-rate outer bounds for an \textit{active}
class of degraded Gaussian\ MARCs for which $R_{r}(\underline{\alpha
}_{\mathcal{K}}^{\ast})\cap R_{d}(\underline{\alpha}_{\mathcal{K}}^{\ast
},\underline{\beta}_{\mathcal{K}}^{\ast})$ belongs to the set of active cases.
Further, DF achieves the same maximum value for all $(\underline{\alpha
}_{\mathcal{K}}^{\ast},\underline{\beta}_{\mathcal{K}}^{\ast})\in
\mathcal{P}_{a}\not =\emptyset$. In Appendix \ref{DGA6_ActPrf}, we show that
for the same choice of the $K$ source-relay correlation coefficients for both
the inner and outer bounds, the outer cutset bounds are at least as large as
the inner DF bounds for all $\mathcal{S}\subseteq\mathcal{K}$. This implies
that for every $(\underline{\alpha}_{\mathcal{K}}^{\ast},\underline{\beta
}_{\mathcal{K}}^{\ast})\in\mathcal{P}_{a}$, there exists a $\underline{\gamma
}_{\mathcal{K}}^{\ast}$ with entries
\begin{equation}%
\begin{array}
[c]{cc}%
\gamma_{k}^{\ast}=\left(  1-\alpha_{k}^{\ast}\right)  \beta_{k}^{\ast} &
\text{for all }k
\end{array}
\label{DGDF_gammastar}%
\end{equation}
that results in an active case for the outer bounds, i.e., DF\ achieves the
sum-capacity for the active class. Note that the outer bounds may also be
maximized by other $(\underline{\alpha}_{\mathcal{K}},\underline{\beta
}_{\mathcal{K}})$ that do not maximize the $K$-user DF\ sum-rate.\newline%
\qquad Finally, as with the outer bounds, the optimization in (\ref{DF_RKmax})
for $\mathcal{P}_{a}=\emptyset$ is not straightforward. Further, comparing the
DF and cutset bounds in (\ref{DF_RKmax}) and (\ref{DGOB_Th3RKmax}),
respectively for the inactive cases, we see that the expression for the outer
bounds involves time-sharing and can in general be larger than the DF\ bound.  
\end{proof}

It is straightforward to find numerical examples for condition $1$ in Theorem
\ref{DGDF_Th3} where DF achieves the capacity region. We focus on condition
$2$ and present two examples where DF achieves the sum-capacity of a two-user
degraded Gaussian MARCs, with $\mathcal{P}_{a}=\mathcal{P}$ for one and
$\mathcal{P}_{a}\subset\mathcal{P}$ for the other. 

\begin{example}
Consider a two-user degraded Gaussian MARC with $P_{1}/N_{r}=6$, $P_{2}%
/N_{r}=4$, $P_{1}/N_{d}=3$, $P_{2}/N_{d}=2$, and $P_{r}/N_{d}=2$. These SNR
values satisfy the condition $2$ given by (\ref{DF_Case2_CondK}) in Theorem
\ref{DGDF_Th3} and thus, the DF\ sum-rate is maximized by a set of
$(\underline{\alpha}_{\mathcal{K}}^{\ast},\underline{\beta}_{\mathcal{K}%
}^{\ast})$ where $\underline{\alpha}_{\mathcal{K}}^{\ast}$ satisfies
\begin{equation}
\left(  1-\alpha_{1}^{\ast}\right)  +\frac{2}{3}\left(  1-\alpha_{2}^{\ast
}\right)  =\left(  q^{\ast}\right)  ^{2}=0.408,\label{DFEg_qstar}%
\end{equation}
and for every choice of $\underline{\alpha}_{\mathcal{K}}^{\ast}$ satisfying
(\ref{DFEg_qstar}), $\underline{\beta}_{\mathcal{K}}^{\ast}$ is given by
(\ref{DGDF_betastar}). The set of feasible $\underline{\alpha}_{\mathcal{K}%
}^{\ast}$ has entries $\alpha_{1}^{\ast}\in(0.83,1]$ with $\alpha_{2}^{\ast}$
for each such $\alpha_{1}^{\ast}$ satisfying (\ref{DFEg_qstar}) such that
$\alpha_{2}^{\ast}\in(0.75,1]$. For these SNR parameters, the set
$\mathcal{P}_{a}=\mathcal{P}$ and for each $(\underline{\alpha}_{\mathcal{K}%
}^{\ast},\underline{\beta}_{\mathcal{K}}^{\ast})\in\mathcal{P}$, the
correlation values $\gamma_{k}^{\ast}=\left(  1-\alpha_{k}^{\ast}\right)
\beta_{k}^{\ast}$, for all $k=1,2$. result in the vector $\underline{\gamma
}_{\mathcal{K}}^{\ast}\in\mathcal{G}_{a}$. 
\end{example}

\begin{example}
We next consider a two-user example with $P_{1}/N_{r}=6$, $P_{2}/N_{r}=0.4$,
$P_{1}/N_{d}=3$, $P_{2}/N_{d}=0.2$, and $P_{r}/N_{d}=2$. These SNR values also
satisfy the condition $2$ given by (\ref{DF_Case2_CondK}) in Theorem
\ref{DGDF_Th3} and thus, the DF\ sum-rate is maximized by a set of
$(\underline{\alpha}_{\mathcal{K}}^{\ast},\underline{\beta}_{\mathcal{K}%
}^{\ast})$ where $\underline{\alpha}_{\mathcal{K}}^{\ast}$ satisfies
\begin{equation}
\left(  1-\alpha_{1}^{\ast}\right)  +\frac{2}{3}\left(  1-\alpha_{2}^{\ast
}\right)  =\left(  q^{\ast}\right)  ^{2}=0.197.\label{DFE2_qstar}%
\end{equation}
The set of feasible $\underline{\alpha}_{\mathcal{K}}^{\ast}$ has entries
$\alpha_{1}^{\ast}\in(0.96,1]$ with $\alpha_{2}^{\ast}$ for each such
$\alpha_{1}^{\ast}$ satisfying (\ref{DFE2_qstar}) such that $\alpha_{2}^{\ast
}\in(0.416,1]$. Note that subject to (\ref{DFE2_qstar}), $\alpha_{2}$
decreases as $\alpha_{1}$ increases and vice-versa. For these SNR parameters,
the set $\mathcal{P}_{a}$ consists of $(\underline{\alpha}_{\mathcal{K}}%
^{\ast},\underline{\beta}_{\mathcal{K}}^{\ast})$ where the entries $\alpha
_{1}^{\ast}$ and $\alpha_{2}^{\ast}$ are restricted to the sets
$(0.961,0.979]$ and $(0.731,1]$, respectively. The remaining values for
$\alpha_{1}^{\ast}$ and $\alpha_{2}^{\ast}$ satisfying (\ref{DFE2_qstar})
result in a polymatroid intersection that belongs to the set of inactive
cases. In fact, all such values result in the inactive case $2$ illustrated in
Fig. \ref{Fig_AllCases} for $K=2.$ 
\end{example}

Finally, for the two-user degraded Gaussian MARC, a numerical example
illustrating $\mathcal{P}_{a}=\emptyset$ does not appear straightforward
despite using a wide range of ratios of $P_{1}$ to $P_{2}$, i.e., not all
rate-maximizing intersections are such that one of the sources achieve better
rates at one of the receivers while the other source achieves a better rate at
the other receiver. A possible reason for this is because, at any receiver,
the noise seen by both sources is the same, and thus, the source with smaller
power typically achieves smaller rates at both receivers. It may be possible
to increase the rate achieved at the destination by increasing the relay
power; however, large values of relay power will result in the bottle-neck
case where condition 1 in Theorem \ref{DGDF_Th3} holds. Thus, it appears that
it may always be possible to find an active case, particularly, one that
maximizes the sum-rate. If this is true for any arbitrary $K$, then DF
achieves the sum-capacity of the degraded Gaussian\ MARC.

\begin{remark}
In the above analysis, we determined the sum-capacity for a degraded
Gaussian\ MARC under a per symbol transmit power constraint at the sources and
relay. One can also consider an average power constraint at every transmitter.
The achievable strategy remains unchanged; for the converse we start with the
convex sums of the outer bounds in (\ref{MARC_OB_cutset}) over $n$ channel
uses. In the $i^{th}$ channel use, the bounds at the relay and destination are
given by $B_{r,\mathcal{S}}$ and $B_{d,\mathcal{S}}$ in (\ref{Con_final_B1S})
and (\ref{Con_final_B2S}), respectively, for all $\mathcal{S}$, except now the
correlation parameters and power parameters are indexed by $i$. Recall that
$B_{d,\mathcal{S}}$ is a concave function of the correlation coefficients and
power. On the other hand, $B_{r,\mathcal{S}}$ for all $\mathcal{S}%
\subset\mathcal{K}$ is not a concave function of the power and
cross-correlation parameters. However, we can use the concavity of
$B_{r,\mathcal{K}}$ to show that the maximum bounds on the sum-rate in Thereom
\ref{DGOB_Th1} remain unchanged. This in conjunction with the steps in Theorem
\ref{DGOB_Th3}, lead to the same sum-capacity results. Finally, we note that
as with the symbol power constraint, here too we require time-sharing to
develop the outer bound rate region.
\end{remark}

\section{\label{DG_Sec6}Concluding Remarks}

In this paper, we have studied the sum-capacity of degraded Gaussian MARCs. In
particular, we have developed the rate regions for the achievable strategy of
DF and the cutset outer bounds. The outer bounds have been obtained using
cut-set bounds for the case of independent sources and have been shown to be
maximized by Gaussian signaling at the sources and relay.

We have also shown that, in general, the rate regions achieved by the inner
and outer bounds are not the same. This difference is due to the fact that the
input distributions and the rate expressions for the inner and outer bounds
are not exactly the same. In fact, the input distribution for the inner bound
uses auxiliary random variables to model the correlation between the inputs at
the sources and the relay and is more restrictive than the distribution for
the outer bound. Despite these differences, in both cases the input
distributions can be quantified by a set of $K$ source-relay cross-correlation
coefficients. Further, in both cases, we have shown that the rate region for
every choice of the appropriate input distribution is an intersection of
polymatroids. We have used the properties of polymatroid intersections to show
that for both the inner and outer bounds the largest $K$-user sum-rate is at
most the maximum of the minimum of the two $K$-user sum rate bounds, with
equality only when the polymatroid intersections belongs to the set of active
cases in which the $K$-user sum rate planes are active.

For both DF and the outer bounds, we have shown that the largest $K$-user
sum-rate can be determined using max-min optimization techniques. In fact, we
have shown that for both the inner and outer bounds the max-min optimization
problem results in one of two unique solutions. The first solution results
when the largest sum-rate from the $K$ sources to the relay is the bottle-neck
rate and for this case, we have shown that DF achieves the capacity region. We
have further shown that the sum-rate maximizing polymatroid intersection for
this case belongs to the set of active cases. Specifically, the sum-capacity
as well as the entire capacity region is achieved by a max-min rule where the
sources and the relay do not allocate any power to cooperatively achieving
coherent combining gains at the destination, i.e., the auxiliary random
variables $V_{k}=0$, for all $k$. Thus, under Gaussian signaling, the capacity
region is achieved by DF\ because the inner and outer bounds at the relay, for
$V_{\mathcal{K}}=0$, are $I(X_{\mathcal{S}};Y_{r}|X_{r}X_{\mathcal{S}^{c}%
})=I(X_{\mathcal{S}};Y_{r}|X_{r}X_{\mathcal{S}^{c}}V_{\mathcal{K}})$ for all
$\mathcal{S}\subseteq\mathcal{K}$ (see (\ref{MARC_OB_cutset}) and
(\ref{Prop_DF_rateregion})).

The second solution results when the largest sum-rates at the relay and the
destination are equal. For this case, we have shown that DF achieves the
sum-capacity for a class of active degraded Gaussian MARCs in which the
sum-rate maximizing polymatroid intersection belongs to the set of active
cases. We have also shown that this class of active degraded Gaussian MARCs
contains the class of symmetric Gaussian\ MARCs. In general, for this class,
we have shown that the max-min DF rule is such that $V_{k}\not =0$ for all
$k$, i.e., a non-empty subset of sources and the relay divide their transmit
power to achieve cooperative combining gains at the destination. We have also
shown that the largest DF sum-rate is achieved by a relay power policy that
maximizes the cooperative gains achieved at the destination, i.e., $X_{r}$ is
a unique weighted sum of $V_{k}$ for all $k$ where the weight for each source
$k$ is proportional to the power allocated by source $k$ to cooperating with
the relay. Our analysis has also shown that the maximum sum-rate admits
several solutions for the power fractions allocated at the sources for
cooperation subject to a constraint that results from the equating the two
bounds on the sum-rate. For the outer bounds, we have shown that the $K$-user
sum-rate outer bound is maximized by a set of cross-correlation coefficients
that are subject to the same constraint as DF and the maximum sum-rate is the
same as that for DF. Furthermore, for the class of active degraded
Gaussian\ MARCs, we have shown that the set of DF max-min rules $(\underline
{\alpha}_{\mathcal{K}}^{\ast},\underline{\beta}_{\mathcal{K}}^{\ast})$ also
maximizes the outer bounds by using the fact that the inner and outer bound
coefficients can be related as $\gamma_{k}=\left(  1-\alpha_{k}\right)
\beta_{k}$, for all $k$. Finally, since a DF max-min rule requires a unique
correlation between $X_{r}$ and $V_{\mathcal{K}}$, conditioning the outer
bound that uses $Y_{r}$ on $X_{r}$ alone suffices to obtain the sum-capacity.

\section{Acknowledgments}

L. Sankar is grateful for numerous detailed discussions on the MARC with
Gerhard Kramer of Bell Labs, Alcatel-Lucent and on polymatroid intersections
with Jan Vondrack of Princeton University.

\pagebreak%

%TCIMACRO{\TeXButton{appendices}{\appendices}}%
%BeginExpansion
\appendices
%EndExpansion
{}

\section{\label{DG_App0_OBProof}Outer Bounds: Proof}

We now develop the proof for Theorem \ref{DGOB_Th0}. Recall that we write
$B_{r,\mathcal{S}}$ and $B_{d,\mathcal{S}}$ to denote, respectively, the first
and second bound on $R_{\mathcal{S}}$ in (\ref{DGMARC_OB_1}) for a constant
$U$. Expanding the bounds on $R_{\mathcal{S}}$ in (\ref{DGMARC_OB_1}) for a
constant $U$, we have%
\begin{equation}
R_{\mathcal{S}}\leq\min\left\{  h(Y_{r}|X_{r}X_{\mathcal{S}^{c}}%
)-h(Z_{r}),h(Y_{d}|X_{\mathcal{S}^{c}})-h(Z_{d})\right\}  .
\end{equation}
For a fixed covariance matrix of the input random variables $X_{\mathcal{K}}$
and $X_{r}$, one can apply a conditional entropy maximization theorem
\cite[Lemma 1]{cap_theorems:JAT01} to show that $h(Y_{r}|X_{r}X_{\mathcal{S}%
^{c}})$ and $h(Y_{d}|X_{\mathcal{S}^{c}})$ are maximized by choosing the
distribution in (\ref{GMARC_converse_inpdist}) as jointly Gaussian. Consider
the bound $B_{r,\mathcal{S}}$. Expanding $Y_{r}$, we have%
\begin{equation}
R_{\mathcal{S}}\leq C\left(  \frac{E\left[  var\left(
%TCIMACRO{\tsum \nolimits_{k\in\mathcal{S}}}%
%BeginExpansion
{\textstyle\sum\nolimits_{k\in\mathcal{S}}}
%EndExpansion
X_{k}|X_{r}X_{\mathcal{S}^{c}}\right)  \right]  }{N_{r}}\right)
.\label{OB1_simple}%
\end{equation}
For Gaussian signals, using the chain rule, we have%
\begin{equation}
E\left[  var\left(
%TCIMACRO{\tsum \limits_{k\in\mathcal{S}}}%
%BeginExpansion
{\textstyle\sum\limits_{k\in\mathcal{S}}}
%EndExpansion
X_{k}|X_{r}X_{\mathcal{S}^{c}}\right)  \right]  =\frac{\det(K_{\underline
{A}|\underline{C}})}{\det(K_{\underline{B}|\underline{C}})}%
\label{DGMARC_Evar_defn}%
\end{equation}
where
\begin{align}
\underline{A} &  =\left[
\begin{array}
[c]{cc}%
%TCIMACRO{\tsum \nolimits_{k\in\mathcal{S}}}%
%BeginExpansion
{\textstyle\sum\nolimits_{k\in\mathcal{S}}}
%EndExpansion
X_{k} & X_{r}%
\end{array}
\right]  ^{T}\\
\underline{B} &  =\left[  X_{r}\right]  \\
\underline{C} &  =\left[  X_{\mathcal{S}^{c}}\right]
\end{align}
and for random vectors $\underline{X}$ and $\underline{Y}$, the conditional
covariance $K_{\underline{X}|\underline{Y}}$ is%
\begin{equation}
K_{\underline{X}|\underline{Y}}=E\left[  \left(  \underline{X}-E\left[
\underline{X}|\underline{Y}\right]  \right)  \left(  \underline{X}-E\left[
\underline{X}|\underline{Y}\right]  \right)  ^{T}\right]
\label{Cond_exp_defn}%
\end{equation}
where $\underline{X}^{T}$ is the transpose of $\underline{X}$. We use the fact
that $X_{\mathcal{S}}$ and $X_{\mathcal{S}^{c}}$ are independent to expand
(\ref{DGMARC_Evar_defn}) as%
\begin{equation}
E\left[  var\left(
%TCIMACRO{\tsum \limits_{k\in\mathcal{S}}}%
%BeginExpansion
{\textstyle\sum\limits_{k\in\mathcal{S}}}
%EndExpansion
X_{k}|X_{r}X_{\mathcal{S}^{c}}\right)  \right]  =var\left(
%TCIMACRO{\tsum \limits_{k\in\mathcal{S}}}%
%BeginExpansion
{\textstyle\sum\limits_{k\in\mathcal{S}}}
%EndExpansion
X_{k}\right)  -\frac{E^{2}\left[
%TCIMACRO{\tsum \limits_{k\in\mathcal{S}}}%
%BeginExpansion
{\textstyle\sum\limits_{k\in\mathcal{S}}}
%EndExpansion
X_{k}\tilde{X}_{r,\mathcal{S}}\right]  }{P_{r,\mathcal{S}}}\label{Cond_varXk}%
\end{equation}
where $\tilde{X}_{r,\mathcal{S}}\overset{\vartriangle}{=}\left(  X_{r}\text{
}-\text{ }E(X_{r}\text{ }|\text{ }X_{\mathcal{S}^{c}})\right)  $ is a Gaussian
random variable with variance
\begin{equation}
P_{r,\mathcal{S}}=E\left[  \tilde{X}_{r,\mathcal{S}}^{2}\right]  =E\left[
var(X_{r}|X_{\mathcal{S}^{c}})\right]  .
\end{equation}
Substituting (\ref{Cond_varXk}) in (\ref{OB1_simple}) and using
(\ref{Pwr_cond}) to bound $var\left(  X_{k}\right)  $ for all $k$, we obtain,%
\begin{align}
R_{\mathcal{S}} &  \leq C\left(  \frac{%
%TCIMACRO{\tsum \limits_{k\in\mathcal{S}}}%
%BeginExpansion
{\textstyle\sum\limits_{k\in\mathcal{S}}}
%EndExpansion
var\left(  X_{k}\right)  -\frac{1}{P_{r,\mathcal{S}}}E^{2}\left[
%TCIMACRO{\tsum \limits_{k\in\mathcal{S}}}%
%BeginExpansion
{\textstyle\sum\limits_{k\in\mathcal{S}}}
%EndExpansion
X_{k}\tilde{X}_{r,\mathcal{S}}\right]  }{N_{r}}\right)  \label{Con_B1_1}\\
&  \leq C\left(  \frac{\left(
%TCIMACRO{\tsum \limits_{k\in\mathcal{S}}}%
%BeginExpansion
{\textstyle\sum\limits_{k\in\mathcal{S}}}
%EndExpansion
P_{k}\right)  -\frac{1}{P_{r,\mathcal{S}}}E^{2}\left[
%TCIMACRO{\tsum \limits_{k\in\mathcal{S}}}%
%BeginExpansion
{\textstyle\sum\limits_{k\in\mathcal{S}}}
%EndExpansion
X_{k}\tilde{X}_{r,\mathcal{S}}\right]  }{N_{r}}\right)  .
\end{align}
We define $\gamma_{k}$, for all $k\in\mathcal{K}$, by%
\begin{equation}
E\left[  X_{k}X_{r}\right]  \overset{\vartriangle}{=}\sqrt{\gamma_{k}%
P_{k}P_{r}}.\label{OB_gammak_def}%
\end{equation}
\qquad Note that by definition,
\begin{equation}%
\begin{array}
[c]{cc}%
\gamma_{k}\in\lbrack0,1] & \text{for all }k\in\mathcal{K}%
\end{array}
\label{OB_gammak_bounds}%
\end{equation}
and%
\begin{equation}
\sum_{k=1}^{K}\gamma_{k}\leq1.\label{OB_gammaK_sum}%
\end{equation}
Using the independence of $X_{k}$ for all $k\,$ and (\ref{OB_gammak_def}), we
write
\begin{equation}
E\left[
%TCIMACRO{\tsum \limits_{k\in\mathcal{S}}}%
%BeginExpansion
{\textstyle\sum\limits_{k\in\mathcal{S}}}
%EndExpansion
X_{k}\tilde{X}_{r}\right]  =%
%TCIMACRO{\tsum \limits_{k\in\mathcal{S}}}%
%BeginExpansion
{\textstyle\sum\limits_{k\in\mathcal{S}}}
%EndExpansion
E\left[  X_{k}X_{r}\right]  =%
%TCIMACRO{\tsum \limits_{k\in\mathcal{S}}}%
%BeginExpansion
{\textstyle\sum\limits_{k\in\mathcal{S}}}
%EndExpansion
\sqrt{\gamma_{k}P_{k}P_{r}}.\label{OB1_EXkXr}%
\end{equation}
Next we use (\ref{OB_gammak_def}) to evaluate $P_{r,\mathcal{S}}$. We start by
considering the random variable%
\begin{equation}
\hat{X}_{r}=X_{r}-E\left[  X_{r}|X_{K}\right]  .
\end{equation}
Using (\ref{OB_gammak_def}) and the independence of $X_{k}$ for all $k$, we
can write the variance of $\hat{X}_{r}$ as
\begin{align}
E\left[  \hat{X}_{r}^{2}\right]   &  =E\left[  var\left(  X_{r}|X_{K}\right)
\right]  \label{EvarXr_Vk}\\
&  =\left(  1-\gamma_{K}\right)  P_{r}.\label{EvarXr_Vk_final}%
\end{align}
where we used (\ref{Cond_exp_defn}) to simplify (\ref{EvarXr_Vk}) to
(\ref{EvarXr_Vk_final}). Continuing thus, we consider the random variable
$\bar{X}_{r}=\hat{X}_{r}-E\left[  \hat{X}_{r}|X_{K-1}\right]  $. Using the
independence of $X_{k}$ for all $k$, we thus have%
\begin{align}
E\left[  \bar{X}_{r}^{2}\right]   &  =E\left[  \hat{X}_{r}^{2}\right]
-E\left[  E^{2}\left[  \hat{X}_{r}|X_{K-1}\right]  \right]  \\
&  =E\left[  var\left(  X_{r}|X_{K-1}X_{K}\right)  \right]  \\
&  =\left(  1-\gamma_{K-1}-\gamma_{K}\right)  P_{r}.
\end{align}
Generalizing the above, we have
\begin{equation}%
\begin{array}
[c]{cc}%
E\left[  var\left(  X_{r}|X_{\mathcal{S}^{c}}\right)  \right]  =\left(  1-%
%TCIMACRO{\tsum \limits_{k\in\mathcal{S}^{c}}}%
%BeginExpansion
{\textstyle\sum\limits_{k\in\mathcal{S}^{c}}}
%EndExpansion
\gamma_{k}\right)  P_{r}\overset{\vartriangle}{=}\overline{\gamma
}_{\mathcal{S}^{c}}P_{r} & \text{for all }\mathcal{S}\subseteq\mathcal{K}%
\text{.}%
\end{array}
\label{Con_varXr_condVs}%
\end{equation}
Finally, we substitute (\ref{Con_varXr_condVs}) and (\ref{OB1_EXkXr}) in
(\ref{Con_B1_1}) to simplify the first bound as%
\begin{equation}
R_{\mathcal{S}}\leq\left\{
\begin{array}
[c]{ll}%
C\left(
%TCIMACRO{\tsum \limits_{k\in\mathcal{S}}}%
%BeginExpansion
{\textstyle\sum\limits_{k\in\mathcal{S}}}
%EndExpansion
\frac{P_{k}}{N_{r}}\right)  , & \text{if }%
%TCIMACRO{\tsum \limits_{k\in\mathcal{S}^{c}}}%
%BeginExpansion
{\textstyle\sum\limits_{k\in\mathcal{S}^{c}}}
%EndExpansion
\gamma_{k}=1\\
C\left(
%TCIMACRO{\tsum \limits_{k\in\mathcal{S}}}%
%BeginExpansion
{\textstyle\sum\limits_{k\in\mathcal{S}}}
%EndExpansion
\frac{P_{k}}{N_{r}}-\frac{\left(  \sum\limits_{k\in\mathcal{S}}\sqrt
{\gamma_{k}P_{k}}\right)  ^{2}}{N_{r}\overline{\gamma}_{\mathcal{S}^{c}}%
}\right)  , & \text{otherwise.}%
\end{array}
\right.  \label{DGA1_BrS}%
\end{equation}
Observe that for $K$ $=$ $1$, we have $V_{1}$ $=$ $X_{r}$ and $\gamma_{1}$ $=$
$1$, and thus, (\ref{Con_final_B1S}) simplifies to the first outer bound in
\cite[theorem 5]{cap_theorems:CEG01} for the classic single source degraded
relay channel. Finally, from (\ref{Con_varXr_condVs}), observe that
$\gamma_{k}$, for all $k$, satisfies%
\begin{equation}%
%TCIMACRO{\tsum \limits_{k\in\mathcal{K}}}%
%BeginExpansion
{\textstyle\sum\limits_{k\in\mathcal{K}}}
%EndExpansion
\gamma_{k}\leq1\text{.}\label{OB_gammak_constraint}%
\end{equation}
Consider the bound $B_{d,\mathcal{S}}$ in (\ref{DGMARC_OB_1}) with $U$ a
constant. Expanding $Y_{d}$ using (\ref{Yd_defn}), we have%
\begin{align}
R_{\mathcal{S}} &  \leq C\left(  \left.  E\left[  var\left(
%TCIMACRO{\tsum \limits_{k\in\mathcal{S}}}%
%BeginExpansion
{\textstyle\sum\limits_{k\in\mathcal{S}}}
%EndExpansion
X_{k}+X_{r}|X_{\mathcal{S}^{c}}\right)  \right]  \right/  N_{d}\right)  \\
&  =C\left(  \frac{%
%TCIMACRO{\tsum \limits_{k\in\mathcal{S}}}%
%BeginExpansion
{\textstyle\sum\limits_{k\in\mathcal{S}}}
%EndExpansion
\left(  P_{k}+2E\left(  X_{k}\tilde{X}_{r,\mathcal{S}}\right)  \right)
+E\left[  var(X_{r}|X_{\mathcal{S}^{c}})\right]  }{N_{d}}\right)
.\label{Con_B2_fin}%
\end{align}
Using (\ref{Pwr_cond}), (\ref{Con_varXr_condVs},) and (\ref{OB1_EXkXr}), we
simplify (\ref{Con_B2_fin}) as%
\begin{equation}
R_{\mathcal{S}}\leq C\left(  \frac{%
%TCIMACRO{\tsum \limits_{k\in\mathcal{S}}}%
%BeginExpansion
{\textstyle\sum\limits_{k\in\mathcal{S}}}
%EndExpansion
P_{k}+\overline{\gamma}_{\mathcal{S}^{c}}P_{r}+2%
%TCIMACRO{\tsum \limits_{k\in\mathcal{S}}}%
%BeginExpansion
{\textstyle\sum\limits_{k\in\mathcal{S}}}
%EndExpansion
\sqrt{\gamma_{k}P_{k}P_{r}}}{N_{d}}\right)  .\label{DGA1_BdS}%
\end{equation}
Writing $B_{r,\mathcal{S}}$ and $B_{d,\mathcal{S}}$ to denote the bounds on
the right-side of (\ref{DGA1_BrS}) and (\ref{DGA1_BdS}), respectively, we have
for a constant $U$,%
\begin{equation}%
\begin{array}
[c]{cc}%
R_{\mathcal{S}}\leq\min\left(  B_{r,\mathcal{S}},B_{d,\mathcal{S}}\right)   &
\text{for all }\mathcal{S}\subseteq\mathcal{K}.
\end{array}
\label{Con_finalbounds}%
\end{equation}

\section{\label{DG_App_PM}Inner and Outer Bounds:\ Polymatroids}

We first prove that the rate regions $\mathcal{R}_{r}^{ob}$ and $\mathcal{R}%
_{d}^{ob}$ given by the cutset bounds are polymatroids. Using similar
techniques, we then show that the DF regions $\mathcal{R}_{r}$ and
$\mathcal{R}_{d}$ are polymatroids.

\subsection{Outer Bounds}

Consider the set functions (see \ref{Prop_DF_rateregion})
\begin{equation}
f_{1}\left(  \mathcal{S}\right)  =\left\{
\begin{array}
[c]{ll}%
I\left(  X_{\mathcal{S}}X_{r};Y_{d}|X_{S^{c}}U\right)   & \mathcal{S}%
\subseteq\mathcal{K},\mathcal{S\not =\emptyset}\\
0 & \mathcal{S}=\emptyset
\end{array}
\right.
\end{equation}
and
\begin{equation}
f_{2}\left(  \mathcal{S}\right)  =\left\{
\begin{array}
[c]{ll}%
I\left(  X_{\mathcal{S}};Y_{r}|X_{S^{c}}X_{r}U\right)   & \mathcal{S}%
\subseteq\mathcal{K},\mathcal{S}\not =\emptyset\\
0 & \mathcal{S}=\emptyset
\end{array}
\right.
\end{equation}
for some distribution satisfying (\ref{GMARC_converse_inpdist}). We claim that
$f_{1}$ and $f_{2}$ are submodular \cite[Ch.~44]{cap_theorems:Schrijver01}. To
see this, we first consider $f_{1}$ and $k_{1}$, $k_{2}$ in $\mathcal{K}$ with
$k_{1}\neq k_{2}$, $k_{1}\notin\mathcal{S}$, $k_{2}\notin\mathcal{S}$, and
expand%
\begin{align}
f_{1}(\mathcal{S}\cup\{k_{1}\})+f_{1}(\mathcal{S}\cup\{k_{2}\}) &
=I(X_{\mathcal{S}}X_{k_{1}}X_{r};Y_{d}|X_{(\mathcal{S}\cup\{k_{1}\})^{C}%
}U)+I(X_{\mathcal{S}}X_{k_{2}}X_{r};Y_{d}|X_{(\mathcal{S}\cup\{k_{2}\})^{C}%
}U)\\
&  =I(X_{k_{1}};Y_{d}|X_{(\mathcal{S}\cup\{k_{1}\})^{C}}U)+I(X_{\mathcal{S}%
}X_{r};Y_{d}|X_{\mathcal{S}^{C}}U)\label{DGAppPM_f1exp}\\
&  \text{ \ \ }+I(X_{\mathcal{S}}X_{k_{2}}X_{r};Y_{d}|X_{(\mathcal{S}%
\cup\{k_{2}\})^{C}}U)
\end{align}
where (\ref{DGAppPM_f1exp}) follows from the chain rule for mutual
information. We lower bound the first term in (\ref{DGAppPM_f1exp}) as%
\begin{align}
&  h(X_{k_{1}}|X_{(\mathcal{S}\cup\{k_{1}\})^{C}}U)-h(X_{k_{1}}%
|X_{(\mathcal{S}\cup\{k_{1}\})^{C}}Y_{d}U)\label{DGAppPM_f1_4}\\
&  =h(X_{k_{1}}|X_{(\mathcal{S}\cup\{k_{1},k_{2}\})^{C}}U)-h(X_{k_{1}%
}|X_{(\mathcal{S}\cup\{k_{1}\})^{C}}Y_{d}U)\\
&  \geq I(X_{k_{1}};Y_{d}|X_{(\mathcal{S}\cup\{k_{1},k_{2}\})^{C}%
}U)\label{DGAppPM_f1_5}%
\end{align}
where (\ref{DGAppPM_f1_4}) follows from the Markov chain $X_{k}-U-X_{j}$ for
all $k,j\in\mathcal{K}$, $k\not =j$ and (\ref{DGAppPM_f1_5}) because
conditioning cannot increase entropy. The expression (\ref{DGAppPM_f1_5})
added to the final term in (\ref{DGAppPM_f1exp}) is
\[
I(X_{\mathcal{S}\cup\{k_{1},k_{2}\}}X_{r};Y_{d}|X_{(\mathcal{S}\cup
\{k_{1},k_{2}\})^{C}}U).
\]
Inserting (\ref{DGAppPM_f1_5}) into (\ref{DGAppPM_f1exp}), we have
\[
f_{1}(\mathcal{S}\cup\{k_{1}\})+f_{1}(\mathcal{S}\cup\{k_{2}\})\geq
f_{1}(\mathcal{S})+f_{1}(\mathcal{S}\cup\{k_{1},k_{2}\})
\]
for all $\mathcal{S}\subseteq\mathcal{K}$. The set function $f_{1}(\cdot)$ is
therefore submodular by \cite[Theorem~44.1, p.~767]{cap_theorems:Schrijver01}.

The above steps show that the rate region $\mathcal{R}_{d}^{ob}$ defined by
the destination cutset bounds (see (\ref{MARC_OB_cutset}))%
\begin{equation}
R_{\mathcal{S}}\leq I(X_{\mathcal{S}}X_{r};Y_{d}|X_{\mathcal{S}^{c}}%
U),\quad\mathcal{S}\subseteq\mathcal{K}%
\end{equation}
is a polymatroid associated with $f_{1}(\cdot)$ (see \cite[p.~767]%
{cap_theorems:Schrijver01}). \newline One can similarly show that $f_{2}%
(\cdot)$ is submodular. To see this, consider $f_{2}$ and $k_{1}$, $k_{2}$ in
$\mathcal{K}$ with $k_{1}\neq k_{2}$, $k_{1}\notin\mathcal{S}$, $k_{2}%
\notin\mathcal{S}$, and expand
\begin{align}
&  f_{2}(\mathcal{S}\cup\{k_{1}\})+f_{2}(\mathcal{S}\cup\{k_{2}\})\nonumber\\
&  =I(X_{\mathcal{S}}X_{k_{1}};Y_{r}|X_{(\mathcal{S}\cup\{k_{1}\})^{C}}%
X_{r}U)+I(X_{\mathcal{S}}X_{k_{2}};Y_{r}|X_{(\mathcal{S}\cup\{k_{2}\})^{C}%
}X_{r}U)\\
&  =I(X_{k_{1}};Y_{r}|X_{(\mathcal{S}\cup\{k_{1}\})^{C}}X_{r}%
U)+I(X_{\mathcal{S}};Y_{r}|X_{\mathcal{S}^{C}}X_{r}U)+I(X_{\mathcal{S}%
}X_{k_{2}};Y_{r}|X_{(\mathcal{S}\cup\{k_{2}\})^{C}}X_{r}U)
\label{DGAppPM_f2exp}%
\end{align}
where (\ref{DGAppPM_f2exp}) follows from the chain rule for mutual
information. We lower bound the first term in (\ref{DGAppPM_f2exp}) as
\begin{align}
&  \quad h(X_{k_{1}}|X_{(\mathcal{S}\cup\{k_{1}\})^{C}}X_{r}U)-h(X_{k_{1}%
}|X_{(\mathcal{S}\cup\{k_{1}\})^{C}}Y_{r}X_{r}U)\nonumber\\
&  =h(X_{k_{1}}|X_{(\mathcal{S}\cup\{k_{1},k_{2}\})^{C}}X_{r}U)-h(X_{k_{1}%
}|X_{(\mathcal{S}\cup\{k_{1}\})^{C}}Y_{r}X_{r}U)\label{DGAppPM_f2_4}\\
&  \geq I(X_{k_{1}};Y_{r}|X_{(\mathcal{S}\cup\{k_{1},k_{2}\})^{C}}X_{r}U)
\label{DGAppPM_f2_5}%
\end{align}
where (\ref{DGAppPM_f2_4}) follows from the independence of $X_{k}$ and
(\ref{DGAppPM_f2_5}) because conditioning cannot increase entropy. The
expression (\ref{DGAppPM_f2_5}) added to the final term in
(\ref{DGAppPM_f2exp}) is
\[
I(X_{\mathcal{S}\cup\{k_{1},k_{2}\}};Y_{r}|X_{(\mathcal{S}\cup\{k_{1}%
,k_{2}\})^{C}}X_{r}U).
\]
Inserting (\ref{DGAppPM_f1_5}) into (\ref{DGAppPM_f1exp}), we have
\[
f_{2}(\mathcal{S}\cup\{k_{1}\})+f_{2}(\mathcal{S}\cup\{k_{2}\})\geq
f_{2}(\mathcal{S})+f_{2}(\mathcal{S}\cup\{k_{1},k_{2}\})
\]
for all $\mathcal{S}\subseteq\mathcal{K}$. The set function $f_{2}(\cdot)$ is
therefore submodular by \cite[Theorem~44.1, p.~767]{cap_theorems:Schrijver01}.
This in turn implies that the rate region $\mathcal{R}_{r}^{ob}$ defined by
the relay cutset bounds (see (\ref{MARC_OB_cutset}))%
\begin{equation}
R_{\mathcal{S}}\leq I(X_{\mathcal{S}};Y_{r}|X_{\mathcal{S}^{c}}X_{r}%
U),\quad\mathcal{S}\subseteq\mathcal{K}%
\end{equation}
is a polymatroid associated with $f_{2}(\cdot)$ (see \cite[p.~767]%
{cap_theorems:Schrijver01}). \newline

\subsection{Inner Bounds}

For the inner DF\ bounds, we consider the set functions (see
\ref{Prop_DF_rateregion})
\begin{equation}
f_{3}\left(  \mathcal{S}\right)  =\left\{
\begin{array}
[c]{ll}%
I\left(  X_{\mathcal{S}}X_{r};Y_{d}|X_{S^{c}}V_{\mathcal{S}^{c}}U\right)  &
\mathcal{S}\subseteq\mathcal{K},\mathcal{S\not =\emptyset}\\
0 & \mathcal{S}=\emptyset
\end{array}
\right.
\end{equation}
and
\begin{equation}
f_{4}\left(  \mathcal{S}\right)  =\left\{
\begin{array}
[c]{ll}%
I\left(  X_{\mathcal{S}};Y_{r}|X_{S^{c}}V_{\mathcal{K}}X_{r}U\right)  &
\mathcal{S}\subseteq\mathcal{K},\mathcal{S}\not =\emptyset\\
0 & \mathcal{S}=\emptyset
\end{array}
\right.
\end{equation}
for some distribution satisfying (\ref{Prop_inp_dist}). The functions
$f_{1}(\cdot)$ and $f_{2}(\cdot)$ differ from $f_{3}(\cdot)$ and $f_{4}%
(\cdot)$, respectively, in the absence of the auxiliary random variables
$V_{\mathcal{K}}$. The proof of sub-modularity of $f_{3}$ and $f_{4}$ follows
along the same lines as those for the outer bounds except now we have the
Markov chain $\left(  X_{k},V_{k}\right)  -U-\left(  X_{j},V_{j}\right)  $ for
all $k\not =j$.

We thus have that the rate region $\mathcal{R}_{r}$ defined by the DF relay
bounds (see (\ref{Prop_DF_rateregion}))%
\begin{equation}
R_{\mathcal{S}}\leq I(X_{\mathcal{S}}X_{r};Y_{d}|X_{\mathcal{S}^{c}%
}V_{\mathcal{S}^{c}}),\quad\mathcal{S}\subseteq\mathcal{K} \label{DF_D_Bounds}%
\end{equation}
is a polymatroid associated with $f_{3}(\cdot)$ (see \cite[p.~767]%
{cap_theorems:Schrijver01}). Similarly, the region $\mathcal{R}_{d}$ defined
by the DF destination bounds (see (\ref{Prop_DF_rateregion}))%
\begin{equation}
R_{\mathcal{S}}\leq I(X_{\mathcal{S}};Y_{r}|X_{\mathcal{S}^{c}}V_{\mathcal{K}%
}X_{r}),\quad\mathcal{S}\subseteq\mathcal{K}%
\end{equation}
\newline is a polymatroid associated with $f_{4}(\cdot)$ (see \cite[p.~767]%
{cap_theorems:Schrijver01}).

\section{\label{DG_AppConvex}Concavity of $B_{d,\mathcal{S}}$ and
$I_{d,\mathcal{S}}$}

\subsection{Outer Bound $B_{d,\mathcal{S}}$}

Recall that the cutset bound at the destination, $B_{d,\mathcal{S}}$, is given
by%
\begin{equation}%
\begin{array}
[c]{cc}%
B_{d,\mathcal{S}}=C\left(  \frac{%
%TCIMACRO{\tsum \limits_{k\in\mathcal{S}}}%
%BeginExpansion
{\textstyle\sum\limits_{k\in\mathcal{S}}}
%EndExpansion
P_{k}}{N_{d}}+\frac{\left(  1-%
%TCIMACRO{\tsum \limits_{k\in\mathcal{S}^{c}}}%
%BeginExpansion
{\textstyle\sum\limits_{k\in\mathcal{S}^{c}}}
%EndExpansion
\gamma_{k}\right)  P_{r}}{N_{d}}+\frac{2%
%TCIMACRO{\tsum \limits_{k\in\mathcal{S}}}%
%BeginExpansion
{\textstyle\sum\limits_{k\in\mathcal{S}}}
%EndExpansion
\sqrt{\gamma_{k}P_{k}P_{r}}}{N_{d}}\right)  & \text{for all }\mathcal{S}%
\subseteq\mathcal{K}.
\end{array}
\label{DGAppCvx_BdS}%
\end{equation}
We show that $B_{d,\mathcal{S}}$ is a concave function of $\underline{\gamma
}_{\mathcal{K}}$. To prove concavity, one has to show that the Hessian or
second derivative of $B_{d,\mathcal{S}}$, $\nabla^{2}B_{d,\mathcal{S}}$, is
negative semi-definite, i.e, $\underline{x}^{T}\nabla^{2}B_{d,\mathcal{S}%
}\underline{x}\leq0$ for all $\underline{x}\in\mathcal{R}^{K}$ \cite[3.1.4]%
{cap_theorems:BVbook01}. We write%
\begin{equation}
B_{d,\mathcal{S}}=\frac{1}{2}\log\left(  K_{0}+2\sum\limits_{k\in\mathcal{S}%
}K_{k}\sqrt{\gamma_{k}}\right)
\end{equation}
where%
\begin{equation}%
\begin{array}
[c]{ll}%
K_{0}=1+\frac{\sum\limits_{k\in\mathcal{S}}P_{k}}{N_{d}}+\frac{P_{r}\left(
1-c\right)  }{N_{d}} & \\
K_{k}=\sqrt{\frac{P_{k}}{N_{d}}\frac{P_{r}}{N_{d}}}\text{ \ \ \ } &
k\in\mathcal{S}\text{.}%
\end{array}
\label{DGACvx_K0Kk}%
\end{equation}
The gradient $\nabla B_{d,\mathcal{S}}$ is given by
\begin{align}
\nabla B_{d,\mathcal{S}}  &  =\left[  \partial B_{d,\mathcal{S}}%
/\partial\gamma_{k}\right]  _{k\in\mathcal{K}}\\
&  =\frac{1}{K_{s}}\left[
\begin{array}
[c]{cc}%
\underline{v}_{\mathcal{S}} & \underline{v}_{\mathcal{S}^{c}}%
\end{array}
\right]  ^{T}\\
&  =\frac{1}{K_{s}}%
\end{align}
where $\underline{v}_{\mathcal{S}}$ is an $\left\vert \mathcal{S}\right\vert
$-length vector with entries $v_{k}=K_{k}\left/  \sqrt{\gamma_{k}}\right.  $
for all $k\in\mathcal{S}$, $\underline{v}_{\mathcal{S}^{c}}$ is an $\left\vert
\mathcal{S}^{c}\right\vert $-length vector with entries $v_{m}=-P_{r}\left/
N_{d}\right.  $ for all $m\in\mathcal{S}^{c}$, and
\begin{equation}
K_{s}=2\left(  K_{0}+2\sum\limits_{k\in\mathcal{S}}K_{k}\sqrt{\gamma_{k}%
}\right)  .
\end{equation}
The Hessian of $B_{d,\mathcal{S}}$, $\nabla^{2}B_{d,\mathcal{S}}$, is given by%
\begin{align}
\nabla^{2}B_{d,\mathcal{S}}  &  =\left[  \partial^{2}B_{d,\mathcal{S}%
}/\partial\gamma_{k}\partial\gamma_{m}\right]  _{\forall k,m\in\mathcal{K}}\\
&  =\frac{-1}{K_{s}}\text{diag}\left(  \underline{d}\right)  -\underline
{z}\underline{z}^{T} \label{DGACvx_BHess}%
\end{align}
where%
\begin{align}
\underline{z}  &  =\sqrt{2}\left(  \nabla B_{d,\mathcal{S}}\right) \\
\underline{d}  &  =\left[
\begin{array}
[c]{cc}%
\underline{d}_{\mathcal{S}} & \underline{d}_{\mathcal{S}^{c}}%
\end{array}
\right]  ^{T}%
\end{align}
such that $\underline{d}_{\mathcal{S}}$ is an $\left\vert \mathcal{S}%
\right\vert $-length vector with entries $d_{k}=K_{k}\left/  2\gamma_{k}%
^{3/2}\right.  $ for all $k\in\mathcal{S}$, and $\underline{d}_{\mathcal{S}%
^{c}}$ is an $\left\vert \mathcal{S}^{c}\right\vert $-length vector with
entries $d_{k}=-2P_{r}^{2}\left/  \left(  N_{d}^{2}K_{s}\right)  \right.  $
for all $k\in\mathcal{S}^{c}$. Using the fact that $K_{k}$ and $\gamma_{k}$
are non-negative for all $k$, from (\ref{DGACvx_BHess}), for any
$\underline{x}\in\mathcal{R}^{K}$, we have
\begin{align}
\underline{x}^{T}\nabla^{2}B_{d,\mathcal{S}}\underline{x}  &  =-\frac{1}%
{K_{s}}\left(
%TCIMACRO{\tsum \limits_{k\in\mathcal{K}}}%
%BeginExpansion
{\textstyle\sum\limits_{k\in\mathcal{K}}}
%EndExpansion
x_{k}^{2}d_{k}\right)  -\left(  \underline{x}^{T}\cdot\underline{z}\right)
^{2}\\
&  \leq0
\end{align}
with equality if and only if $\underline{x}$ $=$ $\underline{0}$. In proving
the concavity of $B_{d,\mathcal{S}}$, we assume only that $\gamma_{k}>0$, for
all $k$. Thus, from continuity, the concavity also holds for all non-negative
$\gamma_{k}$ satisfying (see (\ref{DGOB_TOB}))
\begin{equation}
\sum\limits_{k\in\mathcal{K}}\gamma_{k}\leq1. \label{DGAppCvx_Gammasum}%
\end{equation}

For a fixed $\underline{\gamma}_{\mathcal{S}^{c}}$, we now find the
$\underline{\gamma}_{\mathcal{S}}$ that maximizes $B_{d,\mathcal{S}}$ subject
to (\ref{DGAppCvx_Gammasum}) above. For a $c\in\lbrack0,1)$, we fix
$\underline{\gamma}_{\mathcal{S}^{c}}$ such that its entries $\gamma_{k}$, for
all $k\in\mathcal{S}^{c}$, satisfy
\begin{equation}
\sum\limits_{k\in\mathcal{S}^{c}}\gamma_{k}=1-c, \label{DGAppCvx_GSc}%
\end{equation}
and thus, from (\ref{DGAppCvx_Gammasum}) we have
\begin{equation}
\sum\limits_{k\in\mathcal{S}}\gamma_{k}\leq c. \label{DGACvx_GS}%
\end{equation}
Since $B_{d,\mathcal{S}}$ is a continuous concave function of $\underline
{\gamma}_{\mathcal{S}}$ it achieves its maximum at a \underline{$\gamma$%
}$_{\mathcal{S}}^{\ast}$ where%
\begin{equation}%
\begin{array}
[c]{cc}%
\left.  \frac{\partial B_{d,\mathcal{S}}}{\partial\gamma_{k}}\right\vert
_{\gamma_{k}^{\ast}}=0 & \text{for all }k\in\mathcal{S}\text{.}%
\end{array}
\end{equation}
Using the method of Lagrange multipliers, we find that a \underline{$\gamma$%
}$_{\mathcal{S}}^{\ast}$ that maximizes $B_{d,\mathcal{S}}$ subject to
(\ref{DGAppCvx_GSc}) and (\ref{DGACvx_GS}) has entries%
\begin{equation}
\gamma_{k}^{\ast}=\left\{
\begin{array}
[c]{cc}%
\frac{cP_{k}}{\sum\limits_{k\in\mathcal{S}}P_{k}} & k\in\mathcal{S}%
\end{array}
\right.  .
\end{equation}

\subsection{Inner Bound $I_{d,\mathcal{S}}$}

Recall that the DF bound, $I_{d,\mathcal{S}}$, at the destination is given as%
\begin{equation}%
\begin{array}
[c]{cc}%
I_{d,\mathcal{S}}=C\left(  \frac{%
%TCIMACRO{\tsum \limits_{k\in\mathcal{S}}}%
%BeginExpansion
{\textstyle\sum\limits_{k\in\mathcal{S}}}
%EndExpansion
P_{k}}{N_{d}}+\frac{\left(  1-%
%TCIMACRO{\tsum \limits_{k\in\mathcal{S}^{c}}}%
%BeginExpansion
{\textstyle\sum\limits_{k\in\mathcal{S}^{c}}}
%EndExpansion
\beta_{k}\right)  P_{r}}{N_{d}}+\frac{2%
%TCIMACRO{\tsum \limits_{k\in\mathcal{S}}}%
%BeginExpansion
{\textstyle\sum\limits_{k\in\mathcal{S}}}
%EndExpansion
\sqrt{\left(  1-\alpha_{k}\right)  \beta_{k}P_{k}P_{r}}}{N_{d}}\right)   &
\text{for all }\mathcal{S}\subseteq\mathcal{K}.
\end{array}
\label{DGACvx_IdS}%
\end{equation}
Comparing (\ref{DGAppCvx_BdS}) and (\ref{DGACvx_IdS}), for $\gamma_{k}%
\overset{\vartriangle}{=}\left(  1-\alpha_{k}\right)  \beta_{k}$ for all
$k\in\mathcal{S}$ and $\gamma_{k}\overset{\vartriangle}{=}\beta_{k}$ for all
$k\in\mathcal{S}^{c}$, the DF rate bound in (\ref{DGACvx_IdS}) simplifies to
that for the outer bound in (\ref{DGAppCvx_BdS}), and thus, one can use the
same technique to show that $I_{d,\mathcal{S}}$ is a concave function of
$\underline{\alpha}_{\mathcal{K}}$ and $\underline{\beta}_{\mathcal{K}}$. For
the power fractions $\beta_{k}$, we have
\begin{equation}
\sum\limits_{k\in\mathcal{K}}\beta_{k}\leq1.\label{App_betasum_ineq}%
\end{equation}
For a fixed $\underline{\alpha}_{\mathcal{K}}$, we determine the optimal
$\underline{\beta}_{\mathcal{S}}$ maximizing $I_{d,\mathcal{S}}$ by fixing the
vector $\underline{\beta}_{\mathcal{S}^{c}}$ such that
\begin{align}
\sum\limits_{k\in\mathcal{S}^{c}}\beta_{k} &  =1-c\label{DGACvx_betaSc}\\
\sum\limits_{k\in\mathcal{S}}\beta_{k} &  \leq c.\label{DGACvx_betaS}%
\end{align}
where $c\in\lbrack0,1)$. Since $I_{d,\mathcal{S}}$ is independent of
\underline{$\beta$}$_{\mathcal{S}}$ for \underline{$\alpha$}$_{\mathcal{S}%
}=\underline{1}$, we assume that \underline{$\alpha$}$_{\mathcal{S}}\not =$
\underline{$1$}.%

%TCIMACRO{\FRAME{ftbpFU}{4.6112in}{3.6659in}{0pt}{\Qcb{Rate region achieved at
%the destination for a two-user MARC and $\alpha_{1}=\alpha_{2}=1/2$. }%
%}{\Qlb{Fig_GMARC_DF_destrate_Region}}%
%{gmarc_twouser_destination_region_100406.eps}%
%{\special{ language "Scientific Word";  type "GRAPHIC";  display "USEDEF";
%valid_file "F";  width 4.6112in;  height 3.6659in;  depth 0pt;
%original-width 0pt;  original-height 0pt;  cropleft "0";  croptop "1";
%cropright "1";  cropbottom "0";
%filename 'GMARC_twouser_destination_region_100406.eps';file-properties "XNPEU";}%
%}}%
%BeginExpansion
\begin{figure}
[ptb]
\begin{center}
\includegraphics[
height=3.6659in,
width=4.6112in
]%
{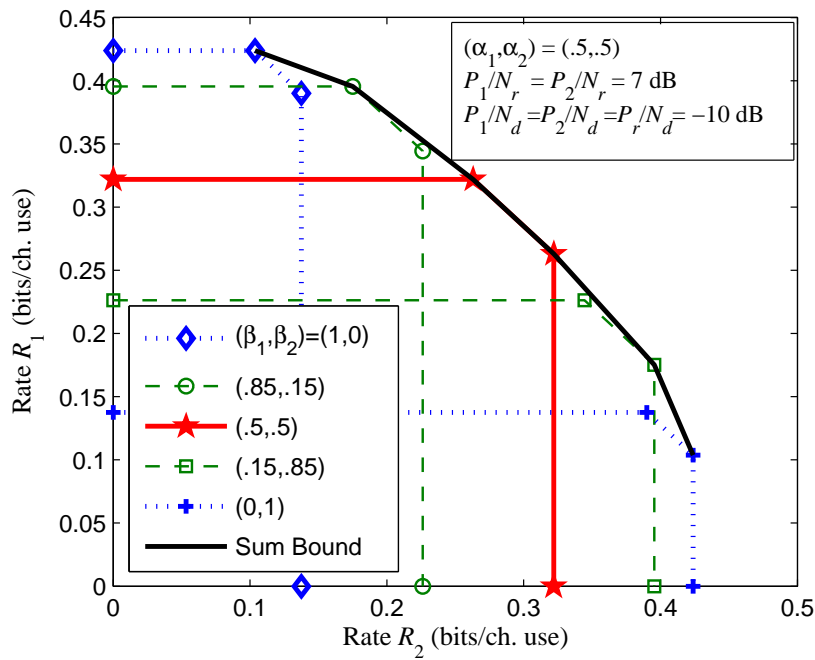}%
\caption{Rate region achieved at the destination for a two-user MARC and
$\alpha_{1}=\alpha_{2}=1/2$. }%
\label{Fig_GMARC_DF_destrate_Region}%
\end{center}
\end{figure}
%EndExpansion

We now consider the special case \thinspace in which $\underline{\alpha
}_{\mathcal{S}}\not =\underline{1}$ and $\underline{\beta}_{\mathcal{S}^{c}}$
are fixed. We determine a $\underline{\beta}_{\mathcal{S}}$ that maximizes
$I_{d,\mathcal{S}}$ subject to (\ref{DGACvx_betaS}) and (\ref{DGACvx_betaSc}).
Since $I_{d,\mathcal{S}}$ is a continuous concave function of $\underline
{\beta}_{\mathcal{S}}$ it achieves its maximum at a \underline{$\beta$%
}$_{\mathcal{S}}^{\ast}$ where%
\begin{equation}%
\begin{array}
[c]{cc}%
\left.  \frac{\partial I_{d,\mathcal{S}}}{\partial\beta_{k}}\right\vert
_{\beta_{k}^{\ast}}=0 & \text{for all }k\in\mathcal{S}\text{.}%
\end{array}
\end{equation}
As before, using Lagrange multipliers, the optimal \underline{$\beta$%
}$_{\mathcal{S}}^{\ast}$ that maximizes $I_{d,\mathcal{S}}$, subject to
(\ref{DGACvx_betaS}), has entries%
\begin{equation}
\beta_{k}^{\ast}=\left\{
\begin{array}
[c]{cc}%
\frac{c\left(  1-\alpha_{k}\right)  P_{k}}{\sum\limits_{k\in\mathcal{S}%
}\left(  1-\alpha_{k}\right)  P_{k}} & k\in\mathcal{S}%
\end{array}
\right.  . \label{App_beta_opt}%
\end{equation}

\textit{Rate region for a fixed }\underline{$\alpha$}$_{\mathcal{K}}$: For any
choice of a non-zero \underline{$\alpha$}$_{\mathcal{K}}$ and a \underline
{$\beta$}$_{\mathcal{K}}$ satisfying (\ref{App_betasum_ineq}), the rate region
given by (\ref{DGACvx_IdS}) for all~$\mathcal{S}$ is a polymatroid. For
\underline{$\alpha$}$_{\mathcal{K}}=\underline{1}$, from (\ref{DGACvx_IdS}) we
see that there are no gains achieved from coherent combining, i.e., it
suffices to choose \underline{$\beta$}$_{\mathcal{K}}=\underline{0}$. Consider
\underline{$\alpha$}$_{\mathcal{K}}\not =\underline{1}$. Since there is at
least one $k$ for which $\alpha_{k}<1$, gains from coherent combining at the
destination are maximized by choosing \underline{$\beta$}$_{\mathcal{K}}$ to
satisfy (\ref{App_betasum_ineq}) with equality. For a fixed \underline
{$\alpha$}$_{\mathcal{K}}$, we then write the rate region at the destination
as a union over all polymatroids, one for each choice of \underline{$\beta$%
}$_{\mathcal{K}}$ satisfying%
\begin{equation}
\sum_{k=1}^{K}\beta_{k}=1.
\end{equation}
Observe that for $\underline{\beta}_{\mathcal{K}}^{\ast}$ with entries given
by (\ref{App_beta_opt}), the bound $I_{d,\mathcal{S}}$ is maximized. In Fig
\ref{Fig_GMARC_DF_destrate_Region}, we illustrate the rate region for a
two-user degraded Gaussian\ MARC with the SNR$\ P_{1}/N_{d}=P_{2}/N_{d}$
chosen as $-10$ dB, \underline{$\alpha$} $=$ $(1/2,1/2)$, and five choices of
\underline{$\beta$}$_{\mathcal{K}}$. Observe that the maximum single-user rate
$R_{1}$ is achieved by setting $\beta_{1}$ to $1$ though this value does not
maximize $R_{2}$ or $R_{1}+R_{2}$. For all other $(\beta_{1},\beta_{2})$ such
as $(0.85,0.15)$, as $\beta_{1}$ decreases and $\beta_{2}$ increases, $R_{1}$
decreases while $R_{2}$ increases achieving its maximum at $\beta_{2}=1$. The
bound on the sum rate $R_{1}+R_{2}$ increases from $(\beta_{1},\beta
_{2})=(1,0)$, achieves its maximum at $(\beta_{1}^{\ast},\beta_{2}^{\ast})$
$=$ $(1/2,1/2)$, and then decreases as $\beta_{2}$ approaches $1$. The
resulting region at the destination is shown in Fig.
\ref{Fig_GMARC_DF_destrate_Region} as a union over all polymatroids, one for
each choice of \underline{$\beta$}$_{\mathcal{K}}$.

\section{\label{DG_Appen_2}$B_{r,\mathcal{S}}$ vs. \underline{$\gamma$%
}$_{\mathcal{K}}$}

We show that the function $B_{r,\mathcal{S}}$ in (\ref{Con_final_B1S}) is a
concave function of $\underline{\gamma}_{\mathcal{S}}$ for a fixed
$\underline{\gamma}_{\mathcal{S}^{c}}$ and for all $\mathcal{S\subseteq K}$.
Recall the expression for $B_{r,\mathcal{S}}$ as
\begin{equation}
B_{r,\mathcal{S}}=C\left(
%TCIMACRO{\tsum \limits_{k\in\mathcal{S}}}%
%BeginExpansion
{\textstyle\sum\limits_{k\in\mathcal{S}}}
%EndExpansion
\frac{P_{k}}{N_{r}}-\frac{\left(  \sum\limits_{k\in\mathcal{S}}\sqrt
{\gamma_{k}P_{k}}\right)  ^{2}}{N_{r}\left(  1-%
%TCIMACRO{\tsum \limits_{k\in\mathcal{S}^{c}}}%
%BeginExpansion
{\textstyle\sum\limits_{k\in\mathcal{S}^{c}}}
%EndExpansion
\gamma_{k}\right)  }\right)  \label{A2_BrS}%
\end{equation}
where we assume that
\begin{equation}
\sum\limits_{k\in\mathcal{S}^{c}}\gamma_{k}=1-c<1.
\end{equation}
Observe that $B_{r,\mathcal{S}}$ is maximized when $c=1$, i.e., $\gamma_{k}=0$
for all $k\in\mathcal{S}$, and minimized for $c=0$. Further, comparing
$B_{r,\mathcal{S}}$ and $B_{d,\mathcal{S}}$, one can see that for
\begin{equation}
\gamma_{k}=\left\{
\begin{array}
[c]{lll}%
P_{k}\left/  \left(  \sum_{k\in\mathcal{S}}P_{k}\right)  \right.  & , &
k\in\mathcal{S}\\
0 & , & k\in\mathcal{S}^{c}%
\end{array}
\right.
\end{equation}
$B_{r,\mathcal{S}}$ achieves its minimum, i.e., $B_{r,\mathcal{S}}=0$.

We write%
\begin{equation}
x\overset{\vartriangle}{=}\left(  \sum\limits_{k\in\mathcal{S}}\sqrt
{\gamma_{k}\lambda_{k}}\right)  \label{A2_x_sub}%
\end{equation}
where
\begin{equation}%
\begin{array}
[c]{ccc}%
P_{\max}=\max_{k\in\mathcal{K}}P_{k} & \text{and} & \lambda_{k}=\left.
P_{k}\right/  P_{\max}\text{.}%
\end{array}
\label{A2_Pmax_def}%
\end{equation}
Substituting (\ref{A2_x_sub}) in the expression for $B_{r,\mathcal{S}}$ in
(\ref{A2_BrS}), we have%
\begin{equation}
B_{r,\mathcal{S}}=C\left(
%TCIMACRO{\tsum \limits_{k\in\mathcal{S}}}%
%BeginExpansion
{\textstyle\sum\limits_{k\in\mathcal{S}}}
%EndExpansion
\frac{P_{k}}{N_{r}}-\frac{x^{2}P_{\max}}{N_{r}c}\right)  . \label{DG_A2_BrS}%
\end{equation}
Differentiating $B_{r,\mathcal{S}}$ with respect to $x$ we have
\begin{align}
\frac{dB_{r,\mathcal{S}}}{dx}  &  =\frac{-xP_{\max}}{N_{r}c}\cdot\left(  1+%
%TCIMACRO{\tsum \limits_{k\in\mathcal{S}}}%
%BeginExpansion
{\textstyle\sum\limits_{k\in\mathcal{S}}}
%EndExpansion
\frac{P_{k}}{N_{r}}-\frac{x^{2}P_{\max}}{N_{r}c}\right)  ^{-1}\\
\frac{d^{2}B_{r,\mathcal{S}}}{dx^{2}}  &  =\frac{-\frac{P_{\max}}{N_{r}%
c}\left(  1+%
%TCIMACRO{\tsum \limits_{k\in\mathcal{S}}}%
%BeginExpansion
{\textstyle\sum\limits_{k\in\mathcal{S}}}
%EndExpansion
\frac{P_{k}}{N_{r}}\right)  -\left(  \frac{xP_{\max}}{N_{r}c}\right)  ^{2}%
}{\left(  1+%
%TCIMACRO{\tsum \limits_{k\in\mathcal{S}}}%
%BeginExpansion
{\textstyle\sum\limits_{k\in\mathcal{S}}}
%EndExpansion
\frac{P_{k}}{N_{r}}-\frac{x^{2}P_{\max}}{N_{r}c}\right)  ^{2}}\label{A2_Br_d2}%
\\
&  <0 \label{A2_Br_d2_ineq}%
\end{align}
where the strict inequality in (\ref{A2_Br_d2_ineq}) follows since all terms
in (\ref{A2_Br_d2}) are positive. Further, for any $c>0\,$, from
(\ref{DG_A2_BrS}) $B_{r,\mathcal{S}}$ is maximized at $x=0$, i.e., for
$\gamma_{k}=0$ for all $k\in\mathcal{S}$. Thus, we see that $B_{r,\mathcal{S}%
}$ is a concave decreasing function of $x$.

\section{\label{DG_App4OBProof}Proof of Theorem \ref{DGOB_Th2}}

We now prove Theorem \ref{DGOB_Th2} and give the solution to the max-min
optimization
\begin{equation}
R_{\mathcal{K}}=\max\limits_{\underline{\gamma}_{\mathcal{K}}\in\Gamma_{OB}%
}\min\left(  B_{r,\mathcal{K}}\left(  \underline{\gamma}_{\mathcal{K}}\right)
,B_{d,\mathcal{K}}\left(  \underline{\gamma}_{\mathcal{K}}\right)  \right)  .
\label{DG_A3_maxmin}%
\end{equation}

Consider the function
\begin{equation}%
\begin{array}
[c]{cc}%
J(\underline{\gamma}_{\mathcal{K}},\delta)=\delta B_{r,\mathcal{K}}\left(
\underline{\gamma}_{\mathcal{K}}\right)  +\left(  1-\delta\right)
B_{d,\mathcal{K}}\left(  \underline{\gamma}_{\mathcal{K}}\right)  , &
\delta\in\lbrack0,1].
\end{array}
\label{DGA3_Jfn_defn}%
\end{equation}
Observe that $J$ is linear in $\delta$ ranging in value from $I_{d,\mathcal{K}%
}$ for $\delta$ $=$ $0$ to $I_{r,\mathcal{K}}$ for $\delta=1$. Thus, the
optimization in (\ref{DGA3_Jfn_defn}) is equivalent to maximizing the minimum
of the two end points of the line $J$ over $\Gamma_{OB}$. Maximizing
$J(\underline{\gamma}_{\mathcal{K}},\delta)$ over $\underline{\gamma
}_{\mathcal{K}}$, we obtain a continuous convex function
\begin{equation}%
\begin{array}
[c]{cc}%
V(\delta)=\max\limits_{\underline{\gamma}_{\mathcal{K}}\in\Gamma_{OB}%
}J(\underline{\gamma}_{\mathcal{K}},\delta), & \delta\in\lbrack0,1].
\end{array}
\label{DGA3_Vdelta_defn}%
\end{equation}
From (\ref{DGA3_Jfn_defn}) and (\ref{DGA3_Vdelta_defn}), we see that for any
$\underline{\gamma}_{\mathcal{K}}$, $J(\underline{\gamma}_{\mathcal{K}}%
,\delta)$ either lies strictly below or is tangential to $V(\delta)$. The
following proposition summarizes a well-known solution to the max-min problem
in (\ref{DG_A3_maxmin}) (see \cite{cap_theorems:LiangVP_ResAllocJrnl}).

\begin{proposition}
\label{Prop_minimax}$\underline{\gamma}_{\mathcal{K},\mathcal{\delta}^{\ast}%
}^{\ast}$ is a max-min rule where
\begin{equation}
\delta^{\ast}=\arg\min_{\delta\in\lbrack0,1]}V(\delta).
\end{equation}
The maximum bound on $R_{\mathcal{K}}$, $V(\delta^{\ast})$, is completely
determined by the following three cases (see Fig. \ref{Fig_Prop_1}).%
\begin{align}
&
\begin{array}
[c]{cl}%
\text{Case 1:} & \delta^{\ast}=0:V(\delta^{\ast})=B_{d,\mathcal{K}}%
(\underline{\gamma}_{\mathcal{K},\mathcal{\delta}^{\ast}}^{\ast}%
)<B_{r,\mathcal{K}}(\underline{\gamma}_{\mathcal{K},\mathcal{\delta}^{\ast}%
}^{\ast})
\end{array}
\label{DGA4_Case1}\\
&
\begin{array}
[c]{cl}%
\text{Case 2:} & \delta^{\ast}=1:V(\delta^{\ast})=B_{r,\mathcal{K}}%
(\underline{\gamma}_{\mathcal{K},\mathcal{\delta}^{\ast}}^{\ast}%
)<B_{d,\mathcal{K}}(\underline{\gamma}_{\mathcal{K},\mathcal{\delta}^{\ast}%
}^{\ast})
\end{array}
\label{DGA4_Case2}\\
&
\begin{array}
[c]{cl}%
\text{Case 3:} & 0<\delta^{\ast}<1:V(\delta^{\ast})=B_{r,\mathcal{K}%
}(\underline{\gamma}_{\mathcal{K},\mathcal{\delta}^{\ast}}^{\ast
})=B_{d,\mathcal{K}}(\underline{\gamma}_{\mathcal{K},\mathcal{\delta}^{\ast}%
}^{\ast}).
\end{array}
\label{DGA4_Case3}%
\end{align}

\end{proposition}

We apply Proposition \ref{Prop_minimax} to determine the maximum bound on
$R_{\mathcal{K}}$. We study each case separately and determine the max-min
rule $\underline{\gamma}_{\mathcal{K}}^{\ast}$ for each case. In general, the
max-min rule $\underline{\gamma}_{\mathcal{K},\mathcal{\delta}^{\ast}}^{\ast}$
depends on an optimal $\mathcal{\delta}^{\ast}$. However, for notational
convenience we henceforth omit the subscript $\delta^{\ast}$ in denoting the
max-min rule. We develop the optimal $\underline{\gamma}_{\mathcal{K}}^{\ast}$
and the maximum sum-rate for each case. We first consider case $1$ and show
that this case is not feasible.%

%TCIMACRO{\TeXButton{B}{\begin{figure*}[tbp] \centering}}%
%BeginExpansion
\begin{figure*}[tbp] \centering
%EndExpansion%
%TCIMACRO{\FRAME{itbpF}{5.9724in}{1.8948in}{0in}{}{}{prop_i__draw.eps}%
%{\special{ language "Scientific Word";  type "GRAPHIC";  display "USEDEF";
%valid_file "F";  width 5.9724in;  height 1.8948in;  depth 0in;
%original-width 8.7121in;  original-height 4.7262in;  cropleft "0";
%filename 'Prop_I__draw.eps';file-properties "XNPEU";}}}%
%BeginExpansion
{\includegraphics[
height=1.8948in,
width=5.9724in
]%
{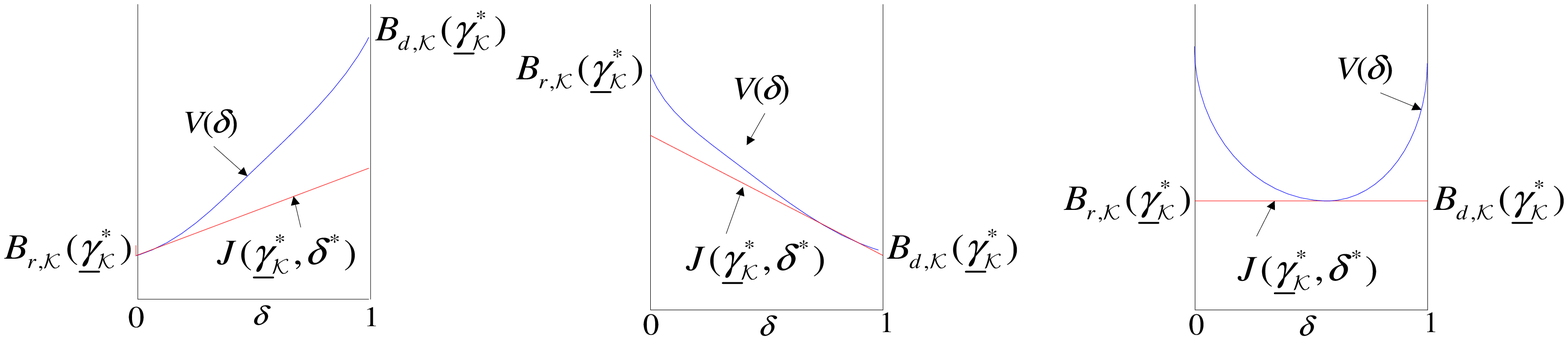}%
}%
%EndExpansion
\caption{Illustration of Cases 1, 2, and 3.}\label{Fig_Prop_1}%
%TCIMACRO{\TeXButton{E}{\end{figure*}}}%
%BeginExpansion
\end{figure*}%
%EndExpansion

\textit{Case 1}: This case occurs when the maximum bound achievable at the
destination is smaller than the bound at the relay. In Appendix
\ref{DG_AppConvex}, we show that the bound $B_{d,\mathcal{K}}(\underline
{\gamma}_{\mathcal{K}})$ is a concave function of $\underline{\gamma
}_{\mathcal{K}}$ and achieves a maximum at \underline{$\gamma$}$_{\mathcal{K}%
}^{\ast}$ whose entries $\gamma_{k}^{\ast}$ satisfy
(\ref{OB_gammak_constraint}) and are given as%
\begin{equation}%
\begin{array}
[c]{cc}%
\gamma_{k}^{\ast}=P_{k}\left/  \left(  \sum_{k\in\mathcal{K}}P_{k}\right)
\right.  , & \text{for all }k\in\mathcal{K}\text{.}%
\end{array}
\label{DGA4_Case1_Gamma}%
\end{equation}
Substituting (\ref{DGA4_Case1_Gamma}) in (\ref{Con_final_B1S}), we have
$B_{r,\mathcal{K}}(\underline{\gamma}_{\mathcal{K}}^{\ast})=0$ which
contradicts the assumption in (\ref{DGA4_Case1}), thus making this case infeasible.

\textit{Case 2}: Consider the condition for case 2 in (\ref{DGA4_Case2}). This
condition implies that the case occurs when the maximum bound achievable at
the relay is smaller than the bound at the destination. From
(\ref{GMARC_DF_IrG}), we observe that $B_{r,\mathcal{K}}$ decreases
monotonically with $\gamma_{k}$ for all $k$ and achieves a maximum of
\begin{equation}
B_{r,\mathcal{K}}(\underline{\gamma}_{\mathcal{K}}^{\ast})=C\left(  \frac{%
%TCIMACRO{\tsum \limits_{k\in\mathcal{K}}}%
%BeginExpansion
{\textstyle\sum\limits_{k\in\mathcal{K}}}
%EndExpansion
P_{k}}{N_{r}}\right)
\end{equation}
at $\underline{\gamma}_{\mathcal{K}}^{\ast}=\underline{0}$. Comparing
(\ref{Con_final_B1S}) and (\ref{Con_final_B2S}) at $\underline{\gamma
}_{\mathcal{K}}^{\ast}$ $=$ $\underline{0}$, we obtain the condition for this
case as%
\begin{equation}
\frac{%
%TCIMACRO{\tsum \limits_{k\in\mathcal{K}}}%
%BeginExpansion
{\textstyle\sum\limits_{k\in\mathcal{K}}}
%EndExpansion
P_{k}}{N_{r}}\leq\frac{%
%TCIMACRO{\tsum \limits_{k\in\mathcal{K}}}%
%BeginExpansion
{\textstyle\sum\limits_{k\in\mathcal{K}}}
%EndExpansion
P_{k}+P_{r}}{N_{d}}\text{.}\label{DGA4_C2_CondK}%
\end{equation}
\textit{Case 3}: Finally, consider the condition for Case 3 in
(\ref{DGA4_Case3}). This case occurs when the maximum rate bound achievable at
the relay and destination are equal. The max-min solution for this case is
obtained by considering two sub-cases. The first is the relatively
straightforward sub-case where $\underline{\gamma}_{\mathcal{K}}^{\ast}$ $=$
$\underline{0}$ is the max-min rule. The resulting maximum sum-rate is the
same as that for \textit{case} $2$ with the condition in (\ref{DGA4_C2_CondK})
satisfied with equality. Consider the second sub-case where $\underline
{\gamma}_{\mathcal{K}}^{\ast}$ $\not =$ $\underline{0}$, i.e., when%
\begin{equation}
\frac{%
%TCIMACRO{\tsum \limits_{k\in\mathcal{K}}}%
%BeginExpansion
{\textstyle\sum\limits_{k\in\mathcal{K}}}
%EndExpansion
P_{k}}{N_{r}}>\frac{%
%TCIMACRO{\tsum \limits_{k\in\mathcal{K}}}%
%BeginExpansion
{\textstyle\sum\limits_{k\in\mathcal{K}}}
%EndExpansion
P_{k}+P_{r}}{N_{d}}\text{.}%
\end{equation}
We formulate the optimization problem for this case as%
\begin{equation}%
\begin{array}
[c]{ll}%
\text{maximize} & B_{r,\mathcal{K}}\left(  \underline{\gamma}\right)  \\
\text{subject to} & B_{r,\mathcal{K}}\left(  \underline{\gamma}\right)
=B_{d,\mathcal{K}}\left(  \underline{\gamma}\right)  .
\end{array}
\label{DGA4_maxprob}%
\end{equation}
We write%
\begin{equation}%
\begin{array}
[c]{cc}%
P_{\max}=\max_{k\in\mathcal{K}}P_{k}, & \lambda_{k}=\left.  P_{k}\right/
P_{\max}\text{,}%
\end{array}
\label{DGA4_lambdadef}%
\end{equation}
and define
\begin{equation}
x\overset{\vartriangle}{=}\sqrt{%
%TCIMACRO{\tsum \limits_{k\in\mathcal{K}}}%
%BeginExpansion
{\textstyle\sum\limits_{k\in\mathcal{K}}}
%EndExpansion
\lambda_{k}\gamma_{k}}.\label{DGA4_xKdef}%
\end{equation}
Substituting (\ref{DGA4_lambdadef}) and (\ref{DGA4_xKdef}) in
(\ref{Con_final_B1S}) and (\ref{Con_final_B2S}), we have%
\begin{align}
B_{r,\mathcal{K}}(x) &  =C\left(  \frac{\left(
%TCIMACRO{\tsum \limits_{k\in\mathcal{K}}}%
%BeginExpansion
{\textstyle\sum\limits_{k\in\mathcal{K}}}
%EndExpansion
P_{k}\right)  -x^{2}P_{\max}}{N_{r}}\right)  \label{DGA4_BrK}\\
B_{d,\mathcal{K}}(x) &  =C\left(  \frac{\left(
%TCIMACRO{\tsum \limits_{k\in\mathcal{K}}}%
%BeginExpansion
{\textstyle\sum\limits_{k\in\mathcal{K}}}
%EndExpansion
P_{k}\right)  +P_{r}+2x\sqrt{P_{\max}P_{r}}}{N_{d}}\right)  .\label{DGA4_BdK}%
\end{align}
Observe that $B_{r,\mathcal{K}}(x)$ and $B_{d,\mathcal{K}}(x)$ are
monotonically decreasing and increasing functions of $x$, respectively, and
thus, the maximization in (\ref{DGA4_maxprob}) simplifies to determining an
$x$ such that
\begin{equation}
\frac{%
%TCIMACRO{\tsum \limits_{k\in\mathcal{K}}}%
%BeginExpansion
{\textstyle\sum\limits_{k\in\mathcal{K}}}
%EndExpansion
P_{k}-x^{2}P_{\max}}{N_{r}}=\frac{%
%TCIMACRO{\tsum \limits_{k\in\mathcal{K}}}%
%BeginExpansion
{\textstyle\sum\limits_{k\in\mathcal{K}}}
%EndExpansion
P_{k}+P_{r}+2x\sqrt{P_{\max}P_{r}}}{N_{d}}.\label{DGA4_quad}%
\end{equation}
We write%
\begin{equation}%
\begin{array}
[c]{lll}%
K_{0}=P_{\max}\left/  N_{r}\right.  , &  & K_{1}=\sqrt{P_{\max}P_{r}}\left/
N_{d}\right.  \\
K_{2}=\frac{\sum\nolimits_{k\in\mathcal{K}}P_{k}}{N_{d}}+\frac{P_{r}}{N_{d}%
}, & \text{and} & K_{3}=\frac{%
%TCIMACRO{\tsum \nolimits_{k\in\mathcal{K}}}%
%BeginExpansion
{\textstyle\sum\nolimits_{k\in\mathcal{K}}}
%EndExpansion
P_{k}}{N_{r}}.
\end{array}
\label{DGA4_Kdefs}%
\end{equation}
From (\ref{DF_Case2_CondK}), since $K_{3}>K_{2}$, the quadratic equation in
(\ref{DGA4_quad}) has only one positive solution given by
\begin{equation}
x^{\ast}=\frac{-K_{1}+\sqrt{K_{1}^{2}+\left(  K_{3}-K_{2}\right)  K_{0}}%
}{K_{0}}.\label{DGA4_xstar}%
\end{equation}
The optimal power policy for this case is then the set $\mathcal{G}$ of
$\underline{\gamma}_{\mathcal{K}}^{\ast}$ for which $\underline{\gamma
}_{\mathcal{K}}^{\ast}$ satisfies (\ref{DGA4_xKdef}) with $x=$ $x^{\ast}$ in
(\ref{DGA4_xstar}). The maximum achievable sum-rate for this case is then
obtained from (\ref{DGA4_BrK}) as
\begin{equation}
C\left(  \frac{%
%TCIMACRO{\tsum \limits_{k\in\mathcal{K}}}%
%BeginExpansion
{\textstyle\sum\limits_{k\in\mathcal{K}}}
%EndExpansion
P_{k}-\left(  x^{\ast}\right)  ^{2}P_{\max}}{N_{r}}\right)
.\label{DGA4_C3_Rate}%
\end{equation}

\section{\label{DG_App5_DFProof}Proof of Theorem \ref{Th_GMARC_DF_2}}

We now prove Theorem \ref{Th_GMARC_DF_2} and give the solution to the max-min
optimization
\begin{equation}
R_{\mathcal{K}}=\max\limits_{(\underline{\alpha}_{\mathcal{K}},\underline
{\beta}_{\mathcal{K}})\in\Gamma}\min\left(  I_{r,\mathcal{K}}\left(
\underline{\alpha}_{\mathcal{K}}\right)  ,I_{d,\mathcal{K}}\left(
\underline{\alpha}_{\mathcal{K}},\underline{\beta}_{\mathcal{K}}\right)
\right)  . \label{DGA5_MM}%
\end{equation}

As in Appendix \ref{DG_App4OBProof}, a solution to the max-min optimization in
(\ref{DGA5_MM}) simplifies to three mutually exclusive cases \cite[II.C]%
{cap_theorems:HVPoor01} such that the max-min rule $(\underline{\alpha
}_{\mathcal{K}}^{\ast},\underline{\beta}_{\mathcal{K}}^{\ast})$ satisfies the
conditions for one of three cases. The conditions for the three cases are%
\begin{align}
&
\begin{array}
[c]{cc}%
\text{Case }1\text{:} & I_{d,\mathcal{K}}(\underline{\alpha}_{\mathcal{K}%
}^{\ast},\underline{\beta}_{\mathcal{K}}^{\ast})<I_{r,\mathcal{K}}%
(\underline{\alpha}_{\mathcal{K}}^{\ast})
\end{array}
\label{DGA5_Case1}\\
&
\begin{array}
[c]{cc}%
\text{Case }2\text{:} & I_{r,\mathcal{K}}(\underline{\alpha}_{\mathcal{K}%
}^{\ast})<I_{d,\mathcal{K}}(\underline{\alpha}_{\mathcal{K}}^{\ast}%
,\underline{\beta}_{\mathcal{K}}^{\ast})
\end{array}
\label{DGA5_Case2}\\
&
\begin{array}
[c]{cc}%
\text{Case }3\text{:} & I_{r,\mathcal{K}}(\underline{\alpha}_{\mathcal{K}%
}^{\ast})=I_{d,\mathcal{K}}(\underline{\alpha}_{\mathcal{K}}^{\ast}%
,\underline{\beta}_{\mathcal{K}}^{\ast})
\end{array}
\label{DGA5_Case3}%
\end{align}
We develop the conditions and determine the max-min rule for each case. We
first consider case $1$ and show that this case is not feasible.

\textit{Case 1}: This case occurs when the maximum bound achievable at the
destination is smaller than the bound at the relay. Observe that
$I_{d,\mathcal{K}}(\underline{\alpha}_{\mathcal{K}},\underline{\beta
}_{\mathcal{K}})$ in (\ref{GMARC_DF_IdG}) decreases monotonically with
$\alpha_{k}$, for all $k$, and, for any \underline{$\beta$}$_{\mathcal{K}}$,
achieves a maximum at \underline{$\alpha$}$_{\mathcal{K}}^{\ast}$ $=$
\underline{$0$} of
\begin{equation}
I_{d,\mathcal{K}}(\underline{\alpha}_{\mathcal{K}}^{\ast},\underline{\beta
}_{\mathcal{K}})=C\left(  \frac{%
%TCIMACRO{\tsum \limits_{k\in\mathcal{K}}}%
%BeginExpansion
{\textstyle\sum\limits_{k\in\mathcal{K}}}
%EndExpansion
P_{k}+P_{r}+2%
%TCIMACRO{\tsum \limits_{k\in\mathcal{K}}}%
%BeginExpansion
{\textstyle\sum\limits_{k\in\mathcal{K}}}
%EndExpansion
\sqrt{\beta_{k}P_{k}P_{r}}}{N_{d}}\right)  .
\end{equation}
However, substituting \underline{$\alpha$}$_{\mathcal{K}}^{\ast}$ $=$
\underline{$0$} in (\ref{GMARC_DF_IrG}), we obtain%
\begin{equation}
I_{r,\mathcal{K}}(\underline{\alpha}_{\mathcal{K}}^{\ast})=0
\end{equation}
which contradicts the assumption in (\ref{DGA5_Case1}), thus making this case infeasible.

\textit{Case 2}: Consider the condition for Case 2 in (\ref{DGA5_Case2}). This
condition implies that the case occurs when the maximum bound achievable at
the relay is smaller than the bound at the destination. From
(\ref{GMARC_DF_IrG}), we observe that $I_{r,\mathcal{K}}$ increases
monotonically with $\alpha_{k}$ for all $k$ and achieves a maximum of
\begin{equation}
I_{r,\mathcal{K}}(\underline{\alpha}_{\mathcal{K}}^{\ast})=C\left(  \frac{%
%TCIMACRO{\tsum \limits_{k\in\mathcal{K}}}%
%BeginExpansion
{\textstyle\sum\limits_{k\in\mathcal{K}}}
%EndExpansion
P_{k}}{N_{r}}\right)
\end{equation}
at $\underline{\alpha}_{\mathcal{K}}^{\ast}=\underline{1}$. Comparing
(\ref{GMARC_DF_IrG}) and (\ref{GMARC_DF_IdG}) at $\underline{\alpha
}_{\mathcal{K}}^{\ast}$ $=$ $\underline{1}$, we obtain the condition for this
case as%
\begin{equation}
\frac{%
%TCIMACRO{\tsum \limits_{k\in\mathcal{K}}}%
%BeginExpansion
{\textstyle\sum\limits_{k\in\mathcal{K}}}
%EndExpansion
P_{k}}{N_{r}}\leq\frac{%
%TCIMACRO{\tsum \limits_{k\in\mathcal{K}}}%
%BeginExpansion
{\textstyle\sum\limits_{k\in\mathcal{K}}}
%EndExpansion
P_{k}+P_{r}}{N_{d}}\text{.} \label{DGA5_C2_CondK}%
\end{equation}

\textit{Case 3}: Finally, consider Case 3 in (\ref{DGA5_Case3}). This case
occurs when the maximum rate bound achievable at the relay and destination are
equal. The max-min solution for this case is obtained by considering two
sub-cases. The first is the relatively straightforward sub-case where
$\underline{\alpha}_{\mathcal{K}}^{\ast}$ $=$ $\underline{1}$ is the max-min
rule. The resulting maximum sum-rate is the same as that for \textit{case} $2$
with the condition in (\ref{DGA5_C2_CondK}) satisfied with equality. Consider
the second sub-case where $\underline{\alpha}_{\mathcal{K}}^{\ast}$ $\not =$
$\underline{1}$, i.e.,
\begin{equation}
\frac{%
%TCIMACRO{\tsum \limits_{k\in\mathcal{K}}}%
%BeginExpansion
{\textstyle\sum\limits_{k\in\mathcal{K}}}
%EndExpansion
P_{k}}{N_{r}}>\frac{%
%TCIMACRO{\tsum \limits_{k\in\mathcal{K}}}%
%BeginExpansion
{\textstyle\sum\limits_{k\in\mathcal{K}}}
%EndExpansion
P_{k}+P_{r}}{N_{d}}\text{.}\label{DGA5_C3_Cond}%
\end{equation}
In Appendix \ref{DG_AppConvex} we show that, for a fixed \underline{$\alpha$%
}$_{\mathcal{K}}$, $I_{d,\mathcal{S}}$, is a concave function of
\underline{$\beta$}$_{\mathcal{K}}$ for all $\mathcal{S}\subseteq\mathcal{K}$.
Furthermore, from (\ref{alp_beta_bounds}), for \underline{$\alpha$%
}$_{\mathcal{K}}$ $\not =$ \underline{$1$}, $I_{d,\mathcal{K}}$ in
(\ref{GMARC_DF_IdG}) is maximized by a $\underline{\beta}_{\mathcal{K}}^{\ast
}$ whose entries $\beta_{k}^{^{\ast}}$, for all $k\in\mathcal{K}$, satisfy%
\begin{equation}%
%TCIMACRO{\tsum \limits_{k\in\mathcal{K}}}%
%BeginExpansion
{\textstyle\sum\limits_{k\in\mathcal{K}}}
%EndExpansion
\beta_{k}^{^{\ast}}=1\label{DGA5_betaKsum}%
\end{equation}
and are given as%
\begin{equation}%
\begin{array}
[c]{cc}%
\beta_{k}^{\ast}=\left\{
\begin{array}
[c]{ll}%
\frac{\left(  1-\alpha_{k}\right)  P_{k}}{\sum_{k=1}^{K}\left(  1-\alpha
_{k}\right)  P_{k}} & \underline{\alpha}_{\mathcal{K}}\not =\underline{1}\\
0 & \underline{\alpha}_{\mathcal{K}}=\underline{1}%
\end{array}
\right.   & \text{for all }k\in\mathcal{K}.
\end{array}
\label{DGA5_betaopt}%
\end{equation}
Observe that the optimal power fraction $\beta_{k}^{^{\ast}}$ that the relay
allocates to cooperating with user $k$ is proportional to the power allocated
by user $k$ to achieve coherent combining gains at the destination. Thus, one
can formulate the optimization problem for this case as%
\begin{equation}%
\begin{array}
[c]{ll}%
\text{maximize} & I_{r,\mathcal{K}}\left(  \underline{\alpha}\right)  \\
\text{subject to} & I_{r,\mathcal{K}}\left(  \underline{\alpha}\right)
=I_{d,\mathcal{K}}\left(  \underline{\alpha},\underline{\beta}\right)  ,\\
&
%TCIMACRO{\tsum \limits_{k\in\mathcal{K}}}%
%BeginExpansion
{\textstyle\sum\limits_{k\in\mathcal{K}}}
%EndExpansion
\beta_{k}=1.
\end{array}
\label{DGA5_Lmax}%
\end{equation}
Using Lagrange multipliers we can show that it suffices to consider $\beta
_{k}=\beta_{k}^{^{\ast}}$ in the maximization. Since the optimal $\beta
_{k}^{\ast}$ in (\ref{DGA5_betaopt}) is a function of \underline{$\alpha$%
}$_{\mathcal{K}}$, $I_{d,\mathcal{K}}(\underline{\alpha}_{\mathcal{K}%
},\underline{\beta}_{\mathcal{K}}^{\ast})$ simplifies to a function of
$\underline{\alpha}_{\mathcal{K}}$ as
\begin{equation}
I_{d,\mathcal{K}}(\underline{\alpha}_{\mathcal{K}},\underline{\beta
}_{\mathcal{K}}^{\ast})=C\left(  \frac{%
%TCIMACRO{\tsum \limits_{k\in\mathcal{K}}}%
%BeginExpansion
{\textstyle\sum\limits_{k\in\mathcal{K}}}
%EndExpansion
P_{k}+P_{r}+2\sqrt{%
%TCIMACRO{\tsum \limits_{k\in\mathcal{K}}}%
%BeginExpansion
{\textstyle\sum\limits_{k\in\mathcal{K}}}
%EndExpansion
\left(  1-\alpha_{k}\right)  P_{k}P_{r}}}{N_{d}}\right)  .
\end{equation}
We further simplify $I_{d,\mathcal{K}}(\underline{\alpha}_{\mathcal{K}%
},\underline{\beta}_{\mathcal{K}}^{\ast})$ and $I_{r,\mathcal{K}}%
(\underline{\alpha}_{\mathcal{K}})$ as follows. We write%
\begin{equation}%
\begin{array}
[c]{cc}%
P_{\max}=\max_{k\in\mathcal{K}}P_{k}, & \lambda_{k}=\left.  P_{k}\right/
P_{\max}\text{,}%
\end{array}
\label{DGA5_lambdadef}%
\end{equation}
and
\begin{equation}
q\overset{\vartriangle}{=}\sqrt{%
%TCIMACRO{\tsum \limits_{k\in\mathcal{K}}}%
%BeginExpansion
{\textstyle\sum\limits_{k\in\mathcal{K}}}
%EndExpansion
\left(  1-\alpha_{k}\right)  \lambda_{k}}.\label{DGA5_qK_defn}%
\end{equation}
Substituting (\ref{DGA5_lambdadef}) and (\ref{DGA5_qK_defn}) in
(\ref{GMARC_DF_IrG}) and (\ref{GMARC_DF_IdG}), we have%
\begin{align}
I_{r,\mathcal{K}}(q) &  =C\left(  \frac{\left(
%TCIMACRO{\tsum \limits_{k\in\mathcal{K}}}%
%BeginExpansion
{\textstyle\sum\limits_{k\in\mathcal{K}}}
%EndExpansion
P_{k}\right)  -q^{2}P_{\max}}{N_{r}}\right)  \\
I_{d,\mathcal{K}}(q) &  =C\left(  \frac{\left(
%TCIMACRO{\tsum \limits_{k\in\mathcal{K}}}%
%BeginExpansion
{\textstyle\sum\limits_{k\in\mathcal{K}}}
%EndExpansion
P_{k}\right)  +P_{r}+2q\sqrt{P_{\max}P_{r}}}{N_{d}}\right)  .
\end{align}
Observe that $I_{r,\mathcal{K}}(q)$ and $I_{d,\mathcal{K}}(q)$ are
monotonically increasing and decreasing functions of $q$ and thus, the
maximization in (\ref{DGA5_Lmax}) simplifies to determining a $q$ such that
\begin{equation}
\frac{%
%TCIMACRO{\tsum \limits_{k\in\mathcal{K}}}%
%BeginExpansion
{\textstyle\sum\limits_{k\in\mathcal{K}}}
%EndExpansion
P_{k}-q^{2}P_{\max}}{N_{r}}=\frac{%
%TCIMACRO{\tsum \limits_{k\in\mathcal{K}}}%
%BeginExpansion
{\textstyle\sum\limits_{k\in\mathcal{K}}}
%EndExpansion
P_{k}+P_{r}+2q\sqrt{P_{\max}P_{r}}}{N_{d}}.\label{DGA5_C3_quadeqn}%
\end{equation}
The condition in (\ref{DGA5_C3_quadeqn}) has the geometric interpretation that
the bounds on $R_{\mathcal{K}}$ are maximized when the $K$-user sum rate plane
achieved at the relay is tangential to the concave sum-rate surface achieved
at the destination at its maximum value. We further simplify
(\ref{DGA5_C3_quadeqn}) by using the definitions in Appendix
\ref{DG_App4OBProof} for $K_{0}$, $K_{1}$, $K_{2}$, and $K_{3}$. From
(\ref{DGA5_C3_Cond}), since $K_{3}>K_{2}$, the quadratic equation in
(\ref{DGA5_C3_quadeqn}) has only one positive solution given by
\begin{equation}
q^{\ast}=\frac{-K_{1}+\sqrt{K_{1}^{2}+\left(  K_{3}-K_{2}\right)  K_{0}}%
}{K_{0}}.\label{DGA5_qKstar}%
\end{equation}
The optimal power policy for this case is then the set $\mathcal{P}$ of
$(\underline{\alpha}_{\mathcal{K}}^{\ast},\underline{\beta}_{\mathcal{K}%
}^{\ast}(\underline{\alpha}_{\mathcal{K}}^{\ast}))$ such that $\underline
{\alpha}_{\mathcal{K}}^{\ast}$ satisfies (\ref{DGA5_qK_defn}) for $q=$
$q^{\ast}$ and for each such choice of $\underline{\alpha}_{\mathcal{K}}%
^{\ast}$, $\underline{\beta}_{\mathcal{K}}^{\ast}$ is given by
(\ref{DGA5_betaopt}). The maximum achievable sum-rate for this case is then
given by
\begin{equation}
C\left(  \frac{\left(
%TCIMACRO{\tsum \limits_{k\in\mathcal{K}}}%
%BeginExpansion
{\textstyle\sum\limits_{k\in\mathcal{K}}}
%EndExpansion
P_{k}\right)  -\left(  q^{\ast}\right)  ^{2}P_{\max}}{N_{r}}\right)
.\label{DGA5_C3_Rate}%
\end{equation}

\begin{remark}
The optimal $q^{\ast}$ in (\ref{DGA5_qKstar}) is the same as the optimal
$x^{\ast}$ in (\ref{DGA4_xstar}). Further, the maximum inner (DF) and outer
bounds on the sum-rate are also the same for the \textit{equal-bounds} case in
(\ref{DGA5_C3_Rate}) and (\ref{DGA4_C3_Rate}), respectively.
\end{remark}

\section{\label{DGA6_ActPrf}Sum-Capacity Proof for the Active Class}

In Theorem \ref{DGDF_Th3}, we proved that DF achieves the sum-capacity for an
active class of degraded Gaussian\ MARCs. In the proof we argue that since the
maximum DF sum-rate is the same as the maximum outer bound sum-rate, every DF
max-min rule $(\underline{\alpha}_{\mathcal{K}}^{\ast},\underline{\beta
}_{\mathcal{K}}^{\ast})\in\mathcal{P}_{a}$ that achieves this maximum
sum-rate, i.e., for which $R_{r}(\underline{\alpha}_{\mathcal{K}}^{\ast})\cap
R_{d}(\underline{\alpha}_{\mathcal{K}}^{\ast},\underline{\beta}_{\mathcal{K}%
}^{\ast})$ belongs to the set of active cases, also achieves the sum-capacity.
We now present a more detailed proof of the argument.

We begin by comparing the inner and outer bounds. As in the symmetric case,
without loss of generality, we write
\begin{equation}%
\begin{array}
[c]{cc}%
\gamma_{k}=\left(  1-\alpha_{k}\right)  \beta_{k} & \text{for all }k
\end{array}
\label{DG_A6gammak}%
\end{equation}
where $(\underline{\alpha}_{\mathcal{K}},\underline{\beta}_{\mathcal{K}}%
)\in\Gamma.$ We then have,
\begin{equation}
B_{r,\mathcal{K}}(\underline{\alpha}_{\mathcal{K}},\underline{\beta
}_{\mathcal{K}})=C\left(
%TCIMACRO{\tsum \limits_{k\in\mathcal{K}}}%
%BeginExpansion
{\textstyle\sum\limits_{k\in\mathcal{K}}}
%EndExpansion
\frac{P_{k}}{N_{r}}-\frac{\left(  \sum\limits_{k\in\mathcal{K}}\sqrt{\left(
1-\alpha_{k}\right)  \beta_{k}P_{k}}\right)  ^{2}}{N_{r}}\right)
\label{DGDF_A6_BrK}%
\end{equation}
and
\begin{equation}
B_{d,\mathcal{K}}(\underline{\alpha}_{\mathcal{K}},\underline{\beta
}_{\mathcal{K}})=C\left(  \frac{%
%TCIMACRO{\tsum \limits_{k\in\mathcal{K}}}%
%BeginExpansion
{\textstyle\sum\limits_{k\in\mathcal{K}}}
%EndExpansion
P_{k}+P_{r}+2%
%TCIMACRO{\tsum \limits_{k\in\mathcal{K}}}%
%BeginExpansion
{\textstyle\sum\limits_{k\in\mathcal{K}}}
%EndExpansion
\sqrt{\left(  1-\alpha_{k}\right)  \beta_{k}P_{k}P_{r}}}{N_{d}}\right)
=I_{d,\mathcal{K}}(\underline{\alpha}_{\mathcal{K}},\underline{\beta
}_{\mathcal{K}}).\label{DGDF_A6_BdK}%
\end{equation}
Choosing $\underline{\beta}_{\mathcal{K}}$ as the DF max-min rule
$\underline{\beta}_{\mathcal{K}}^{\ast}$ in (\ref{DGDF_betastar}), simplifies
(\ref{DGDF_A6_BrK}) to
\begin{equation}
B_{r,\mathcal{K}}(\underline{\alpha}_{\mathcal{K}},\underline{\beta
}_{\mathcal{K}}^{\ast})=C\left(
%TCIMACRO{\tsum \limits_{k\in\mathcal{K}}}%
%BeginExpansion
{\textstyle\sum\limits_{k\in\mathcal{K}}}
%EndExpansion
\frac{\alpha_{k}P_{k}}{N_{r}}\right)  =I_{r,\mathcal{K}}(\underline{\alpha
}_{\mathcal{K}}).\label{DGA6_BrKIrK}%
\end{equation}
Using theorem \ref{Th_GMARC_DF_2}, one can then verify that $B_{r,\mathcal{K}%
}(\underline{\alpha}_{\mathcal{K}}^{\ast},\underline{\beta}_{\mathcal{K}%
}^{\ast})=B_{d,\mathcal{K}}(\underline{\alpha}_{\mathcal{K}}^{\ast}%
,\underline{\beta}_{\mathcal{K}}^{\ast})$ is achieved by all $\underline
{\alpha}_{\mathcal{K}}^{\ast}\in\mathcal{P}$. Consider a $\underline{\alpha
}_{\mathcal{K}}^{\ast}\in\mathcal{P}_{a}$ and a corresponding $\beta
_{\mathcal{K}}^{\ast}$ such that the DF region $\mathcal{R}_{r}(\underline
{\alpha}_{\mathcal{K}}^{\ast})\cap\mathcal{R}_{d}(\underline{\alpha
}_{\mathcal{K}}^{\ast},\underline{\beta}_{\mathcal{K}}^{\ast})$ belongs to the
set of active cases. From Theorem \ref{Th_GMARC_DF_2}, this implies that
\begin{equation}%
\begin{array}
[c]{cc}%
I_{d,\mathcal{A}}(\underline{\alpha}_{\mathcal{K}}^{\ast},\underline{\beta
}_{\mathcal{K}}^{\ast})+I_{r,\mathcal{A}^{c}}(\underline{\alpha}_{\mathcal{K}%
}^{\ast})>I^{\ast}=B^{\ast} & \text{for all }\mathcal{A}\subset\mathcal{K}.
\end{array}
\label{DGDF_A6_IntCase}%
\end{equation}
Using (\ref{DG_A6gammak}), we expand $B_{d,\mathcal{S}}$ in
(\ref{Con_final_B2S}) as a function of $(\underline{\alpha}_{\mathcal{K}%
}^{\ast},\underline{\beta}_{\mathcal{K}}^{\ast})$ as
\begin{align}
B_{d,\mathcal{S}}(\underline{\alpha}_{\mathcal{K}}^{\ast},\underline{\beta
}_{\mathcal{K}}^{\ast})  & =C\left(  \frac{%
%TCIMACRO{\tsum \limits_{k\in\mathcal{K}}}%
%BeginExpansion
{\textstyle\sum\limits_{k\in\mathcal{K}}}
%EndExpansion
P_{k}+\left(  1-\left(
%TCIMACRO{\tsum \limits_{k\in\mathcal{S}^{c}}}%
%BeginExpansion
{\textstyle\sum\limits_{k\in\mathcal{S}^{c}}}
%EndExpansion
\left(  1-\alpha_{k}^{\ast}\right)  \beta_{k}^{\ast}\right)  \right)  P_{r}+2%
%TCIMACRO{\tsum \limits_{k\in\mathcal{S}}}%
%BeginExpansion
{\textstyle\sum\limits_{k\in\mathcal{S}}}
%EndExpansion
\sqrt{\left(  1-\alpha_{k}^{\ast}\right)  \beta_{k}^{\ast}P_{k}P_{r}}}{N_{d}%
}\right)  \\
& \geq I_{d,\mathcal{S}}(\underline{\alpha}_{\mathcal{K}}^{\ast}%
,\underline{\beta}_{\mathcal{K}}^{\ast})\label{DGA6_BdId}%
\end{align}
where (\ref{DGA6_BdId}) follows from the fact that $\left(  1-\alpha_{k}%
^{\ast}\right)  \beta_{k}^{\ast}\leq\beta_{k}^{\ast}$, for all $k$ and for all
$(\underline{\alpha}_{\mathcal{K}}^{\ast},\underline{\beta}_{\mathcal{K}%
}^{\ast})$. It is, however, not easy to compare $B_{r,\mathcal{S}}%
(\underline{\alpha}_{\mathcal{K}}^{\ast},\underline{\beta}_{\mathcal{K}}%
^{\ast})$ with $I_{r,\mathcal{S}}(\underline{\alpha}_{\mathcal{K}}^{\ast})$.
Note, however, that the choice of $\gamma_{k}$ in (\ref{DG_A6gammak}) requires
the same source-relay correlation values for both the inner and outer bounds.
Furthermore, for every choice of Gaussian input distribution with the same $K$
correlation values for both bounds, comparing the degraded cutset and DF
bounds in (\ref{DGMARC_OB_1}) and (\ref{Prop_DF_rateregion}), respectively,
for a constant $U$, we have
\begin{equation}%
\begin{array}
[c]{cc}%
I(X_{\mathcal{S}};Y_{r}|X_{\mathcal{S}^{c}}X_{r})\geq I(X_{\mathcal{S}}%
;Y_{r}|X_{\mathcal{S}^{c}}V_{\mathcal{K}}X_{r}) & \text{for all }%
\mathcal{S}\subseteq\mathcal{K}%
\end{array}
\label{DGA6_OBDFr}%
\end{equation}
where in (\ref{DGA6_OBDFr}) we use the fact that conditioning does not
increase entropy to show that the cutset bounds at the relay are less
restrictive than the corresponding DF bounds. From (\ref{DGA6_BrKIrK}), the
inequality in (\ref{DGA6_OBDFr}) simplifies to an equality for $\mathcal{S}%
=\mathcal{K}$ and for $(\underline{\alpha}_{\mathcal{K}}^{\ast},\underline
{\beta}_{\mathcal{K}}^{\ast})\in\mathcal{P}_{a}$ when $\gamma_{k}$ is given by
(\ref{DG_A6gammak}). Combining these inequalities with (\ref{DGDF_A6_IntCase}%
), we then have
\begin{equation}%
\begin{array}
[c]{cc}%
B_{d,\mathcal{A}}(\underline{\alpha}_{\mathcal{K}}^{\ast},\underline{\beta
}_{\mathcal{K}}^{\ast})+B_{r,\mathcal{A}^{c}}(\underline{\alpha}_{\mathcal{K}%
}^{\ast})>I^{\ast}=B^{\ast} & \text{for all }\mathcal{A}\subset\mathcal{K}.
\end{array}
\end{equation}
Thus, every DF max-min rule that results in an active case polymatroid
intersection, i.e., every $(\underline{\alpha}_{\mathcal{K}}^{\ast}%
,\underline{\beta}_{\mathcal{K}}^{\ast})\in\mathcal{P}_{a}$, also results in
an active case for the outer bounds when $\gamma_{k}$ is given by
(\ref{DG_A6gammak}).

\bibliographystyle{IEEEtran}
\bibliography{MARC_refs}

\end{document}